%% file: KW_Community.tex
\documentclass[letterpaper]{vldb}

\usepackage{times}
\usepackage{balance} % for  \balance command ON LAST PAGE  (only there!)
\usepackage{amsmath,epsfig,subfigure}
\usepackage{algorithm}
\usepackage[noend]{algorithmic}
\usepackage{graphicx}
\usepackage{amssymb}
\usepackage{xspace}
\usepackage{multirow}
\usepackage{url}
\usepackage{epsfig}
\usepackage{amsfonts}
\usepackage[usenames,dvipsnames]{xcolor}

\usepackage[hide]{chato-notes}
\usepackage[normalem]{ulem}
\usepackage{ upgreek }

\newcommand{\PP}{./}
\newcommand{\FP}{./}
\newcommand{\EP}{./}

\newcommand{\taubar}{\bar\tau}
\newcommand{\ktruss}{$k$-truss\xspace}

\newcommand{\ignore}[1]{}
\newcommand{\nop}[1]{}
\newcommand{\eat}[1]{}

\newcommand{\kw}[1]{{\ensuremath {\mathsf{#1}}}\xspace}

\newcommand{\stitle}[1]{\vspace{1ex} \noindent{\bf #1}}
\long\def\comment#1{}

\newcommand{\red}[1]{{\color{red}{#1}}}
\newcommand{\LL}[1]{{#1}} 
\newcommand{\TobeDelete}[1]{\eat{}}

\newcommand{\xin}[1]{{#1}}

\newcommand{\add}[1]{\eat{}} 
\newcommand{\cut}[1]{{#1}}

\newcommand{\RV}[1]{{#1}} 
\newcommand{\new}[1]{{#1}} 

\newtheorem{definition}{Definition}

\newtheorem{lemma}{Lemma}
\newtheorem{theorem}{Theorem}
\newtheorem{example}{Example}

\newtheorem{corollary}{Corollary}
\newtheorem{remark}{Remark}
\newtheorem{problem}{Problem}

\newtheorem{observation}{Observation}

\newcommand{\dist}{\kw{dist}}
\newcommand{\diam}{\kw{diam}}
\newcommand{\con}{\kw{connect}}

\newcommand{\LCTC}{\kw{LCTC}}

\newcommand{\MDC}{\kw{MDC}}

\newcommand{\ACC}{\kw{ACC}}

\newcommand{\doc}{\kw{attr }}
\newcommand{\kcs}{\kw{ACS }}
\newcommand{\acs}{\kw{ACS }}
\newcommand{\cs}{\kw{CS }}
\newcommand{\ks}{\kw{KS }}
\newcommand{\tf}{\kw{TF }}
\newcommand{\ktc}{\kw{ATC }}
\newcommand{\atc}{\kw{ATC }}
\newcommand{\ktcs}{\kw{ATCs }}

\newcommand{\keyw}{\kw{f }}

\newcommand{\den}{\kw{coh }}
\newcommand{\com}{\kw{com }}
\newcommand{\VEDalK}{\kw{WDalK }-\kw{Problem}}
\newcommand{\WDalK}{\kw{WDalK }-\kw{Problem}}
\newcommand{\WDK}{\kw{WDK }-\kw{Problem}}
\newcommand{\score}{\kw{score }}

\newcommand{\ktcp}{\kw{ATC}-\kw{Problem}}
\newcommand{\ktcrp}{\kw{ATCr}-\kw{Problem}}

\newcommand{\atcp}{\kw{ATC}-\kw{Problem}}

\newcommand{\prin}{\kw{Principle}}
\newcommand{\basic}{\kw{Basic}}

\newcommand{\invk}{\kw{invA}}
\newcommand{\invA}{\kw{invA}}
\newcommand{\ati}{\kw{AT index}}

\newcommand{\gain}{\kw{gain}}

\newcommand{\LATC}{\kw{LocATC}}
\newcommand{\bulk}{\kw{BULK}}

\newcommand{\kdtruss}{\kw{(k,d)}-\kw{truss}}
\newcommand{\kdtrusses}{\kw{(k,d)}-\kw{trusses}}

\newcommand{\w}{\kw{w}}

\begin{document}

% \begin{center}
% \textbf{Querying Closest Community with Quality Gurantees in Large Networks}\\[1ex]
% \end{center}
% \medskip
% \medskip
\title{Attribute-Driven Community Search}
%\title{Attribute Truss Community Search}

%\eat{
\author{
{Xin Huang, Laks V.S. Lakshmanan}%
\vspace{1.6mm}\\
\fontsize{10}{10}\selectfont\itshape
University of British Columbia\\
%\fontsize{10}{10}\selectfont\itshape
\fontsize{9}{9}\selectfont\ttfamily\upshape
\{xin0,laks\}@cs.ubc.ca
}
%}

\maketitle

\begin{abstract}
\add{Recently, community search over graphs has gained significant interest. In applications such as analysis of protein-protein interaction (PPI) networks, citation graphs, and collaboration networks, nodes tend to have attributes. Unfortunately, \new{most} previous community search algorithms ignore attributes and result in communities with poor cohesion w.r.t. their node attributes. In this paper, we study the problem of attribute-driven community search, that is, given an undirected graph $G$ where nodes are associated with attributes, and an input query $Q$ consisting of nodes $V_q$ and attributes $W_q$, find the communities containing $V_q$, in which most community members are densely inter-connected and have similar attributes. }

\add{We formulate this problem as finding attributed truss communities (ATC), i.e., finding connected and close k-truss subgraphs containing $V_q$, with the largest attribute relevance score. We design a framework of desirable properties that good score function should satisfy. We show that the problem is NP-hard. However, we develop an efficient greedy algorithmic framework to iteratively remove nodes with the least popular attributes, and shrink the graph into an \atc. In addition, we also build an elegant index to maintain $k$-truss structure and attribute information, and propose efficient query processing algorithms. Extensive experiments on large real-world networks with ground-truth communities show that our algorithms significantly outperform the state of the art and demonstrates their efficiency and effectiveness.  
}
\cut{Recently, community search over graphs has attracted  significant 
attention and many algorithms have been developed for finding dense
subgraphs from large graphs that contain given query nodes. In applications such as analysis of protein protein interaction (PPI) networks, citation graphs, and collaboration networks, nodes tend to have attributes. Unfortunately, \new{most} previously developed community search algorithms ignore these attributes and result in communities with poor cohesion w.r.t. their node attributes. In this paper, we study the problem of attribute-driven community search, that is, given an undirected graph $G$ where nodes are associated with attributes, and an input query $Q$ consisting of nodes $V_q$ and attributes $W_q$, find the communities containing $V_q$, in which most community members are densely inter-connected and have similar attributes. 

We formulate our problem of finding attributed truss communities (ATC), as finding all connected and close k-truss subgraphs containing $V_q$, that are locally maximal and have the largest attribute relevance score among such subgraphs. We design a novel attribute relevance score function and establish its desirable properties. The problem is shown to be NP-hard. % Since users are commonly interested in the top-$r$ answers, we also study the problem of finding $r$ communities with the highest scores, and show that finding the top-1 community is NP-hard.
 However, we develop an efficient greedy algorithmic framework, which finds a maximal $k$-truss containing $V_q$, and then iteratively removes the nodes with the least popular attributes and shrinks the graph so as to satisfy community constraints. We also build an elegant index to maintain the known $k$-truss structure and attribute information, and propose efficient query processing algorithms. Extensive experiments on large real-world networks with ground-truth communities shows the efficiency and effectiveness of our proposed methods.} %We conduct an extensive experimental study using several large real-world networks, and the results demonstrate the efficiency and effectiveness of our proposed methods, which outperform several natural baselines.}
\end{abstract}

\section{Introduction}
\label{sec:intro} 
\input{intro}

\section{Preliminaries and Desiderata} 
\vspace{-0.3cm}
\label{sec:prelims}

\input{framework}

\section{Related Work}
\label{sec:related} 
\input{relate}

\section{Attributed Community Model}
\label{sec:prob} 
\input{problem}

\section{Problem Analysis}
\label{sec:analysis} 
\input{character}

\section{Top-down Greedy Algorithm}
\label{sec:algo} 
\input{basic}

\section{Index-based Search Algorithm}
\label{sec:index} 
\input{index}

\section{Experiments}\label{sec.exp}
\label{sec:exp} 
\input{exp}
\section{Conclusion}
\label{sec:conc} 
In this work, we propose an attributed truss community (\atc) model that allows to find a community containing query nodes with cohesive and tight structure, also sharing homogeneous query attributes. The problem of finding an \atc is \emph{NP}-hard. We also show that the attribute score function is not monotone, submodular, or supermodular, indicating approximation algorithms may not be easy to find. We propose several carefully designed strategies to quickly find high-quality communities. We design an elegant and compact index, \ati, and implement an efficient query processing algorithm, which exploits  local exploration and bulk deletion. Extensive experiments reveal that ground-truth communities and social circles can be accurately found by our \atc model, and that our model and algorithms significantly outperform previous approaches. Several interesting questions remain. Some examples include attributed community search over heterogeneous graphs and edge-weighted graphs, and w.r.t. weighted query attributes.

\bibliographystyle{abbrv}
{\scriptsize
\bibliography{truss}
%\bibliography{minTruss}
%\bibliography{minTruss-Revised}
}

\end{document}

%% file: intro.tex
\eat{ 
\note[Laks]{See comments about motivation. They need to be addressed. Here are some notes for use during my pass about the anticipated flow of the intro. A few general remarks about CS and KS, including review of key literature in these areas. A few lines and an in-line motivational example illustrating the limitation of CS and KS, also saying what the desired ``answer'' is and why CS and KS won't catch it. This example needs to be concrete and not abstract. This is critical. Then we can follow it up with a more detailed discussion as in the current intro. Finally, challenges in semantics, score function, and ranking. Then list of contributions.} 

\note[Xin]{I have added new examples to show the weakness of CS and KS in Figure 1, 2 and 3.}
}

\cut{
Graphs have emerged as a powerful model for representing different types of data. For instance, unstructured data (e.g., text documents), semi-structured data (e.g., XML databases) and structured data (e.g., relational databases) can all be modeled as graphs, where the vertices(nodes) are respectively documents, elements, and tuples, and the edges can respectively be hyperlinks, parent-child relationships, and primary-foreign-key relationships \cite{li2008ease}.  %In addition, graph databases become more and more popular recently, such as Neo4j, OrientDB, Titan and so on. 
In these graphs, communities naturally exist as groups of nodes that are densely interconnected. %and have similar attributes. 
Finding communities in large networks has found extensive applications in protein-protein interaction networks, sensor/communication networks, and collaboration networks. Consequently, community detection, i.e., finding all communities in a given network, serves as a global network-wide analysis tool, and has been extensively studied in the literature. Specifically, various definitions of communities based on different notions of dense subgraphs have been proposed and studied: quasi-clique \cite{CuiXWLW13}, densest subgraph \cite{wu2015robust}, $k$-core \cite{sozio2010,li2015influential,cui2014local,barbieri2015efficient}, and k-truss \cite{huang2014,huang2015approximate}. More recently, a related but different problem called community search has generated considerable interest. It is motivated by the need to make answers more meaningful and personalized to the user \cite{mcauley2012learning, huang2014}. For a given set of query nodes, community search seeks to find the communities containing the query nodes.}
\add{Graphs have emerged as a powerful model for representing different types of data, such as protein-protein interaction networks, sensor/communication networks, and collaboration networks. In these graphs, communities naturally exist as groups of nodes (vertices) that are densely interconnected. Consequently, community detection, i.e., finding all communities in a given network, serves as a global network-wide analysis tool, and has been extensively studied in the literature. \cut{Specifically, various definitions of communities based on different notions of dense subgraphs have been proposed and studied: quasi-clique \cite{CuiXWLW13}, densest subgraph \cite{wu2015robust}, $k$-core \cite{sozio2010,li2015influential,cui2014local,barbieri2015efficient}, and k-truss \cite{huang2014,huang2015approximate}.} More recently, a related but different problem called community search has generated considerable interest. It is motivated by the need to make answers more meaningful and personalized to the user \cite{mcauley2012learning, huang2014}. For a given set of query nodes, community search seeks to find the communities containing the query nodes. }

\begin{figure}[t]
%\small
\scriptsize
\vskip -0.2in
\centering
\includegraphics[width=0.8\linewidth]{\FP 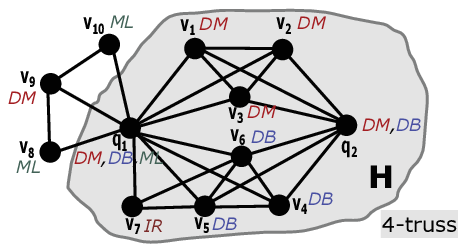}
\vskip -0.1in
\caption{An example attributed graph $G$}\vskip -0.2in
\label{fig.community}
\end{figure}

\eat{
\note[Laks]{Subgraph $H$ should be highlighted.} 
\note[Xin]{Done}
}

%%%%%%%%%%%%%%%%%%%%%%%%%%%%%%%% 
\begin{figure*}[t]
\scriptsize
\vspace{-0.6cm}
\centering
{
\subfigure[\textbf{$H_1$}. \small{4-truss community on $V_q=\{q_1, q_2\}$, $W_q=\{DB\}$}]{\includegraphics[width=0.22\linewidth]{\FP 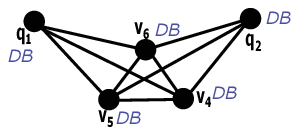} }\hskip 0.10in
\subfigure[$H_2$. \small{4-truss community on $V_q=$ $\{q_1,$$ q_2\}$, $W_q=\{DB, DM\}$}]{\includegraphics[width=0.27\linewidth]{\FP 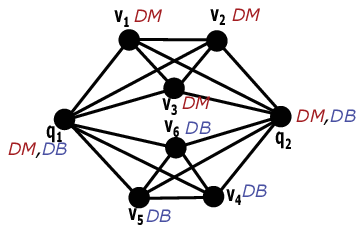} } \hskip 0.10in
\subfigure[$H_3$. \small{4-truss community on $V_q=\{q_1, q_2\}$, $W_q=\{DM\}$}]{\includegraphics[width=0.22\linewidth]{\FP 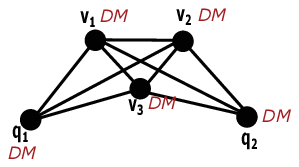} }\hskip 0.10in
\subfigure[$H_4$. \small{3-truss community on $V_q=\{q_1\}$, $W_q=\{ML\}$}]{\includegraphics[width=0.16\linewidth]{\FP 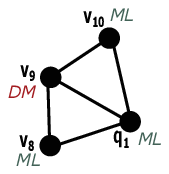} }\hskip 0.15in
}
\vspace{-0.4cm}
\caption{Attribute Communities for queries on different query nodes $V_q$ and query attributes $W_q$.}
%\caption{Communities of query nodes $V_q=\{q_1\}$ respectively on different attribute values of $DB$, $DM$ and $ML$}
\label{fig.subcom}\vspace{-0.4cm}
\end{figure*}
%%%%%%%%%%%%%%%%%%%%%%%%%%%%%%%% 

\eat{ 
\note[Laks]{Figure and caption alignment need adjusting. For (a)-(d), query keywords are shown but not query nodes. Let's include $V_q$ as well in each case.}
\note[Xin]{Done} 
} 

In the aforementioned applications, the entities modeled by the network nodes  often have properties which are important for making sense of communities. E.g., authors in collaboration networks have areas of expertise; %sensors have the parameters they measure (e.g., temperature, pressure, etc.);  
proteins have molecular functions, biological processes, and cellular components as properties. Such networks can be modeled using \emph{attributed graphs} \cite{zhou2009graph} where attributes associated with nodes capture their properties. E.g., Figure \ref{fig.community} shows an example of a collaboration network. The nodes $q_i, v_j, ...$ represent authors. Node attributes (e.g., DB, ML) represent authors' topics of expertise. In finding communities (with or without query nodes) over attributed graphs, we might want to ensure that the nodes in the discovered communities have homogeneous attributes. For instance, it has been found that communities with homogeneous attributes among nodes more accurately predict protein complexes \cite{hu2013utilizing}. Furthermore, we might wish to query, not just using query nodes, but also using query attributes. To illustrate, consider searching for communities containing the nodes $\{q_1, q_2\}$. Based on structure alone, the subgraph $H$ shown in Figure \ref{fig.community} is a good candidate answer for this search, as it is densely connected. However, attributes of the authors in this community are not homogeneous: the community is a mix of authors working in different topics -- DB, DM, IR, and ML. 
% Previous  community search methods, such as those based on $k$-core \cite{sozio2010, li2015influential, cui2014local}, $k$-truss  \cite{huang2015approximate}, %\cite{huang2014, huang2015approximate},
% and 1.0-quasi-$k$-clique-$\ell$-adjacent community \cite{CuiXWLW13}, for $k=4$ and $\ell=3$,  will all report $H$ as the top answer and are thus unsatisfactory. 
\new{Previous  community search methods include those based on $k$-core \cite{sozio2010, li2015influential, cui2014local}, $k$-truss  \cite{huang2015approximate}, %\cite{huang2014, huang2015approximate},
 and 1.0-quasi-$k$-clique-$\ell$-adjacent community \cite{CuiXWLW13}.
A $k$-core \cite{li2015influential} is a subgraph in which each vertex has at least $k$ neighbors within the subgraph. A $k$-truss  \cite{huang2015approximate} is a subgraph in which each edge is contained in at least $(k-2)$ triangles within
the subgraph. The 1.0-quasi-$k$-clique-$\ell$-adjacent community model \cite{CuiXWLW13} allows two $k$-cliques overlapping in $\ell$ vertices to be merged into one community. In Figure~\ref{fig.community}, for $k=4$ and $\ell=3$, all these community models will report $H$ as the top answer and are thus unsatisfactory. }
\add{Thus, in general, communities found by \new{most} previous community search methods can be hard to interpret owing to the heterogeneity of node attributes. Furthermore, the communities reported could contain smaller dense subgraphs with more homogeneity in attributes, which are missed by \new{most} previous methods. 
\new{A recent work \cite{FangCLH16} proposed an attribute community model. A detailed comparison of \cite{FangCLH16} with our model appears in Section \ref{sec:related}.} }
\cut{The subgraph $H_2$ obtained from $H$ by removing node $v_7$ with unique attribute IR, is a more homogeneous community than $H$ and is just as densely connected (see Figure~\ref{fig.subcom}(b)). Intuitively, it is a better answer than $H$. Thus, in general, communities found by \new{most} previous community search methods can be hard to interpret owing to the heterogeneity of node attributes. Furthermore, the communities reported could contain smaller dense subgraphs with more homogeneity in attributes, which are missed by \new{most} previous methods. \new{A recent work \cite{FangCLH16} proposed an attribute community model. A detailed comparison of \cite{FangCLH16} with our model can be found in Section \ref{sec:related}.}}  Consider now querying the graph of Figure~\ref{fig.community} with query nodes $\{q_1, q_2\}$ {\sl and} attributes (i.e., keywords) \{DB, DM\}. We would expect this search to return subgraph $H_2$ (Figure~\ref{fig.subcom}(b)). On the other hand, for the same query nodes, if we search with attribute \{DB\} (resp., \{DM\}), we expect the subgraph $H_1$ (resp., $H_3$) to be returned as the answer (Figure~\ref{fig.subcom}(a)\&(c)). Both $H_1$ and $H_3$ are dense subgraphs where all authors share a common topic (DB or DM). 

Given a query consisting of nodes and attributes (keywords), one may wonder whether we can filter out nodes not having those attributes and then run a conventional community search method on the filtered graph. To see how well this may work, consider querying the graph in Figure~\ref{fig.community} with query node $q_1$ and query attribute ML. Filtering out nodes without attribute ML and applying community search yields the chain consisting of $v_{10}, q_1, v_8$, which is not densely connected. On the other hand, the subgraph induced by $\{q_1, v_8, v_9, v_{10}\}$ is a 3-truss in Figure~\ref{fig.subcom}(d). Even though it includes one node without ML it is more densely connected than the chain above and is a better answer than the chain as it brings out denser collaboration structure among the authors in the community. Thus, a simple filtering based approach will not work. \cut{As some denser subgraphs may be less homogeneous in their node attributes than some sparser ones and a careful balance has to be struck between density and attribute homogeneity. }

%%%%%%%%%%%%%%%%%%%%%%%%% 
\begin{figure}[t]
\scriptsize
 \vspace{-0.2cm}
\centering
{
{\includegraphics[width=0.35\linewidth]{\FP 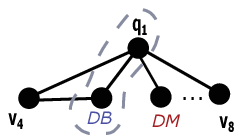} }%\hskip -0.15in
\hskip 0.15in
}
\vspace{-0.6cm}
\caption{Keyword Search with query $W_q=\{q_1, DB\}$}
\label{fig.ks}\vspace{-0.5cm}
\end{figure}

%%%%%%%%%%%%%%%%%%%%%%%%% 

Another topic related to our problem is keyword search over graphs, which has been extensively studied~\cite{dbxplorer02,hristidis2002discover,hristidis2003efficient,bhalotia2002keyword,kacholia2005bidirectional,ding2007finding}. A natural question is whether we can model the information suitably and leverage keyword search to find the right communities. We could model authors' attributes also as nodes and directly connect them to the author nodes and query the resulting graph with the union of the author id's and the keywords. Figure~\ref{fig.ks} illustrates this for a small subgraph of Figure~\ref{fig.community} and a query. Keyword search finds answers corresponding to trees or subgraphs with minimum communication cost that connect the input keywords/nodes, where the communication cost is based on diameter, query distance, weight of spanning tree or steiner tree. On this graph, if we search for the query node $q_1$ and attribute DB, we will get the single edge connecting $q_1$ and DB as the answer as this is the subgraph with minimum communication cost connecting these two nodes. Clearly, this is unsatisfactory as a community. 

In sum, attributed graphs present novel opportunities for community search by combining dense structure of subgraphs with the level of homogeneity of node attributes in the subgraph. \new{Most} previous work in community search fails to produce satisfactory answers over attributed graphs, while keyword search based techniques do not find dense subgraphs. 
\cut{The main problem we study in this paper is finding top-$r$ communities from attributed graphs, given a community search query consisting of query nodes and query attributes. This raises the following major challenges. }
\add{This raises the following major challenges for community search in attributed graphs.} 
Firstly, how should we combine dense connectedness with the distribution of attributes over the community nodes? We need a community definition that promotes dense structure as well as attribute homogeneity. However, there can be tension between these goals: as illustrated in the example above, some denser subgraphs may be less homogeneous in their node attributes than some sparser ones. Secondly, the definition should capture the intuition that the more input attributes that are covered by a community, the better the community. 
\add{Finally, the query processing algorithm should be efficient for large graphs.} 
\cut{Finally, we need to find the answer communities from large input graphs in an efficient manner. }

\eat{ 
\note[Xin]{I make the modifications and fill the contributions as follow.}
} 

\note[Laks]{We need to justify why for query nodes we require that all of them be contained whereas for query attributes we do not strictly require this. Why not require containment for both or why not require containment for query attributes and relax it for query nodes? What will go wrong?} 
\note[Xin]{The query nodes may not form a community under every attribute. That's why we relax the condition. Since our problem is user-query style, it is more easy(familiar) for user to give the query nodes than give attributes. That's why we consider all query nodes are included into a community. It is a hard question. I will continue updating it if I have a better answer.}

To tackle these challenges, we propose an attributed truss community (\atc) model. Given a query $Q =(V_q, W_q)$ consisting of a set of query nodes $V_q$ and a set of query attributes $W_q$, a good community $H$ must be a dense subgraph which contains all query nodes and attributes $W_q$ must be contained in numerous nodes of the community. The more nodes with attribute $w\in W_q$, the more importance to $w$ commonly accorded by the community members.  Additionally, the nodes must share as many attributes as possible. Notice that these two conditions are not necessarily equivalent. 
\note[Laks]{See if the two lines above are what you intended.} 
\eat{
Capturing these intuitions, we define an attribute score function that strikes a balance between attribute homogeneity and coverage and define an \atc as a $k$-truss containing all query nodes and query attributes, such that $k$ is as large as possible and the attribute score is maximized, subject to a constraint on the distance of community nodes from query nodes. We make the following contributions. 
}
%Laks's intro.
Capturing these intuitions, we define an attribute score function that strikes a balance between attribute homogeneity and coverage. Moreover, as a qualifying cohesive and tight structure, we define a novel concept of \kdtruss for modeling a densely connected community. 
\eat{ 
\note[Laks]{Is \kdtruss novel? Wasn't it already defined in our previous paper? Just checking.}  
\note[Xin]{We haven't define it before.}
} 
A \kdtruss is a connected $k$-truss containing all query nodes, where each node has a distance no more than $d$ from every query node. This inherits many nice structural properties, such as bounded diameter, $k$-edge connectivity, and hierarchical structure. Thus, based on attribute score function and \kdtruss, we propose a novel community model as \textbf{attributed truss community} (\atc), which is a \kdtruss with the maximum attribute score. In this paper, we make the following contributions. 

\note[Laks]{Given that there has been previous work on attributed CD, shouldn't we compare our desiderata and our def. with those used in previous attributed CD papers?} 
\note[Xin]{I am not sure where is best to talk the difference. Here is my answer. Attribute community detection targets all communities in the whole network and usually applies a global criterion to find communities, like graph clustering methods(Zhou et al.). On the other hand,  attribute community detection finds a community with homogeneous attributes, but they treat each attribute equally. In contrast, attribute community search provides personalized community detection not only for given query nodes but also for perfering attributes.}

\begin{itemize} 
\vspace{-0.2cm}
\item We motivate the problem of attributed community search, and identify the desiderata of a good attributed community \eat{using four criteria -- participation, cohesiveness, communication cost, and attribute coverage and correlation} (Section \ref{sec:prelims}). 
%, that is, finding a densely conncted containing given query nodes with regard to query attributes. Then, we identify the desiderata of a good attributed community using four
\eat{
\item We make a comprehensive comparison with related work on several topics, such as community search, keyword search, team formation and attributed community detection. (Section \ref{sec:related}) 
\note[Laks]{Remove the previous bullet on related work as it's not a technical contribution.} 
\note[Xin]{Agree.}
} 
\vspace{-0.15cm}
\item \cut{We propose a novel dense and tight subgraph, \kdtruss, and design an attribute score function satisfying the desiderata set out above. Based on this, we propose a community model called attributed truss community (\atc), and formulate the problem of attributed community search as finding \atc (Section \ref{sec:prob}). }
\add{We propose a novel dense and tight subgraph, \kdtruss, and design an attribute score function satisfying the desiderata set out above. Based on this, we propose \atc community model, and formulate the problem (Section \ref{sec:prob}).}
\vspace{-0.15cm}
\item We analyze the structural properties of \atc and show that it is non-monotone, non-submodular and non-supermodular, which signal huge computational challenges. We also formally prove that the problem is NP-hard (Section \ref{sec:analysis}). 
\vspace{-0.15cm}
\item We develop a greedy algorithmic framework to find an \atc containing given query nodes w.r.t. given query attributes. It first \eat{globally} finds a maximal \kdtruss, and then iteratively removes nodes with smallest attribute score contribution. For improving the efficiency and quality, we design a revised attribute marginal gain function and a bulk removal strategy for cutting down the number of iterations (Section \ref{sec:algo}). 
\vspace{-0.15cm}
\item \add{For further improving efficiency, we construct a novel index called \ati. Based on \ati, we develop an algorithm for efficiently exploring the local neighborhood of query nodes to search for an \atc.} \cut{For further improving efficiency, we explore the local neighborhood of query nodes to search an \atc. This algorithm first generates a Steiner tree connecting all query nodes, and then expands the tree to a dense subgraph with the insertion of carefully selected nodes, that have highly correlated attributes and densely connected structure  (Section \ref{sec:index}). }
\vspace{-0.15cm}
\item We conduct extensive experiments on 7 real datasets, and show  that our attribute community model can efficiently and effectively find  ground-truth communities and social circles over real-world networks, significantly outperforming previous work (Section \ref{sec:exp}). 
\note[Laks]{If the real datasets include biological and sensor networks too, what ground-truth communities are we talking about here?} 
\note[Xin]{In biological, the community can be a group of protein with similar molecular functions or cellular components. Try to incldue the protein datasets.}
\vspace{-0.15cm}
\end{itemize}

We discuss  related work in Section \ref{sec:related}, and conclude the paper with a summary in Section \ref{sec:conc}.

\eat{ 
\note[Laks]{Xin, can you fill in these contributions, preferably along with a reference to the section in which each contribution is made?} 
} 

\note[Xin]{Do we need to give a defintion of k-truss, that is, every edge has triangle no less than k-2 in this k-truss.}

%% file: framework.tex
\eat{ 
\note[Laks]{Table 1: It may turn out that in our final formalization, we may model some of these as an objective function and constrain others. So referring to everything as a constraint may be confusing.} 
\note[Xin]{How about the ``Condition'' in Table 1.}
}

\subsection{Preliminaries} 
We consider an undirected, unweighted simple graph $G=(V,$ $ E)$ with $n=|V(G)|$ vertices and $m=|E(G)|$ edges. We denote the set of neighbors of a vertex $v$ by $N(v)$, and the degree of $v$ by $d(v)=|N(v)|$. We let  \LL{$d_{max} = \max_{v\in V}d(v)$} denote the maximum vertex degree in $G$. W.l.o.g. we assume that the graphs we consider are connected. Note that this implies that $m\geq n-1$.  We consider attributed graphs and denote the set of all attributes in a graph by $\mathcal{A}$.
%, with $\gamma=|\mathcal{A}|$. 
Each node $v\in V$ contains a set of zero or more attributes, denoted by $\doc(v)\subseteq \mathcal{A}$. The multiset union of attributes of all nodes in $G$ is denoted $\doc(V)$. Note that $|\doc(V)|=\sum_{v\in V}|\doc(v)|$. We use $V_{w} \subseteq V$ to denote the set of nodes having attribute $w$, i.e., $V_{w} = \{v\in V \mid w \in \doc(v)\}$. 
\note[Laks]{Need to define notions like $k$-core, $k$-truss.}

\subsection{Desiderata of a good community}
Given a query $Q=(V_q, W_q)$ with a set of query nodes $V_q \subseteq V$ and a set of query attributes $W_q$, the attributed community search (\acs) problem is to find a subgraph $H\subseteq G$ containing all query nodes $V_q$, where the vertices are densely inter-connected, cover as many query attributes $W_q$ as possible and share numerous attributes. In addition, the communication cost of $H$ should be low.  We call the query $Q=(V_q, W_q)$  an \kcs query. 
%Before formalizing this problem, we first identify the desiderata of a good attributed community. 
\new{Before formalizing the problem, we first identify the commonly accepted desiderata of a good attributed community. }

\eat{ 
\note[Laks]{Turn the following ``def.'' into desiderata. Formalize cohesiveness, community, attribute score etc. Then define the problem.} 
\note[Xin]{You means the def. to be our models? It need and can make use of several definitions in Problem section. Here, I would like to show the general definition of KCS problem, and  following works by other person maybe be  like to refer this one.}

\begin{definition}[Attribute-Driven Community Search] 
Given a graph $G(V,E)$ and a \kcs query $Q=(V_q, W_q)$, the attributed community search is to find all communities such that each community as a connected subgraph $H=(V(H), E(H))\subseteq G$ satisfies
\begin{enumerate}
\item  (Participant Condition) $H$ contains all query nodes $Q$, i.e., $Q\subseteq V(H)$;
\item  (Attribute Condition) A attribute function $\keyw(H, W_q)$ that measures the attributes of vertices in $H$ correlated with $W_q$,  is maximized or exceeding a given threshold. 
\item  (Cohesiveness Condition) A cohesiveness function $\den(H)$ that measures the cohesive structure of $H$, is maximized or exceeding a given threshold. 
\item  (Communication Condition) A communication cost function $\com(H)$ that measures the  distance of vertices in $H$, should be minimized or no greater than a given threshold. 
\end{enumerate}
\label{def.kcs}
\end{definition}
} 

\noindent 
{\bf Criteria of a good attributed community}: Given a graph $G(V,E)$ and a \kcs query $Q=(V_q, W_q)$, an attributed community is a connected subgraph $H=(V(H), E(H))\subseteq G$ that satisfies: 
\begin{enumerate} 
\vspace{-0.1cm}
\item (Participation) $H$ contains all query nodes as $V_q\subseteq V(H)$;
\vspace{-0.11cm}
\item (Cohesiveness) A cohesiveness function $\den(H)$ that measures the cohesive structure of $H$ is high. 
\vspace{-0.1cm}
\item  (Attribute Coverage and Correlation) An attribute score function $\keyw(H, W_q)$ that measures the coverage and correlation of query attributes in vertices of $H$ is high. 
\vspace{-0.1cm}
\item  (Communication Cost) A communication cost function $\com(H)$ that measures the distance of vertices in $H$ is low. 
\vspace{-0.1cm}
\end{enumerate}

The participation condition is straightforward. The cohesiveness condition is also straightforward since communities are supposed to be densely connected subgraphs. One can use any notion of dense subgraph previously studied, such as $k$-core, $k$-truss, etc. The third condition captures the intuition that more query attributes covered by $H$, the higher $\keyw(H, W_q)$; also more attributes shared by vertices of $H$, the higher $\keyw(H, W_q)$. This motivates designing functions $\keyw(.,.)$ with this property. Finally, keeping the communication cost low helps avoid irrelevant vertices in a community. This is related to the so-called free rider effect, studied in \cite{huang2015approximate,wu2015robust}. Intuitively, the closer the community nodes to query nodes, subject to all other conditions, the more relevant they are likely to be to the query. Notice that sometimes a node that does not contain query attributes may still act as a ``bridge'' between other nodes and help improve the density. A general remark is that other than the first condition, for conditions 2--4, we may either optimize a suitable metric or constrain that the metric be above a threshold (below a threshold for Condition 4). We formalize this intuition in Section~\ref{sec:prob} and give a precise definition of an attributed community\cut{and formally state the main problem studied in the paper} . 

\eat{
Second, the attribute constraint can help detecting the community $H$ that is closely realted with the input attributes. Futhermore, it can distinguish different attribute relavance communities for ranking. One basic function $\keyw(H, W_q)$ can require that for each attribute $w\in W_q$, there exists one node $v\in V(H)$ with $w\in \doc(v)$. Third, the cohesiveness constraint can make a gurantee of the discovered community is densely connected. The cohesiveness function $\den(H)$ can be formulated the minimum degree of vetices in $H$. Last but not least, the communication contraints, which can control the size of community  for avoiding the Free Rider Effect. The communication cost of $H$ should be either minimized or no greater than a given threshold, such the graph size $|H|$, the diameter, the query distance, the cost of Stenier tree and so on. 
Practically, we can unifie two or more constraints into one. 
} 

\note[Laks]{Xin, I have changed the def. to a list of desirable conditions or properties. It's not yet a def. I think the formal def. of \kcs problem will follow in the next section or so.

Also, I'm commenting out the detailed comparison you had here with KS, TF, and CS. Some of that material can be reused in related work. Perhaps related work section can be moved up here instead of being postponed to the end of the paper?}

%% file: relate.tex
Work related to this paper can be classified into community search, keyword search, team formation, and community detection in attributed graphs. Table \ref{tab.cmp} shows a detailed comparison of representative works on these topics. %community search (\cs), keyword search (\ks), team formation (\tf) and attribute community search(\acs).

\begin{table}[t!]
\vskip -0.4cm
\centering
\scriptsize
\begin{tabular}{|l|c|c|c|c|c|c|}
\hline
\multirow{2}{*}{Method}  & \multirow{2}{*}{Topic}  & Participation  & Attribute  & Cohesiveness  & Communication  \\
       &  & Condition  & Function  & Constraint  & Cost  \\
\hline \hline
\cite{bhalotia2002keyword} & \ks  &  $\upchi$  & \checkmark  & $\upchi$   & \checkmark\\ \hline
\cite{ding2007finding}  & \ks &  $\upchi$  & \checkmark  & $\upchi$   & \checkmark\\ \hline
\cite{li2008ease} & \ks &  $\upchi$  & \checkmark  & $\upchi$   & \checkmark\\ \hline
\cite{lappas2009finding} & \tf &  $\upchi$ & \checkmark  &  $\upchi$  & \checkmark\\ \hline%
\cite{gajewar2012multi} &  \tf &  $\upchi$  & \checkmark  & \checkmark   & \checkmark\\ \hline
\cite{kargar2011discovering}  & \tf &  $\upchi$ & \checkmark  &  $\upchi$  & \checkmark\\ \hline%
\cite{sozio2010} & \cs  & \checkmark &  $\upchi$  & \checkmark  & \checkmark\\ \hline
\cite{CuiXWLW13} & \cs  & \checkmark &  $\upchi$  & \checkmark  &  $\upchi$\\ \hline
\cite{huang2015approximate} & \cs  & \checkmark &  $\upchi$  & \checkmark  & \checkmark\\ \hline
\cite{FangCLH16} & \acs  & \checkmark &  \checkmark & \checkmark  & $\upchi$\\ \hline
Ours & \acs  & \checkmark &   \checkmark  & \checkmark  & \checkmark\\ \hline
%0.9 & 10/65  & 45/2080 & 9/16 & 1.0/1.0 & \xin{\bf{0.992/0.441}}  & \xin{\bf{0.992/0.418}} \\ \hline
\end{tabular}
\vskip -0.2cm
\caption{A comparison of representative works on keyword search (\ks), team formation (\tf), community search (\cs) and attributed community search (\acs). }\label{tab.cmp}
\vskip -0.5cm
\end{table}

\stitle{Community Search.} Community search on a graph aims to find densely connected communities containing query nodes, and has attracted a great deal of attention recently. Various models based on different dense subgraphs have been proposed and studied: quasi-clique \cite{CuiXWLW13}, densest subgraph \cite{wu2015robust}, k-core \cite{sozio2010,cui2014local,barbieri2015efficient} and k-truss \cite{huang2014,huang2015approximate}. All these works focus on the structure of the community while ignoring node attributes. This can result in communities with poor cohesion in the attribute sets of the community nodes. In particular, while \cite{huang2014,huang2015approximate} use $k$-truss as the basis structure of communities, the $k$-truss communities they find are not guaranteed to have high cohesion in the attribute sets of the nodes.

\stitle{Keyword Search.} Keyword search in relational databases has been extensively studied. Most of the works focus on finding minimal connected tuple trees from a relational database \cite{dbxplorer02,hristidis2002discover,hristidis2003efficient,bhalotia2002keyword,kacholia2005bidirectional,ding2007finding}.  There are two basic approaches:  DBXplorer \cite{dbxplorer02} DISCOVER-I \cite{hristidis2002discover}, and  DISCOVER-II \cite{hristidis2003efficient} use SQL to find tuple-trees. The other approach materializes a relational database as a graph, and finds trees from the graph: e.g., see BANKS-I  \cite{bhalotia2002keyword} and BANKS-II \cite{kacholia2005bidirectional}. Keyword search over graphs finds a substructure containing all or a subset of the input keywords. The works  \cite{li2008ease,qin2009querying}  report  subgraphs instead of trees as keyword search answers. However, keyword search does  not consider the cohesive structure involving the query nodes and keywords. As illustrated in the introduction, keyword search cannot return the right communities over attributed graphs. 
\eat{As a result, keyword search can be viewed as a restricted version of  community search over attributed graphs.} 

%Good properties of subgraphs. Recently, it has been shown that finding subgraphs rather than trees canbe more useful and informative for the users. However, the current tree or graph based methods may produce answers in which some content nodes (i.e., nodes that contain input keywords) are not very close to each other. In addition, when searching for answers, these methods may explore the whole graph rather than only the content nodes. This may lead to poor performance in execution time

\stitle{Team Formation.} Lappas et al. \cite{lappas2009finding} introduced the problem of discovering a team of experts from a social network, that satisfies all attributed skills required for a given task with low communication cost. Kargar and An \cite{kargar2011discovering} study the team formation problem with a team leader who communicates with each team member to monitor and coordinate the project.  Most of the team formation studies focus on a tree substructure, as opposed to densely connected subgraph required by community search. Gajewar and Sarma \cite{gajewar2012multi} extend the team formation problem to allow for potentially more than one member possessing each required skill, and use  maximum density measure or  minimum diameter as the  objective. Compared with our problem, these studies do not consider both dense structure and distance constraint at the same time, and also have no constraint on query nodes. 

\note[Laks]{Has anyone studied the TF problem as forming a team containing a given set of members?}

\stitle{Community Detection in Attributed Graphs. } Community detection in attributed graphs is to find all densely connected components with homogeneous attributes \cite{zhou2009graph,  cheng2012clustering, ruan2013efficient}. Zhou et al.\cite{zhou2009graph} model the community detection problem as graph clustering, and combine structural and attribute similarities through a unified distance measure. When high-dimensional attributed communities are hard to interpret or discover, \cite{huang2015dense,gunnemann2011db} consider subspace clustering on high-dimensional attributed graphs. A survey of clustering on attributed graphs can be found in \cite{bothorel2015clustering}. \cut{Community detection in attributed graphs is to find all communities of the entire graph, which is clearly different from our  goal of query-based community search. \xin{Moreover, it is practically hard and inefficient to adapt the above community detection approaches \cite{zhou2009graph, huang2015dense, ruan2013efficient} for online attributed community search: community detection is inherently global and much of the work involved may be irrelevant to the community being searched.}}
\add{It is practically hard and inefficient to adapt the above community detection approaches \cite{zhou2009graph, yang2013overlapping, huang2015dense, ruan2013efficient} for online attributed community search: community detection is inherently global and much of the work involved may be irrelevant to the community being searched.}
\eat{due to the global criteria of community detections and the inefficiency computations on highly dynamic massive graphs.}

Recently, Yang et al.\cite{FangCLH16} have proposed a model for community search over attributed graphs based on $k$-cores. The key distinction with our work is as follows. (1) Our community model is based on $k$-trusses, which have well-known advantages over $k$-cores such as denser structure. A connected $k$-core has no guarantee to be $2$-edge-connected, even with a large core value $k$. (2) Our search supports  multiple query nodes whereas theirs is limited to a single query node. (3) Their approach may miss useful communities. E.g., \new{consider the example graph in Figure~\ref{fig.community} with query node $\{q_2\}$ {\sl and} attributes \{DB, DM\}, and parameter $k=3$. \LL{Their model will return the subgraphs $H_1$ (Figure~\ref{fig.subcom}(a)) and $H_3$ (Figure~\ref{fig.subcom}(c)) as answers.} However, the subgraph $H_2$ (Figure~\ref{fig.subcom}(a)) will not be discovered, due to their strict homogeneity constraints. }
\cut{(4) Furthermore, unlike them, we minimize the query distance of the community which has the benefit of avoiding the free rider effect.} (5) Finally, unlike them, we validate our model with experiments over datasets with ground-truth communities. 

\eat{
\note[Laks]{We need to say something about whether it's easily possible to adapt some of the works in this paragraph to solve attributed community search. Also, why can't one apply attribute community detection and then search through all found communities to find the communities of interest, instead of doing community search?} 
\note[Xin]{Please refer to the above.}
}

%% file: problem.tex
\eat{ 
\note[Laks]{Somewhere rationalize the design choices made here for structure, cohesiveness, and attibute score.} 
\note[Xin]{See the following descriptions.}

In this section, we define basic definitions in graph theory and design three detailed constraints of cohesiveness, communication and attributes.  Based on these three constraints, we propose our attribute-driven truss community(\ktc) model. \xin{First, we introduce the structural bias of our model, that is built on the dense subgraph concept of k-truss. Then, we design the communication metric based on the definition of query distance, which measures the maximum distance from community member to query nodes. In addition, we will formulate the attibute functions that trend to perfer the community with a large number of homogeneous attributes.}
} 

In this section, we develop a notion of attributed community by formalizing the %intuitions behind 
the desiderata discussed in Section~\ref{sec:prelims}.  We focus our discussion on conditions 2--4.

\subsection{(k, d)-truss}

In the following, we introduce a novel definition of dense and tight substructure called \kdtruss by paying attention to cohesiveness and communication cost. 

\eat{Before that, we respectively describe the bisas to satisfy the criterion of cohesiveness and communication cost.} 

%\subsection{Cohesiveness}
\stitle{Cohesiveness.} While a number of definitions for dense subgraphs have been proposed over the years, we adopt the $k$-truss model, proposed by Cohen~\cite{cohen2008}, which has gained popularity and has been found to satisfy nice properties. 

A subgraph $H\subseteq G$ is a $k$-core, if every vertex in $H$ has degree at least $k$. 
A \emph{triangle} in $G$ is a cycle of length 3.  We denote a triangle involving vertices $u, v, w \in V$ as $\triangle_{uvw}$.  The \emph{support} of an edge $e(u,v)\in E$ in $G$, denoted $sup_{G}(e)$, is the number of triangles containing $e$, i.e., $sup_{G}(e) = |\{\triangle_{uvw}: w\in V\}|$. When the context is obvious, we drop the subscript and denote the support as $sup(e)$. \eat{Based on the definition of $k$-truss  \cite{cohen2008, WangC12}, we define a connected $k$-truss below.} \new{Since the definition of $k$-truss  \cite{cohen2008, WangC12} allows a $k$-truss to be disconnected, we define a connected $k$-truss below.}

\eat{
\vskip -0.1in
\begin{table}[t]
\begin{center}\vspace*{-0.45cm}
\scriptsize
%\small
\caption[]{\textbf{Frequently Used Notations}}\label{tab:notations}
\begin{tabular}{|c|c|}
\hline
Notation & Description \\ \hline \hline
$G = (V(G), E(G))$ &  An undirected and connected simple graph $G$\\ \hline
$n;m$ & The number of vertices/edges in $G$ \\ \hline
$N(v)$	& The set of neighbors of $v$\\ \hline
$\sup_H(e)$	& The \emph{support} of edge $e$ in $H$ \\ \hline
$\tau(H)$	& Trussness of graph $H$  \\ \hline
%, $\tau(H) = 2+\min_{e\in E(H)}\{sup_{H}(e)\}$\\ \hline
$\tau(e)$	& Trussness of edge $e$ \\ \hline
%, $\tau(e)= \max_{H\subseteq G \wedge e\in E(H)}\{\tau(H)\}$\\ \hline	
$\tau(v)$	& Trussness of vertex $v$ \\ \hline
%, $\tau(v)= \max_{H\subseteq G \wedge v \in V(H)} $ $\{\tau(H)\}$ \\ \hline
$\taubar(S)$  & The maximum trussness of connected graphs containing $S$ \\ \hline
$\diam(H)$ & The diameter of graph $H$ \\ \hline
$\dist_H(v, u)$ & The shortest distance between $v$ and $u$ in $H$ \\ \hline
$\dist_H(R, Q)$ & $\dist_H(R, Q) = \max_{v\in R, u\in Q} \dist_H(v,u)$ \\ \hline
\end{tabular}\vspace*{-0.6cm}
\end{center}
\end{table}
\vskip -0.15in
}

\begin{definition}
 [Connected K-Truss] Given a graph $G$ and an integer $k$, a connected $k$-truss is a connected subgraph $H \subseteq G$,  such that $\forall e\in E(H)$, $sup_{H}(e)$ $\geq (k-2)$. 
\end{definition} \label{def.ktruss}

Intuitively, a connected $k$-truss is a connected subgraph in which each connection (edge) $(u,v)$ is ``endorsed'' by $k-2$ common neighbors of $u$ and $v$ \cite{cohen2008}. A connected $k$-truss with a large value of $k$ signifies strong inner-connections between members of the subgraph. In a $k$-truss, each node has degree at least $k-1$, i.e., it is a $(k-1)$-core\add{.}
\cut{, \xin{and a connected $k$-truss is also $(k-1)$-edge-connected, i.e., it remains connected if fewer than $(k-1)$ edges are removed \cite{BatageljZ03}. }}

%and a connected $k$-truss is also $(k-1)$-connected, i.e., it remains connected if fewer than $(k-1)$ vertices are removed \cite{BatageljZ03}. 
%\xin{A $k$-core is that each vertex has at least $k$ degree within this graph.}

\eat{ 
\note[Laks]{$k$-core not defined so far.} 
\note[Xin]{$k$-core is defined above.}
} 

%LAKS: THIS IS HOW FAR I GOT AS OF 12:45 PM THURSDAY. 

\begin{example}
%\note[Xin]{Will add the vertex id into Figure 1, 2, 3.}
Consider the graph $G$ (Figure~\ref{fig.community}). The edge $e(v_1, v_2)$ is contained in three triangles $\triangle_{q_1v_1v_2}$, $\triangle_{q_2v_1v_2}$ and $\triangle_{v_3v_1v_2}$, thus its support is $sup_{G}(e)=3$. Consider the subgraph $H_3$ of $G$ (Figure \ref{fig.subcom}(c)). Every edge of $H_3$ has support $\geq 2$, thus $H_3$ is a 4-truss. Note that even though the edge $e(v_1, v_2)$ has support 3, there exists no 5-truss in the graph $G$ in Figure \ref{fig.community}. % Note that the trussness of an edge $e$ of a graph $G$ could be less than $sup_G(e)+2$, e.g., $\tau(e(q_2, v_2))$ $ = 4 < 5$ $ = sup(e(q_2, v_2))+2$.  Moreover, the vertex trussness of $q_2$ is also 4, i.e. $\tau(q_2) = 4$.
\end{example}

\note[Laks]{Where are nodes $q_2$ and $v_5$ in Fig. 1? Need to check the correctness of the claims in the above example.} 
\note[Xin]{Will do.}

%\subsection{Communication Cost}
\stitle{Communication Cost.}  For two nodes $u, v \in G$,  let  $\dist_{G}(u, v)$ denote the length of the shortest path between $u$ and $v$ in $G$, where $\dist_{G}(u, v) = +\infty$ if $u$ and $v$ are not connected. The diameter of a graph $G$ is the maximum length of a shortest path in $G$, i.e., $\diam(G) =  \max_{u,v \in G} \{\dist_G(u,v)\}$. 
\eat{Obviously, if a graph is disconnected, the diameter is $+\infty$.}   %\laks{We make use of the notions of graph query distance and diameter in the rest of the paper.} 
We make use of the notion of graph query distance in the following.

\eat{ 
\note[Laks]{Should we cite our or previous papers where these notions were first defined?} 
\note[Xin]{Done.}
} 

\begin{definition}
[Query Distance \cite{huang2015approximate}] Given a graph $G$ and query nodes $V_q \subseteq V$, the vertex query distance of vertex $v\in V$ is the maximum length of a shortest path from $v$ to a query node $q\in V_q$ in $G$, i.e., $\dist_G(v, V_q) = \max_{q\in V_q} $ $ \dist_G(v, q)$. Given a subgraph $H\subseteq G$ and $V_q\subseteq V(H)$, the graph query distance of $H$ is defined as $\dist_H(H, V_q) = $ $\max_{u\in H} \dist_{H}(u, V_q) $ $ =\max_{u\in H, q\in V_q}\dist_H(u, q).$
%Given a subgraph $H\subseteq G$ and a set of query nodes $V_q \subseteq V(H)$, the vertex query distance of vertex $v\in H$ is the maximum length of a shortest path from $v$ to a query node $q\in V_q$, i.e., $\dist_H(v, V_q) = \max_{q\in V_q} $ $ \dist_H(v, q)$. The graph query distance of $H$ is defined as $\dist_H(H, V_q) = $ $\max_{u\in H} \dist_{H}(u, V_q) $ $ =\max_{u\in H, q\in V_q}\dist_H(u, q).$
\end{definition}  \label{def.maxdis}

\cut{Given a subgraph $H\subseteq G$ and $V_q \subseteq V(H)$, the query distance $\dist_H(H, V_q)$ measures the communication cost between the members of $H$ and the query nodes. A good community should have a low communication cost with small $\dist_H(H, V_q)$. }
\add{Given a subgraph $H\subseteq G$, the query distance $\dist_H(H, V_q)$ measures the communication cost between the members of $H$ and the query nodes. A good community should have a low communication cost with small $\dist_H(H, V_q)$. }

For the graph $G$ in Figure \ref{fig.community} and query nodes $V_q=\{q_1, q_2\}$, the vertex query distance of $v_7$ is $\dist_G(v_7, V_q)= $ $\max_{q\in V_q}$ $\{\dist_G(v_7, q)\}$ $= 2$. Consider the subgraph $H_1$ in Figure~\ref{fig.subcom}(a). Then graph query distance of $H_1$ is $\dist_{H_1}(H_1,V_q)= \dist_{H_1}(q_1, q_2) = 2$. The diameter of $H_1$ is $\diam(H_1) = 2$. 
%
%The graph query distance $\dist_G(G,Q) = $ $ \max_{u\in G} $ $ \dist_{G}(u, Q) = $ $\dist(p_1, q_2)$ $ = 3$. The diameter of $G$ is $\diam(G) = \laks{3}$\laks{, while that of the subgraph shaded gray is $4$}. \qed  

\note[Laks]{Make sure the claims in the example above are correct.} 
\note[Xin]{Will do.}

%\stitle{\kdtruss.} 
\stitle{(k, d)-truss.} We adapt the notions of $k$-truss and query distance, and propose a new notion of \kdtruss capturing dense cohesiveness and low communication cost. 

\begin{definition}
[\kdtruss] Given a graph $H$, query nodes $V_q$, and  numbers $k$ and $d$, we say that $H$ is a \kdtruss iff $H$ is a connected $k$-truss containing $V_q$ and $\dist_H(H, V_q)\leq d$. 
\end{definition}

By definition, the cohesiveness of a \kdtruss increases with $k$, and its proximity to query nodes increases with decreasing $d$. For instance, the community $H_1$ in Figure \ref{fig.subcom} (a) for $V_q=\{q_1, q_2\}$ is a \kdtruss with $k=4$ and $d=2$.

\subsection{Attribute Score Function} 
\eat{ 
\note[Laks]{Do you think it's a good idea to state the 2-3 principles guiding the choice of a score function, that I wrote on the board in my office when we were discussing?} 
\note[Xin]{Yes, it is a good idea. Consider the following principles.}
} 

We first identify key properties that should be obeyed by a good attribute score function for a community. Let $\keyw(H, W_q)$ denote the attribute score of community $H$ w.r.t. query attributes $W_q$. We say that a node $v$ of $H$ covers an attribute $w\in W_q$, if $w \in \doc(v)$. We say that a node of $H$ is irrelevant to the query if it does not cover any of the query attributes. \\ 
\underline{\prin1}: The more query attributes that are covered by some node(s) of $H$, the higher should be the score $\keyw(H, W_q)$. The rationale is obvious. \\ 
\underline{\prin2}: The more nodes contain an attribute $w\in W_q$, the higher the contribution of $w$ should be toward the overall score $\keyw(H, W_q)$. The intuition is that attributes that are covered by more nodes of $H$ signify homogeneity within the community w.r.t. shared query attributes. \\  \underline{\prin3}: The more nodes of $H$ that are irrelevant to the query, the lower the score $\keyw(H, W_q)$. 

\eat{ 
\xin{Obviously, there exists several basic principle to be satisfied by attribute function for $\keyw(H, W_q)$. First of all, (\prin1), for a given $H$, the more input attributes covered by somenodes in $H$, the better the overall score. Therefore, the attribute function needs consider the coverage of input attributes $W_q$ and the a total of contributions by every attribute attribute $w\in W_q$ in $H$.  Thus, it can be modeled as  $\keyw(H, W_q) = \sum_{w\in W_q} \score(H, w)$, where $\score(H, w)$ is the relevance score of $H$ on input attribute $w \in W_q$.  Second, (\prin2), for each attribute $w$, the more nodes in $H$ contains $w$, better the score for $w$, i.e., $\score(H, w)$ should be positively correlated with $|V(H)\cap V_{w}|$. In addition, (\prin3), the more nodes irrelavent with input attributes, the worse the total score, due to the homogenous attributes in one community.}
} 

We next discuss a few choices for defining $\keyw(H, W_q)$ and analyze their pros and cons, \LL{before presenting an example function that satisfies all three principles.} \LL{Note that the scores $\keyw(H, W_q)$ are always compared between subgraphs $H$ that meet the same structural constraint of \kdtruss.} An obvious choice is to define $\keyw(H, W_q) := \sum_{w\in W_q} \score(H, w)$, where $\score(H, w)$, the contribution of attribute $w$ to the overall score, can be viewed as the relevance of $H$ w.r.t. $w$. This embodies \prin1 above. Inspired by \prin2, we could define $\score(H,w) := |V(H)\cap V_{w}|$, i.e., the number of nodes of $H$ that cover $w$. Unfortunately, this choice suffers from some limitations by virtue of treating all query attributes alike. Some attributes may not be shared by many community nodes while others are and this distinction is ignored by the above definition of $\keyw(H, W_q)$. To illustrate, consider the community $H_1$ in Figure \ref{fig.subcom}(a) and the query $Q=(\{q_1\}, \{DB\})$; $H_1$ has 5 vertices associated with the attribute $DB$ and achieves a score of 5. The subgraph $H$ of the graph $G$ shown in Figure \ref{fig.community} also has the same score of 5. However, while the community in Figure \ref{fig.subcom}(a) is clearly a good community, as all nodes carry attribute $DB$,  the subgraph $H$ in Figure \ref{fig.community} includes several irrelevant nodes without attribute $DB$. Notice that both $H_1$ and $H$ are 4-trusses so we have no way of discriminating between them, which is undesirable. 

An alternative is to define $\score(H, w)$ as $\frac{| V_{w} \cap V(H)|}{|V(H)|}$ as this captures the popularity of attribute $w$. Unfortunately, this fails to reward larger commumities. For instance, consider the query $Q = (\{q_1,v_4\}, \{DB\})$ over the graph $G$ in Figure~\ref{fig.community}. The subgraph $H_1$ in Figure~\ref{fig.subcom}(a) as well as its subgraph obtained by removing $q_2$ is a 4-truss and both will be assigned a score of 1. 

In view of these considerations, we define $\keyw(H, W_q)$ as a weighted sum of the score contribution of each query attribute, where the weight reflects the popularity of the attribute.

\begin{definition}[Attribute Score]
Given a subgraph $H\subseteq G$ and an attribute $w$, the weight of an attribute $w$ is $\theta(H,w) = \frac{| V_{w} \cap V(H)|}{|V(H)|}$, i.e., the fraction of nodes of $H$ covering $w$. For a query $Q=(V_q,W_q)$ and a community $H$, the attribute score of $H$ is defined as $\keyw(H,W_q) = \sum_{w\in W_q} \theta(H,w)\times \score(H,w)$, where $\score(H,w) = |V_w\cap V(H)|$ is the number of nodes covering $w$. 
\eat{ 
with $w$ is defined as $\score(H, w)$ $= |V_{w} \cap V(H)|$ $\cdot$ $\theta(H, w)$ $=\frac{| V_{w} \cap V(H)|^2}{|V(H)|}$, where $|V_{w} \cap V(H)| = |\{v: v\in V(H), w \in \doc(v)\}|$ is the number of vertices in $H$ associated with the attribute $w$. } 
\end{definition}

The contribution of an attribute $w$ to the overall score is $\theta(H,w)\times  \score(H,w) = \frac{| V_{w} \cap V(H)|^2}{|V(H)|}$. This depends not only on the number of vertices covering $w$ but also on $w$'s popularity in the community $H$. This choice discourages vertices unrelated to the query attributes $W_q$ which decrease the relevance score, without necessarily increasing the cohesion (e.g., trussness). \LL{At the same time, it permits the inclusion of essential nodes, which are added to a community to reduce the cost of connecting query nodes. 
%\footnote{Steiner nodes represent important hub vertices that are connected other vertices together. Note that they are not extra vertices out of this graph.} 
They act as an important link between nodes that are related to the query, leading to a higher relevance score. We refer to such additional nodes as \emph{steiner nodes}.} E.g., consider the query $Q = (\{q_1\}, \{ML\})$ on the graph $G$ in Figure~\ref{fig.community}. As discussed in Section~\ref{sec:intro}, the community $H_4$ in Figure \ref{fig.subcom}(d) is preferable to the chain of nodes $v_8, q_1, v_{10}$. 
%\new{\\================\\The following node $v_3$ should be $v_9$.\\================\\}
\LL{Notice that it includes $v_9$ with attribute $DM$ (but not $ML$); $v_9$ is thus a steiner node.} 
It can be verified that $\keyw(H_4,W_q) = \frac{9}{4}$ which is smaller than the attribute score of the chain, which is $3$. However, $H_4$ is a 3-truss whereas the chain is a 2-truss. It is easy to see that any supergraph of $H_4$ in Figure~\ref{fig.community} is at most a 3-truss and has a strictly smaller attribute score.

The more query attributes a community has that are shared by more of its nodes, the higher its attribute score. For example, consider the query $Q=(\{q_1\}, \{DB, DM\})$ on our running example graph of Figure~\ref{fig.community}. The communities $H_1, H_2, H_3$ in Figure~\ref{fig.subcom} are all potential answers for this query. We find that $\keyw(H_1, W_q) = 5\cdot 1 + 2\cdot \frac{2}{5}=5.8$; by symmetry, $\keyw(H_3, W_q) = 5.8$; on the other hand, $\keyw(H_2, W_q) = 5\cdot \frac{5}{8} + 5\cdot \frac{5}{8}=6.25$. Intuitively, we can see that $H_1$ and $H_3$ are mainly focused in one area (DB or DM) whereas $H_2$ has 5 nodes covering DB and DM each and also has the highest attribute score.

\begin{remark}
\LL{We stress that the main contribution of this subsection is the identification of key principles that an attribute score function must satisfy in order to be effective in measuring the goodness of an attributed community. Specifically, these principles capture the important properties of high attribute coverage and high attribute correlation within a community and minimal number of nodes irrelevant to given query. Any score function can be employed as long as it satisfies these principles. The algorithmic framework we propose in Section \ref{sec.basic} is flexible enough to  handle an \atc community model equipped with any such score function. 

We note that a natural candidate for attribute scoring is the entropy-based score function, defined as $\kw{f_{entropy}}(H, W_q)= $ $ \sum_{w \in W_q} $  $ - \frac{| V_{w} \cap V(H)|}{|V(H)|}$  $ \log{\frac{| V_{w} \cap V(H)|}{|V(H)|}}$. It measures homogeneity of query attributes very well. However, it fails to reward larger communities, specifically violating Principle 1. E.g., consider the query $Q = (\{q_1,v_4\}, \{DB\})$ on the graph $G$ in Figure~\ref{fig.community}. The subgraph $H_1$ in Figure~\ref{fig.subcom}(a) and its subgraph obtained by removing $q_2$ are both 4-trusses and both are assigned a score of 0. Clearly, $H_1$ has more nodes containing the query attribute DB. 
} 
\end{remark}

\subsection{Attributed Truss Community Model}

\eat{
On the basis of the definitions of connected $k$-truss and graph query distance, we define the $d$-close-$k$-truss as follow.

\begin{definition}
[$d$-close-$k$-truss] Given a graph $H$, and two parameters $k$ and $d$, $H$ is a $d$-close-$k$-truss iff $H$ is a connected $k$-truss and $\dist_H(H, V_q)\leq d$. 
\end{definition}
}

Combining the structure constraint of \kdtruss and the attribute score  function $\keyw(H, W_q)$, we define an  \emph{attributed truss community} (\atc) as follows. 

\begin{definition}\label{def.community} 
[Attribute Truss Community] Given a graph $G$ and a query $Q = (V_q, W_q)$ and two numbers $k$ and $d$,   
%= $\{q_1, ..., q_r\}$, 
$H$ is an attribute truss community (\atc), if $H$ satisfies the following conditions:
\begin{enumerate}
\vspace{-0.2cm}
  \item $H$ is a \kdtruss containing $V_q$.
  \vspace{-0.12cm}
  \item $H$ has the maximum attribute score $\keyw(H, W_q)$ among subgraphs satisfying condition (1). 
  %\item[(3) Maximized attribute Function.] $G'$ is a subgraph with the maximum attribute relevance satisfying conditions (1) and (2). That is, $\nexists G'' \subseteq G$, such that $\keyw(G'', W_q) $ $ > \keyw(G', W_q)$, and $G''$ satisfies conditions (1) and (2). 
    \vspace{-0.22cm}
 %\vspace{-0.1cm}
\end{enumerate}
\end{definition}

In terms of structure and communication cost, condition (1) not only requires that the community containing the query nodes $V_q$ be densely connected, but also that each node be close to the query nodes. In terms of query attribute coverage and correlation, condition (2) ensures that as many query attributes as possible are covered by as many nodes as possible.

\begin{example}
For the graph $G$ in Figure \ref{fig.community}, and query $Q=(\{q_1, q_2\}, \{DB, DM\})$ with $k=4$ and $d=2$, $H_2$ in \ref{fig.subcom}(b) is the corresponding \atc, since $H_2$ is a $(4,2)$-truss with the largest score $\keyw(H, W_q) = 6.25$ as seen before. 
\end{example}

The \atcp studied in this paper can be formally formulated as follows.

%\stitle{\atcp:} 
\stitle{Problem Statement:} Given a graph $G(V,E)$, query $Q = (V_q, W_q)$ and two parameters $k$ and $d$, find an \atc $H$, such that $H$ is a \kdtruss with the maximum attribute score $\keyw(H, W_q)$.

We remark that in place of the \kdtruss with the highest attribute score, we could consider the problem of finding the $r$ \kdtrusses with the highest attribute score. Our technical results and algorithms easily generalize to this extension.

%% file: character.tex
\eat{In this section, we discuss the properties of our community model \ktc, in terms of bias structure and attribute function. In addition,, we discuss the variant of keyword functions and the parameter setting.  Finally, we also analyze the hardness of \ktcp, and show its \emph{NP}-hard. } 

\eat{ 
There is a structural component and an attribute score component to our definition of \ktc. In this section, we analyze both of them and establish some useful properties which will be exploited by our algorithms. Furthermore, we also prove that the \ktcp is NP-hard. 
} 

In this section, we analyze the complexity of the problem and show that it is NP-hard. We then analyze the properties of the structure and attribute score function of our problem. Our algorithms for community search exploit these properties. 

%\subsection{Community Property Analysis}

\subsection{Hardness}
\input{hardness}

In view of the hardness, a natural question is whether efficient approximation algorithms can be designed for \ktcp. Thereto, we investigate the properties of the problem in the next subsections. Observe that from the proof, it is apparent that the hardness comes mainly from maximizing the attribute score of a \ktc. 

\subsection{Properties of \kdtruss}
Our attribute truss community model is based on the concept of $k$-truss, so the communities inherit good structural properties of $k$-trusses, such as  \emph{k-edge-connected}, \emph{bounded diameter} and \emph{hierarchical structure}. In addition, since the attribute truss community is required to have a bounded query distance, it will have a small diameter, as explained below.  
\note[Laks]{Hierarchical structure is not explained.} 
\note[Xin]{? Moreover, $k$-truss based community has \emph{hierarchical structure} that represents the cores of a community at different levels of granularity \cite{huang2014}, that is, $k$-truss is always contained in the $(k-1)$-truss for any $k\geq 3$.}

%\stitle{K-edge-connected, hierarchical structure, small diameter.} 
A $k$-truss community is ($k-1$)-edge-connected, since it remains connected whenever fewer than $k-1$ edges are deleted from the community \cite{cohen2008}. Moreover, a $k$-truss based community has \emph{hierarchical structure} that represents the hearts of the community at different levels of granularity \cite{huang2014}, i.e., a $k$-truss is always contained in some $(k-1)$-truss, for $k\geq 3$.  In addition, for a connected $k$-truss with $n$ vertices, the diameter is at most  $\lfloor\frac{2n-2}{k}\rfloor$ \cite{cohen2008}. \RV{Small diameter is considered an important property of a good community \cite{Edachery99graphclustering}. } 

Since the distance function satisfies the triangle inequality, i.e., for all nodes $u, v, w$, $\dist_G(u,v) \leq \dist_G(u,w)+\dist_G(w, v)$, we can express the lower and upper bounds on the community diameter in terms of the query distance as follows. 

%****Should we say that $H$ is a \kdtruss?****  $H(V,E)$

\begin{observation}\label{lemma.disdia}
%\cite{huang2015approximate} 
For a \kdtruss $H$ and a set of nodes $V_q\subseteq H$, we have $d\leq \diam(H) \leq \min\{\frac{2|V(H)|-2}{k}, 2d\}$.
\end{observation}

\eat{ 
% \begin{proof}
% First, the diameter $\diam(H) = \max_{v, u\in H} \dist_H(v, u)$, which is clearly no less than  $\dist_H(H, Q) = \max_{v\in H, q \in Q}$ $ \dist_{H}(v, q)$ for $Q\subseteq G$. Thus, $\dist_H(H, Q) \leq \diam(H)$. 
% Second, suppose that the longest shortest path in $G$ is between $v$ and $u$. Then  $\forall q\in Q$, then we have $\diam(H) = \dist(v, u)\leq \dist(v, q) + $ $\dist(q, u) \leq $ $2\dist_H(H, Q)$. The lemma follows.
% \end{proof}
According to the query distance constraint in \atc model, a $d$-close-$k$-truss has the diameter of value no greater than $2d$, i.e., the distance between any pair of nodes  will not exceed $2d$ in the \ktc. As a result, a tighter diameter bound can be directly derived as follow.  

\begin{corollary}\label{lemma.diam}
Given an \atc $H$ as a $d$-close-$k$-truss, the diameter of $H$ has $\diam(H) \leq \min\{\frac{2|V(H)|-2}{k}, 2d\}$.
\end{corollary}
}

\begin{remark}
\new{Besides $k$-truss, there exist several other definitions of dense subgraphs including:  $k$-($r$,$s$)-nucleus \cite{sariyuce2015finding}, quasi-clique \cite{CuiXWLW13}, densest subgraph \cite{wu2015robust}, and $k$-core \cite{sozio2010}. %It is well known that finding  quasi-cliques and the densest subgraph is NP-hard problem. 
\LL{A $k$-($r$,$s$)-nucleus, for positive integers $k$ and $r< s$, is a maximal union of $s$-cliques in which every $r$-clique is present in at least $k$ $s$-cliques, and every pair of $r$-cliques in that subgraph is connected via a sequence of $s$-cliques containing them. Thus, $k$-($r$,$s$)-nucleus is a generalized concept of $k$-truss, which can achieve very dense strucutre for large parameters $k$, $r$, and $s$. However, 
%due to the clique compuation by definition, 
finding $k$-($r$,$s$)-nucleus incurs a cost of $\Omega(m^{\frac{s}{2}})$ time  where $m$ is the number of edges, which is more expensive than the $O(m^{1.5})$ time taken for computing $k$-trusses, whenever $s>3$. 
A detailed comparison of $k$-truss and other dense subgraph models can be found in \cite{huang2015approximate}. In summary, $k$-truss is a good dense subgraph model that strikes a balance between good structural property and efficient computation.} }
\end{remark}

\subsection{Properties of attribute score function} \label{sec.scorepro}
\input{keyword}

% $\kw(H, W_q)$ satisfies the monotonical submodurity on $H$.

%% file: hardness.tex
%\subsection{Hardness and Approximation.}

\eat{ 
In this section, we will introduce the definition of edge density and a variant of densest subgraph problems. Based on the edge density, we proposed a new density function called vertex-edge density, and give a theorem showing the NP-hard problem of computing the maixmum vertex-edge density with at least $k$ vertices (\VEDalK-problem) in Theorem \ref{theorem.wden}. Then, we reduce the \VEDalK-problem to \ktcp, which shows the NP-hardness of \ktcp in Theorem \ref{theorem.kcs}.
} 

Our main result in this section is that the \ktcp is NP-hard (Theorem~\ref{theorem.kcs}). The crux of our proof idea comes from the hardness of finding the densest subgraph with $\ge k$ vertices \cite{khuller2009finding}. Unfortunately, that problem cannot be directly reduced to our \ktcp. To bridge this gap, we extend the notion of graph density to account for vertex weights and define a helper problem called \WDalK\ -- given a graph, find the subgraph with maximum ``weighed density'' with at least $k$ vertices. We then show that it is NP-hard and then reduce the \WDalK to our problem. 

\note[Laks]{I think overall the logic of the hardness proof is sound. However, the presentation is complex and hard to follow. Here is what I think we need to show. Set up a new problem called VEDalK (need a better name) and show that finding if a given graph $G$ contains a subgraph with $\ge K$ vertices whose ``density'' is greater than a given threshold is NP-hard. This is achieved in two stages: first we establish this for VEDK (exactly K) and then modify the proof for VEDalK. 

Next, we reduce this problem to that of finding a subgraph with maximum kw score. More precisely, given $G$, $K$, and some threshold for density, we construct an instance $G'$ with the same $K$ and a slightly modified threshold', s.t. $G$ contains a subgraph with $\ge$ $K$ vertices whose density is $\ge$ threshold iff $G'$ contains a subgraph with $\ge K$ vertices whose attribute score is $\ge$ threshold'. I don't see why the last para is needed.} 

\note[Xin]{Totally agree. Using this framework, we can make the proof simpler. Reduce the problem that a given graph $G'$ contains a subgraph with $\ge K$ vertices whose ``density'' is greater than a given threshold $\gamma=(K-1)/2$, to VEDK. }

\stitle{Weighted Density.} Let $G=(V, E)$ be an undirected graph. Let $\w(v)$ be a non-negative weight associated with each vertex $v\in V$. Given a subset $S\subseteq V$, the subgraph of $G$ induced by $S$ is $G_{S}=(S, E(S))$, where $E(S)= \{(u,v)\in E\mid u, v\in S\}$. For a vertex $v$ in a subgraph $H\subseteq G$, its degree is $\deg_{H}(v)=|\{(u,v)\mid (u,v)\in E(H)\}|$. Next, we define: 

\begin{definition}
[Weighted Density.] Given a subset of vertices $S\subseteq V$ of a weighted graph $G$, the weighted density of subgraph $G_S$ is defined as $\chi(G_S)= \sum_{v\in S}\frac{\deg_{G_S}(v)+\w(v)}{|S|}$. 
\label{def.wden}
\end{definition}

Recall that traditional edge density of an induced subgraph $G_S$ is $\rho(G_S)=\frac{|E(S)|}{|S|} $ $=\sum_{v\in S} \frac{\deg_{G_S}(v)}{2|S|} $ \cite{khuller2009finding, bahmani2012densest}. That is, $\rho(G_S)$ is twice the average degree of a vertex in $G_S$. 
Notice that in Definition \ref{def.wden}, if the weight of $v$ is $\w(v)=0$,  $\forall v$, then the weighted density $\chi(G_S) = 2 \rho(G_S)$. It is well known that finding a subgraph with the maximum edge density can be solved optimally using parametric flow or linear programming relaxation \cite{bahmani2012densest}. However, given a number $k$, finding the maximum density of a subgraph $G_S$ containing at least $k$ vertices is NP-hard \cite{khuller2009finding}. 

Define the weight of a vertex $v$ in a graph $G$ as its degree in $G$, i.e., $\w(v)=\deg_{G}(v)$. Then, $\chi(G_S) = \sum_{v\in S} \frac{\deg_{G_S}(v)+\deg_{G}(v)}{|S|} = 2\rho(G_S) +\sum_{v\in S} \frac{\deg_{G}(v)}{|S|}$. 
\eat{ 
We use $\chi^*_{\geq k}(G) =  \max_{S\subseteq V, |S|\geq k} \{\chi(G_S)\}$  to denote the maximum weighted density of a subgraph $G_S$ containing at least $k$ vertices.
} 
We define a problem, the \WDalK, as follows: given a graph $G$ with weights as defined above, and a density threshold $\alpha$, check whether $G$ contains an induced subgraph $H$ with at least $k$ vertices such that $\chi(H)\ge \alpha$. We show it is NP-hard in Theorem \ref{theorem.wden}.
\cut{To establish this, we first show that the \WDK, i.e., finding whether $G$ has a subgraph $H$ with \emph{exactly} $k$ vertices with weighted density at least $\alpha$, i.e., $\chi(H) \ge \alpha$, is NP-hard.\footnote{Notice that the hardness of finding the maximum density subgraph with $\geq k$ vertices does \emph{not} imply hardness of \WDK for a specific weight function over the vertices and thus it needs to be proved.} We then extend this result to the hardness of the \WDalK.}

\cut{
\eat{ 
To prove it, we firstly show the NP-hard of \WDK, that is, given number $\alpha$, checking whether there exists a subgraph $G_S$ containing exactly $k$ nodes with weighted density no less than $\alpha$, i.e., $\chi^*_{= k}(G) =  \max_{S\subseteq V, |S|= k} \{\chi(G_S)\}\geq \alpha$. We show \WDK as NP-hard as follow. 
} 

\begin{lemma}\label{lemma.wden}
\WDK is NP-hard.
\end{lemma}

\begin{proof} 
We reduce the well-known NP-complete problem, CLIQUE, to \WDK. Given a graph $G = (V,E)$ with $n$ vertices and a number $k$, construct a graph $G' = (V\cup V', E\cup E')$ as follows. For each vertex $v\in V$, add $n-deg_G(v)$ new dummy vertices. $G'$ contains an edge connecting each $v\in V$ with each of its associated dummy vertices. Notice that the maximum degree of any node in $G'$ is $n$. In particular, every vertex in $V$ has degree $n$ in $G'$ whereas every dummy vertex in $V'$ has degree $1$. So for any $S\subseteq V\cup V'$, $\chi(G'_S) = 2\rho(G'_S) + \frac{\sum_{v\in G'_S}deg_{G'}(v)}{|S|} \le 2\rho(G'_S) + n$. Set $\alpha = n + k - 1$. We claim that $G$ contains a clique of size $k$ iff $\chi(G') \ge \alpha$. 

\noindent 
$(\Rightarrow)$: Suppose $H\subset G$ is a $k$-clique. Since each vertex $v$ of $G$ has degree $n$ in $G'$ and $2\rho(H)$ is the average degree of a vertex in $H$, we have $\chi(H) = 2\rho(H) + n = k-1+n$. 

\noindent 
$(\Leftarrow)$: Suppose $G'$ contains an induced subgraph $G'_S$ with $|S| = k$ and with $\chi(G'_S) \ge n+k-1$. It is clear that for any $S$ with $S\cap V'\ne \emptyset$, $\chi(G'_S) < n+k-1$. The reason is that vertices in $V'$ have degree $1 < n$ in $G'$. Thus, we must have $S\subset V$. Now, for any $S\subseteq V\cup V'$,  $\chi(G'_S)$ is upper bounded by $n+k-1$. Thus, $\chi(G'_S) = n+k-1$, and we can infer that $2\rho(G'_S) = k-1$, implying $G'_S = G_S$ is a $k$-clique.  
\end{proof}
}

\begin{theorem}\label{theorem.wden}
\WDalK is NP-hard.
\end{theorem}

\begin{proof}
\cut{We can reduce \WDK to \WDalK, using the  ideas similar to those used in reducing the densest $k$ subgraph problem to the densest at least k subgraph problem \cite{khuller2009finding}. }
\add{We reduce the well-known NP-hard problem of Maximum Clique (decision version) to this problem.  A complete proof is reported in the arXiv article \cite{ArxivATC}.}
\end{proof}

\begin{theorem}\label{theorem.kcs}
\ktcp is NP-hard.
\end{theorem}

\add{\begin{proof}
We reduce the \VEDalK 
 to \ktcp. The complete proof is
available in \cite{ArxivATC}.\qed 
\end{proof}	}

\cut{
\begin{proof}
We reduce the \VEDalK 
 to \ktcp. Given a graph $G=(V,E)$ with $|V|=n$ vertices, construct an instance $G'$ as follows. 
\eat{
number $k$ and $\alpha$, \VEDalK is to check whether exists a subset $S\subseteq V$ with $|S|\geq k$ having $\chi(G_S)\geq \alpha$. From this, construct an instance of \ktcp on a new graph $G'$ whether there exists an \kdtruss $H \subset G'$ with $\keyw(H, W_q) \geq \alpha$ for a query $Q=(V_q, W_q)$.   
} 
$G'$ is a complete graph over $n$ vertices. For simplicity, we use $V$ to refer to the vertex set of both $G$ and $G'$, without causing confusion. For each edge $(u,v) \in E(G)$, create a distinct attribute $w_{uv}$ \xin{for $G'$}. We assume $w_{uv}$ and $w_{vu}$ denote the same attribute. Then, with each vertex $v \in V$ \xin{in $G'$}, associate a set of attributes: $\doc(v) = \{w_{vu}: (v, u) \in E(G)\}$. Notice that  the cardinality of $\doc(v)$ is $|\doc(v)| = \deg_G(v)$. 
\eat{For any edge $e=(v, u)$ in $G$, $v$ and $u$ have the same attribute $w_{vu}$ in $G'$. Here, the attribute $w_{vu}$ and $w_{uv}$ are regarded as the same.} Also, an attribute $w_{vu}$ is present only in the attribute sets of $v$ and $u$, i.e., $V_{w_{vu}}=\{v, u\}$. 

For a vertex set $S\subset V$, we will show that $\keyw(G'_S, W_q)= \chi(G_S)$, where $G_S=(S, E(S))$ is the induced subgraph of $G$ by $S$, $G'_S=(S, E'(S))$ is the induced subgraph of $G'$ by $S$, and $W_q = \{w_{vu}: (v,u)\in E(G)\}$. That is, the query attributes are the set of attributes associated with every edge of $G$. 
We have 
\begin{scriptsize}
\begin{eqnarray} 
\keyw(G'_S, W_q) = \sum_{w_{vu}\in W_q} \frac{|V_{w_{vu}} \cap S |^2}{|S|}&&  \nonumber \\  
=  \sum_{w_{vu}\in W_q} \frac{(|V_{w_{vu}} \cap S |)(|V_{w_{vu}} \cap S |-1)}{|S|} +\sum_{w_{vu}\in W_q} \frac{|V_{w_{vu}} \cap S |}{|S|} 
\end{eqnarray} 
\end{scriptsize} 

\eat{
{
\scriptsize
\begin{eqnarray}
&&\keyw(H) = \sum_{w_{vu}\in W_q} \frac{|V_{w_{vu}} \cap V(H) |^2}{n}  \nonumber \\
&=&  \sum_{w_{vu}\in W_q} \frac{(|V_{w_{vu}} \cap V(H) |)(|V_{w_{vu}} \cap V(H) |-1)}{|V(H)|} +\sum_{w_{vu}\in W_q} \frac{|V_{w_{vu}} \cap V(H) |}{|V(H)|}  \nonumber \\
\end{eqnarray}
}
}

For every attribute $w_{vu}\in W_q$, exactly one of the following conditions holds: 

%\xin{
\begin{itemize} 
\item $u, v \in S$: In this case $(u,v) \in E(S)$. Clearly, $|V_{w_{vu}} \cap S| = 2$, so $(|V_{w_{vu}} \cap S|)(|V_{w_{vu}} \cap S |-1) = 2$. 

\item exactly one of $u, v$ belongs to $S$ and $(u,v) \in E \setminus E(S)$. In this case, $|V_{w_{vu}} \cap S|= 1$, so $(|V_{w_{vu}} \cap S|)(|V_{w_{vu}} \cap S |-1) = 0$. 

\item $u, v \not\in S$. In this case, clearly $(u,v) \not\in E(S)$ and $|V_{w_{vu}} \cap S | = 0$, so $(|V_{w_{vu}} \cap S |)(|V_{w_{vu}} \cap S |-1) = 0$. 
\end{itemize} 
%}

\eat{
\begin{itemize} 
\item $u, v \in S$: In this case $(u,v) \in E(S)$. Clearly, $|V_{w_{vu}} \cap V(H)| = 2$, so $(|V_{w_{vu}} \cap V(H) |)(|V_{w_{vu}} \cap V(H) |-1) = 2$. 

\item exactly one of $u, v$ belongs to $S$ and $(u,v) \in E \setminus E(S)$. In this case, $|V_{w_{vu}} \cap V(H) |= 1$, so $(|V_{w_{vu}} \cap V(H) |)(|V_{w_{vu}} \cap V(H) |-1) = 0$. 

\item $u, v \not\in S$. In this case, clearly $(u,v) \not\in E(S)$ and $|V_{w_{vu}} \cap V(H) | = 0$, so $(|V_{w_{vu}} \cap V(H) |)(|V_{w_{vu}} \cap V(H) |-1) = 0$. 
\end{itemize} 
}

\eat{ 
\begin{description}
\item[(a)$|V_{w_{vu}} \cap S |=2.$] This corresponds to  $v\in S$, $u\in S$, and $(v, u) \in E(S)$. This implies $(|V_{w_{vu}} \cap V(H) |)(|V_{w_{vu}} \cap V(H) |-1) = 2$. 
\item[(b)$|V_{w_{vu}} \cap S |=1.$] It corresponds to exactly one of the vertices $v, u$ being in $S$, and $(v, u) \in E$. Due to either $v\notin S$ or $u\notin S$, $(v, u) \notin E(S)$. Meanwhile, $(|V_{w_{vu}} \cap S|)(|V_{w_{vu}} \cap S |-1) = 0$. 
\item[(c)$|V_{w_{vu}} \cap S |=0.$] $u\notin S$, $v \notin S$, and $(v, u) \notin E(S)$. 
\end{description}
} 

Therefore,  
\begin{scriptsize} 
\begin{align} 
\sum_{w_{vu}\in W_q} \hspace*{-2ex} \frac{(|V_{w_{vu}} \cap S |)(|V_{w_{vu}} \cap S |-1)}{|S|} %&& \nonumber \\
= \hspace*{-2ex} \sum_{(v,u)\in E(S)} \hspace*{-2ex} \frac{2}{|S|}  = 2\frac{|E(S)|}{|S|} = 2\rho(G_S).   
\end{align}
\end{scriptsize}

%$= \sum_{v\in S}\sum_{(v,u)\in E(S)} \frac{1}{|S|}$\\
%$= \sum_{v\in S}\frac{\sum_{(v,u)\in E(S)} 1}{|S|} $\\
%$= \sum_{v\in S} \frac{\deg_{G_S}(v)}{|S|}$\\
%$= 2\rho(G_S)$\\

On the other hand, we have 
\begin{scriptsize}
\begin{align}
\sum_{w_{vu}\in W_q} \frac{|V_{w_{vu}} \cap S |}{|S|} %&&  \nonumber  \\
= \sum_{(v,u)\in E(S)} \frac{2}{|S|} + \sum_{v\in S, (v,u)\in E \setminus E(S)} \frac{1}{|S|} && \nonumber  \\
=\frac{\sum_{v\in S}\sum_{(v, u)\in E(S)} 1+  \sum_{v\in S}\sum_{(v,u)\in E(G)\setminus E(S)} 1} {|S|} %&& \nonumber  \\
=\frac{\sum_{v\in S} \deg_{G}(v)}{|S|}. 
\end{align} 
\end{scriptsize} 

Overall,  $\keyw(G'_S, W_q) = $ $2\rho(G_S) + \frac{\sum_{v\in S} \deg_{G}(v)}{|S|} = \chi(G_S)$. Next, we show that an instance of \WDalK is a YES-instance iff for the corresponding instance of \ktcp, has a weighted density above a threshold, w.r.t.  the query $Q=(V_q, W_q)$ where $V_q = \emptyset$ and $W_q = \{w_{vu}: (v,u)\in E\}$ and the parameter $d=0$.\footnote{Since $V_q = \emptyset$, we can set $d$ to any value; we choose to set it to the tightest value.} The hardness follows from this. 

\noindent 
$(\Leftarrow):$ Suppose $G$ is a YES-instance of \WDalK, i.e., there exists a subset $S^*\subset V$ such that for the induced subgraph $G_{S^*}$ of $G$, we have $\chi(G_{S^*}) \geq \alpha$. Then, the subgraph $G'_{S^*}=(S^*, E'(S^*))$ has $\keyw(G'_{S^*}, W_q) = \chi(G_{S^*}) \geq \alpha$. In addition, since $|S^*|\geq k$ and $G'$ is an $n$-clique, $G'_{S^*}$ is a $k$-clique, and hence a $k$-truss. For $V_q=\emptyset$, trivially $V_q\subseteq S^*$ and $G'_{S^*}$ satisfies the communication constraint on query distance. Thus, $G'_{S^*}$ is a \kdtruss with  $\keyw(G'_{S^*}, W_q) \geq \alpha$, showing $G'$ is a YES-instance of \ktcp.

\noindent 
$(\Rightarrow):$ Supose there exists a \kdtruss $G'_{S^*}=(S^*, E'(S^*))$, a subgraph of $G'$ induced by $S^*\subset V$, with  $\keyw(G'_{S^*}, W_q) \geq \alpha$. Then, we have $G_{S^*}=(S^*, E(S^*))$ and $\chi(G_{S^*}) = \keyw(G'_{S^*}, W_q) \geq \alpha$. Since $G'_{S^*}$ is a $k$-truss and $|S^*|\geq k$, we have $ \chi(G_{S^*}) $ $\geq \alpha$, showing $G$ is a YES-instance of \WDalK. \end{proof}
}

%% file: keyword.tex
We next investigate the properties of the attribute score 
function, in search of prospects for an approximation algorithm
for finding \atc. From the definition of attribute score function $\keyw(H, W_q)$, we can infer the following useful properties. 

\stitle{Positive influence of relevant attributes.} The more relevant attributes a community $H$ has, the higher the score $\keyw(H, W_q)$. E.g., consider the community $H_4$ and $W_q=\{ML\}$ in Figure \ref{fig.subcom} (d). If the additional attribute ``ML'' is added to the vertex $v_9$, then it can be verified that the score  $\keyw(H_4, \{ML\})$ will increase. We have: 

\begin{observation}\label{lemma.vinc}
Given a \ktc $H$ and a vertex $v\in H$, let a new input attribute $w \in W_q \setminus \doc(v)$ be added to $v$, and $H'$ denote the resulting community.  Then $\keyw(H', W_q) > \keyw(H, W_q)$. 
\end{observation}

In addition, we have the following easily verified observation. 

\begin{observation}\label{lemma.kwinc}
Given a \ktc $H$ and   query attribute sets $W_{q} \subseteq W_{q'}$, we have  $\keyw(H, W_q)\leq \keyw(H, W_{q'})$.
\end{observation}

\eat{
For a community $H$, more kinds of attributes $H$ has, the higher the score $\keyw(H, W_q)$.
} 

\note[Laks]{Most of the lemmas above are straightforward and are perhaps better called facts.} 
\note[Xin]{Agree.}

\stitle{Negative influence of irrelevant vertices.} Adding irrelevant vertices with no query attributes to a \ktc  will decrease  its attribute score. \xin{For example,  for $W_q = \{DB\}$, if we insert the  vertex $v_7$ with attribute $IR$ into the community $H_1$ in Figure \ref{fig.subcom}(b), it  decreases the score of the community w.r.t. the above query attribute \xin{$W_q=\{DB\}$} %$V_w \{DB\}$
, i.e.,  $\keyw(H_1\cup\{v_7\}, \{DB\}) < \keyw(H_1, \{DB\})$.} The following observation formalizes this property. 

\note[Laks]{There is no node $v_7$ in Figure \ref{fig.subcom}(a). In fact, currently it has only 5 nodes, making me wonder if the added node should be called $v_6$ instead. BTW, it's a good idea to number all nodes in the example graphs. Similar comment on Figure \ref{fig.subcom}(d). In fact, I don't follow these examples.} 

\note[Laks]{At some point, a serious spell-check pass must be made: there are lots of typos. I corrected some.} 

\begin{observation}\label{lemma.dec}
Given two \ktc's $H$ and $H'$ where $H\subset H'$, suppose $\forall v\in V(H') \setminus V(H)$ and $\forall w\in W_q$, $\doc(v)\cap V_w =\emptyset$. Then $\keyw(H', W_q) < \keyw(H, W_q)$.
\end{observation}

%****I CHANGED $\leq$ ABOVE TO $<$. CHECK!**** 

\stitle{Non-monotone property and majority attributes.} 
The  attribute score function is in general non-monotone w.r.t. the size of the community, even when vertices with query related attributes are added. For instance, for the community $H_1$ in Figure \ref{fig.subcom}(a), with $W_q = \{DB, IR\}$, $\keyw(H_1, W_q)= 4\cdot\frac{4}{4}=4$. Let us add vertex $v_7$ with attribute $IR$ into $H_1$ and represent the resulting graph as $H_5$, then  $\keyw(H_5, W_q)= 4\cdot\frac{4}{5} + 1\cdot\frac{1}{5}=\frac{17}{5}< \keyw(H_1, W_q)$.   
%****I CHANGED OCCURRENCES OF $H_4$ ABOVE TO $H_1$. CHECK!**** 
If vertex $v_7$ has attribute DB instead of IR, then it is easy to verify that the attribute score of the resulting graph w.r.t. $W_q$ is strictly higher than $4$. Thus, $\keyw(.,.)$ is neither monotone nor anti-monotone. 
\note[Laks]{Waiting for the example.} \note[Xin]{Done.} This behavior raises  challenges for finding \ktc with the maximum attibute score. Based on the above examples, we have the following observation.  
\note[Laks]{I'd like to see the examples first. BTW, this section needs to be reorganized: lemmas $\rightarrow$ simple observations or facts; they can all be combined into a discussion para where we can list the facts one by one and illustrate them with examples or provide simple arguments.} 

\begin{observation}\label{lemma.nomonotone}
There exist \ktc's $H$ and $H'$ with $V(H') = V(H)\cup \{v\}$, and $\doc(v) \cap W_q \neq \emptyset$, such that $\keyw(H', W_q) < \keyw(H, W_q)$, and there exist \ktc's $H$ and $H'$ with $V(H') = V(H)\cup \{v\}$, and $\doc(v) \cap W_q \neq \emptyset$, for which $\keyw(H', W_q) > \keyw(H, W_q)$. 
\eat{
Given two \ktc's $H$ and $H'$ where $V(H') = V(H)\cup \{v\}$, even if $\doc(v) \cap W_q \neq \emptyset$, then either $\keyw(H', W_q) \leq \keyw(H, W_q)$ or $\keyw(H', W_q) \geq \keyw(H, W_q)$  can hold.}
\end{observation}

\note[Laks]{Not clear. Perhaps the examples need to be better connected with these ``lemmas''. Also, I suspect there are some typos in the handdrawn examples.} 

The key difference between the two examples above is that DB is a ``majority attribute'' in $H_1$, a notion we formalize next. 
\eat{ 
If such inserted node $v$ carries with one majority attibute in $H$, the new community $H' = H\cup \{v\}$ will have a greater attibute score. For example, consider the community $H_1$ in Figure \ref{fig.subcom} (a) with $W_q=\{DB, IR\}$, the attribute ``DB'' is a certain majority attribute, since every node has the attribute ``DB''. Assume that the vertex $v_7$ carries with the attribute ``DB'', not with ``IR'' any more, then a new graph as $H_1$ with the insertion of $v_7$ will increase the attribute score.
} 
Formally, given a community $H$ and query $W_q$, we say that a set of attributes $X$ includes majority attributes of $H$, and $\theta(H, W_q\cap X) = $ $\sum_{w\in W_q\cap X} \theta(H, w)$ $\geq \frac{\keyw(H, W_q)}{2|V(H)|}$.  Recall that $\theta(H, w)$ is the fraction of vertices of $H$ containing the attribute $w$. We have: 
%if $\theta(H, W_q\cap X) = $ $\sum_{w\in W_q\cap X} \theta(H, w)$ $\geq \frac{\keyw(H, W_q)}{2|V(H)|}$.  Recall that $\theta(H, w)$ is the fraction of vertices of $H$ containing the attribute $w$. We have: 

\eat{
\xin{This principle will help designing the bottom-up algorithms, by iteratively adding the relevant attibutes to increase attibute score. } 
} 

\note[Laks]{What is majority keyword/attribute?} 
\note[Xin]{For each attribute, the attribute score is voted by all nodes. If one attribute presents on most of nodes, it is the majority attribute. Thus, a new node with majority attribute is added into community, the attribute score of this community will increase, i.e., this node is not violated by most of members in this community. }

% \eat{
% \begin{lemma}\label{lemma.majority}
% Let $H$ be a \ktc of a graph $G$ and $w\in W_q$ be a majority attribute of $H$. Suppose there is a vertex $v \not\in V(H)$ such that $\doc(v)\cap W_q = \{w\}$ and that adding $v$ to $H$ results in a \ktc $H'$ of $G$. Then $\keyw(H', W_q) > \keyw(H, W_q)$ holds. 
% \eat{ 
% Given two \ktcs $H$ and $H'$ where $V(H') = V(H)\cup \{v\}$, there exist an attibute $w\in W_q$ as the majority attibute in $H$, i.e., $\theta(H, w)\geq \frac{\keyw(H, W_q)}{2|V(H)|}$ and  $\doc(v) \cap W_q =\{w\}$, then $\keyw(H', W_q) > \keyw(H, W_q)$ holds.
% } 
% \end{lemma}

% \begin{proof}
% Suppose $W_q= \{w_1, ..., w_l\}$ and w.l.o.g., let the majority attibute be  $w_1$. 
% %In addition, we denote by $|V(H)|= b >0$ due to $\keyw(H, W_q) >0$ by Definition \ref{def.kwf}, and 
% Let $|V(H)|= b$, and for each attribute $w_i\in W_q$, let $|V(H)\cap V_{w_i}| = b_i$.   Since $w_1$ is the majority attibute of $H$,  $\theta(H, w_1)= \frac{b_1}{b} \geq \frac{\keyw(H, W_q)}{2b}$, so we have $2b_1 \geq \keyw(H, W_q)$.  

% We have $\keyw(H, W_q) = \sum_{i=1}^{l}\frac{|V(H) \cap V_{w_1}|^2}{|V(H)|}$ $=\sum_{i=1}^{l} \frac{b_i^2}{b}$, and $\keyw(H', W_q)$ $ = \frac{(b_1+1)^2}{b+1}$ $ +$ $ \sum_{i=2}^{l}\frac{b_i^2}{b+1}$. As a result, 
% $\keyw(H', W_q) - \keyw(H, W_q)$ 
% $= \frac{b\cdot(2b_1+1) - \sum_{i=1}^{l}b_i^2}{b(b+1)}$
% $\geq \frac{b\cdot \keyw(H, W_q) +b - b\cdot \keyw(H, W_q)}{b(b+1)}$ $= \frac{1}{b+1}$ $>0$.
% \end{proof}
% }

\begin{lemma}\label{lemma.majority}
Let $H$ be a \ktc of a graph $G$. Suppose there is a vertex $v \not\in V(H)$ such that the set of attributes $W_q \cap \doc(v)$ includes the majority attributes of $H$ and that adding $v$ to $H$ results in a \ktc $H'$ of $G$. Then $\keyw(H', W_q) > \keyw(H, W_q)$ holds. 
\end{lemma}

\begin{proof}
\add{The proof can be found in \cite{ArxivATC}.}
\cut{Suppose $W_q= \{w_1, ..., w_l\}$ and w.l.o.g., let $W_q\cap\doc(v)$ $= \{w_1, ..., w_r\}$, where $1\leq r\leq l$. 
%In addition, we denote by $|V(H)|= b >0$ due to $\keyw(H, W_q) >0$ by Definition \ref{def.kwf}, and 
Let $|V(H)|= b$, and for each attribute $w_i\in W_q$, let $|V(H)\cap V_{w_i}| = b_i$.   Since $W_q\cap\doc(v)$ includes the majority attributes of $H$,  $\theta(H, W_q\cap \doc(v))= \frac{\sum_{i=1}^{r} b_i}{b} \geq \frac{\keyw(H, W_q)}{2b}$, so we have $\sum_{i=1}^{r} 2b_i \geq \keyw(H, W_q)$.  

We have $\keyw(H, W_q) = \sum_{i=1}^{l}\frac{|V(H) \cap V_{w_i}|^2}{|V(H)|}$ $=\sum_{i=1}^{l} \frac{b_i^2}{b}$, and $\keyw(H', W_q)$ $ = \sum_{i=1}^{r}\frac{(b_i+1)^2}{b+1}$ $ +$ $ \sum_{i=r+1}^{l}\frac{b_i^2}{b+1}$. As a result, 
$\keyw(H', W_q) - \keyw(H, W_q)$ 
$= \frac{b\cdot\sum_{i=1}^{r}(2b_i+1) - \sum_{i=1}^{l}b_i^2}{b(b+1)}$
$\geq \frac{b\cdot \keyw(H, W_q) +rb - b\cdot \keyw(H, W_q)}{b(b+1)}$ $= \frac{r}{b+1}$ $>0$.}
\end{proof}

This lemma will be helpful in designing bottom-up algorithms, by iteratively adding vertices with majority attributes to increase attribute score.

\note[Laks]{1. Why is this lemma interesting or useful? 2. There is a typo: in the last line in the numerator, the sum should be from $i=2$ to $l$ instead of $i=1$ to $l$, but the conclusion is still valid.}
\note[Xin]{The sum from $i=1$ to $l$ is exact, if the detailed computation are shown up.} 

%%%%%%%%%%%%%%%%%%%%%% 
\eat{ 
****ONE ISSUE WITH THE EXAMPLES BELOW IS THAT FOR THE SECOND EXAMPLE, WHERE $v^* = q_2$ is added, the trussness of $G_2$ decreasesIS ADDED, THE TRUSSNESS OF $G_2$ DECREASES. CAN YOU GIVE AN EXAMPLE WHERE THE TRUSSNESS REMAINS THE SAME?**** 

\xin{****AS FAR AS I CAN SEE, IN THE SECOND EXAMPLE, THE TRUSSNESS OF $G_2$ IS 3, WHICH KEEPS THE SAME WITH VERTEX INSERTION OF $v^* = q_2$. MOREOVER, WE HAVE SET K=2, I.E., HERE WE ONLY CONSIDER THE KEYWORD SCORE FUNCTION, WITHOUT CONSIDERING THE STRUCTURE SERIOUSLY. **** }

****DID YOU DEFINE SUBMODULARITY AND SUPERMODULARITY IN THE PAPER?**** 

\xin{**** TAKE A LOOK AT THE FOLLOWING DEFINITIONS. **** }
} 
%%%%%%%%%%%%%%%%%%%%%% 

\stitle{Non-submodularity and Non-supermodularity.} \label{sec.modular}
A set function $g: 2^U \rightarrow \mathbb{R}^{\ge 0}$ is said to be submodular provided for all sets $S\subset T\subset U$ and element $x\in U\setminus T$, $g(T\cup\{x\}) - g(T) \le g(S\cup\{x\}) - g(S)$, i.e., the marginal gain of an element has the so-called ``diminishing returns'' property. The function $g(.)$ is said to be supermodular if $-g(.)$ is submodular. Optimization problems over submodular functions lend themselves to efficient approximation. We thus study whether our attribute score function $f(.,.)$ is submodular w.r.t. its first argument, viz., set of vertices. 

Consider the graph $G$ in Figure \ref{fig.community} and query $W_q=\{DB, DM\}$ with $k=2$. Let the induced subgraphs of $G$ by the vertex sets $S_1= \{q_1, v_4\}$ and $S_2= \{q_1, v_4, v_5\}$ respectively be denoted $G_1$ and $G_2$;   $G_1 \subseteq G_2$. Let $v^*$ be a vertex not in $G_2$. Let us compare the marginal gains  $\keyw(G_1\cup \{v^*\}, W_q)-\keyw(G_1, W_q)$ and  $\keyw(G_2\cup \{v^*\}, W_q)-\kw(G_2, W_q)$, from adding the new vertex $v^*$ to $G_1$ and $G_2$. 
Suppose $v^*= v_6$ with attribute ``DB'', then we have $\keyw(G_2\cup \{v_6\}, W_q)-\keyw(G_2, W_q) = (4+1/4)-(3+1/3)=11/12 > \keyw(G_1\cup \{v_6\}, W_q)-\keyw(G_1, W_q)=(3+1/3)-(2+1/2)=5/6$, violating submodularity of the attribute score function $\keyw(., .)$. On the other hand, suppose $v^*= q_2$ with attributes ``DB'' and ``DM''. Then we have $\keyw(G_2\cup \{q_2\}, W_q)-\keyw(G_2, W_q) = (4+1) - (3+1/3) =5/3 < \keyw(G_1\cup \{q_2\}, W_q)-\keyw(G_1, W_q)=(3+4/3)-(2+1/2)=11/6$, which violates supermodularity. %Given that the attribute score function is neither submodular nor supermodular, we infer that the prospects for an efficient approximation algorithm are not promising. 
We just proved:
\begin{lemma}\label{lemma.nonmodular}
The attribute score function $\keyw(H, W_q)$ is neither submodular or supermodular.
\end{lemma}
In view of this result, we infer that the prospects for an efficient approximation algorithm are not promising.

%% file: basic.tex
In this section, we develop a greedy algorithmic framework
for finding a \atc. It leverages the notions of attribute score contribution and attribute marginal gain that we define. Our algorithm first finds a \kdtruss, and then iteratively removes vertices with smallest attribute score contribution. Then, we analyze the time and space complexity of our algorithm. 
We also propose a more efficient algorithm with better quality, based on attribute marginal gain and bulk deletion.

%WedevelopagreedyalgorithmicfreameworktofindanATC containing given query nodes w.r.t. given query attributes. It first finds a (k, d)-truss, and then iteratively removes nodes with smallest attribute score contribution. For improving the efficiency and quality, we design a revised attribute marginal gain function and a bulk removal strategy for cutting down the number of iterations

%, and then discuss how to generalize the proposed algorithm for \ktcrp. 
\note[Laks]{This is the first time in the paper when ``online'' is introduced. Are we really searching online? In what sense?}
\note[Xin]{This algorithm will not use index(off-line reuslts), i.e., it search the results from scratch.}
%Finally, we describe procedures for an efficient algorithm implementation and analyze the time and space complexity of our algorithm.

\note[xin]{Remark.  Without any specific notes, all queries consider in this papaer are under the parameter stting of $k=4$ and $d=2$.}

\subsection{Basic Algorithm} \label{sec.basic}
We begin with attribute score contribution. Given a subgraph $H\subset G$, a vertex $v\in V(H)$, and attribute query $W_q$, let us examine the change to the score $f(H,W_q)$ from dropping $v$. 

\begin{scriptsize}
\begin{eqnarray} 
 &&\keyw(H-\{v\}, W_q)  \nonumber \\ 
 &&=\sum_{w \in W_q} \frac{|V_w \cap (V(H)-\{v\})|^2}{|V(H)|-1}  \nonumber \\  
 &&= \sum_{w \in W_q-\doc(v)} \frac{|V_w \cap V(H)|^2} {|V(H)|-1} +\sum_{w\in W_q \cap \doc(v)} \frac{(|V_w \cap V(H)|-1)^2}{|V(H)|-1} \nonumber \\  
&&=  \new{ \sum_{w \in W_q} \frac{|V_w \cap V(H)|^2 } {|V(H)|-1} - \sum_{w\in \doc(v) \cap W_q} \frac{|V_w \cap V(H)|^2}{|V(H)|-1} +} \nonumber\\ && \new{+\sum_{w\in W_q \cap \doc(v)} \frac{(|V_w \cap V(H)|-1)^2}{|V(H)|-1} } \nonumber \\ 
 &&= \sum_{w \in W_q} \frac{|V_w \cap V(H)|^2 } {|V(H)|-1} -\sum_{w\in W_q \cap \doc(v)} \frac{2| V_{w} \cap V(H)|-1}{|V(H)|-1} \nonumber \\ 
&&=  \frac{\keyw(H, W_q) \cdot |V(H)| } {|V(H)|-1} -\frac{\sum_{w\in W_q \cap \doc(v)} (2| V_{w} \cap V(H)|-1)}{|V(H)|-1} \nonumber 
\end{eqnarray} 
\end{scriptsize} 

\eat{ 
As we can see, for a vertex $v$ in $H$, the smaller value of $\sum_{w\in W_q \cap \doc(v)} (2| V_{w} \cap V(H)|-1)$ is, the larger attribute score of $\keyw(H-\{v\}, W_q)$ is. Thus, we define the attribute score contribution as follow. 
}

The second term represents the drop in the attribute score of $H$ from removing $v$. We would like to remove vertices with the least drop in score. This motivates the following. 

\begin{definition}[Attribute Score Contribution]\label{def.nas}
Given a graph $H$ and attribute query $W_q$, the attribute score contribution of a vertex $v\in V(H)$ is defined as $\keyw_{H}(v, W_q)$ $= \sum_{w\in W_q \cap \doc(v)} $ $2| V_{w} \cap V(H)|-1$. 
\end{definition}

 \LL{The intuition behind dropping a vertex $v$ from $H$ is as follows. Since $\keyw(H, W_q)$ is non-monotone (Section \ref{sec.scorepro}), the updated score from dropping $v$ from $H$ may increase or decrease, so we check if $\keyw(H-v,W_q) > \keyw(H,W)$.}

\stitle{Algorithm overview.} Our first greedy algorithm, called \basic, has three steps. First, it finds the maximal \kdtruss of $G$ as a candidate. Second, it iteratively removes vertices with smallest attribute score contribution from the candidate graph, and maintains the remaining graph as a  \kdtruss, until no longer possible. Finally, it returns a \kdtruss with the maximum attribute score among all  generated candidate graphs as the answer.

The details of the algorithm follow. First, we find the maximal \kdtruss of $G$ as $G_0$. Based on the given $d$, we compute a set of vertices $S$ having query distance no greater than $d$, i.e., $S_0 = \{u: \dist_G(u, Q)\leq d\}$. \eat{ 
Thus, we can have the induced subgraph of $G$ by $S_0$, denoted as $G_0$, in which each vertex in $G_0$ has distance no greater than $d$ from query nodes, i.e., $\dist_{G_0}(G_0, V_Q) \leq d$.} 
Let $G_0\subset G$ be the subgraph of $G$ induced by $S_0$. Since $G_0$ may contain edges with support $< (k-2)$, we invoke the following steps to prune $G_0$ into a \kdtruss. 
\eat{In $G_0$, the edge may has less than $k-2$ triangles and violate the condition of $k$-truss, we invoke the following \kdtruss maintenance procedure to shrink $G_0$ into a \kdtruss. 
} 

 %which satisfies the distance constraint. Next, for a given $k$, we iteratively remove the edge has support less than $k-2$ from $G_S$, by the Definition \ref{def.ktruss} of $k$-truss. However, since the edge removal leads to the distance increasing in $G_S$, the new graph $G_S$ may not satisfy the distance constraint any more. Therefore, we iteratively invoke the following procedure of \stitle{maintaing a close $k$-truss} until the resulting graph $G_S$ is a connected $k$-truss and $\dist_{G_S}(G_S, V_Q)\geq d$.

\stitle{\kdtruss maintenance}: repeat until no longer possible: 

(i) \emph{\ktruss}: remove edges contained in $< (k-2)$ triangles; \\ 
\indent (ii) \emph{query distance}: remove vertices with query distance $> d$, \indent and their incident edges; 

\eat{
\stitle{(a) edge removal for truss violation,} is to remove the edges having less than $k-2$ traingles from $G$. 

\stitle{(b) vertex removal for distance violation,} is that for the vertices having  distance greater than $d$ from query nodes, we remove these vertices and its incident edges from $G$. 
} 

Notice that the two steps above can trigger each other: removing edges can increase query distance and removing vertices can reduce edge support.
\eat{ 
However, since the edges removal and vertices removal can affect each other, leading to more vertices and edges violating $k$-truss and distance constraints. Therefore, to obtain a \kdtruss, we need to iteratively invoke above two steps until the graph has no vertices and edges to be removed.}  %Obviously, if $G_S$ contains all query nodes $V_q$ and have $\keyw(G)>0$, then $G$ is a qualified candidate of attibute truss community and the maximal one. 
In the following, we start from the maximal \kdtruss $G_l$ where $l=0$, and find a \kdtruss with large attribute score by deleting a vertex with the smallest attribute score contribution.

\stitle{Finding a \kdtruss with large attribute score.} $G_0$ is our first candidate answer. In general, given $G_l$, we find a vertex $v\in V(G_l) \setminus V_q$ with the smallest attribute score contribution and remove it from $G_l$. Notice that $v$ cannot be one of the query vertices. The removal may violate the \kdtruss constraint so we invoke the \kdtruss maintenance procedure above to find the next candidate answer. We repeat this procedure until $G_l$ is not a \kdtruss any more. 
\eat{ 
Here, we design a greedy strategy to remove vertex $u^{*} \in V(G_l)-V_q$ with the smallest attribute score contribution $\keyw_{G_l}(u^{*}, W_q)$ from $G_l$. Due to the deletion of $u^{*}$ and its incident edges, the remaining graph may not satisfy the constraints of \kdtruss again.  We invokes the procedure of \kdtruss maintenance to keep the remaining graph $G_l$ as \kdtruss. Then, we assign the updated graph as a new $G_l$, and continue the above procedure until $G_l$ is not a \kdtruss any more. 
} 
Finally, the candidate answer with the maximum attribute score generated during this process is returned as the final answer, i.e., $ \arg\max_{G'\in\{G_0, ..., G_{l-1}\}}{\keyw(G', W_q)}$. The detailed description is presented in Algorithm \ref{algo:basic}.

\note[Laks]{Is the function $kw(.,)$ submodular or supermodular? That is, given a fixed set of attributes $W_q$ and a graph $G$, is $kw(W_q, H)$ submodular or supermodular in $H$ where $H$ is any induced subgraph of $G$?} 
\note[Xin]{Probably not. As explained in gtalk.}

%\rho it has shown to be very small in many real-world graphs 

\begin{algorithm}[t]
\small
\caption{\basic($G$, $Q$)} \label{algo:basic}
\textbf{Input:} A graph $G=(V, E)$, a query $Q=(V_q, W_q)$, numbers $k$ and $d$.\\
\textbf{Output:} A \kdtruss $H$ with the maximum $\keyw(H, W_q)$.\\
\vskip -0.6cm
\
\begin{algorithmic}[1]

\STATE  Find a set of vertices $S_0$ having the query distance $\leq d$, i.e., $S_0 = \{u: \dist_G(u, Q)\leq d\}$.

\STATE  Let $G_0$ be the induced subgraph of $S$, i.e., $G_0= (S_0, E(S_0))$, where $E(S_0)= \{(v,u): v, u\in S_0, (v, u)\in E \}$.

\STATE  Maintain $G_0$ as a \kdtruss.

\STATE  Let  $l\leftarrow 0$;

\STATE	\textbf{while}  $\con_{G_{l}}(Q)=$ \textbf{true} \textbf{do}

\STATE  \hspace{0.3cm}  Compute the attribute score of $\keyw(G_l, W_q)$;

\STATE  \hspace{0.3cm}  Compute $\keyw_{G_{l}}(u, W_q) =\sum_{w\in W_q \cap \doc(u)} 2|V(H)\cap V_w|-1$, $\forall u\in G_{l}$;% the shortest distances from each node $q\in Q$ to each node $u\in  G_{l}$;

\STATE  \hspace{0.3cm}  $u^{*} \leftarrow \arg\min_{u\in V(G_{l}) - V_q} \keyw_{G_{l}}(u, W_q)$;

\STATE  \hspace{0.3cm}  Delete $u^{*}$ and its incident edges from $G_{l}$;

\STATE  \hspace{0.3cm}  Maintain $G_l$ as a \kdtruss.

\STATE  \hspace{0.3cm}  $G_{l+1} \leftarrow G_{l}$; $l \leftarrow l+1$;

\STATE  $H \leftarrow \arg\max_{G'\in\{G_0, ..., G_{l-1}\}}{\keyw(G', W_q)}$;

\end{algorithmic}
\end{algorithm}

\begin{example}
We apply Algorithm \ref{algo:basic} on the graph $G$ in Figure \ref{fig.community} with query \ $Q=(\{q_1\},\{DB, DM\})$, for $k=4$ and $d=2$. First, the algorithm finds the \kdtruss $G_0$ as the subgraph $H$ shown in Figure \ref{fig.community}. Next, we select vertex $v_7$ with the minimum attribute score contribution $\keyw_{G_0}(v_7, W_q)= 0$ and remove it from $G_0$. Indeed it contains neither of the query attributes. Finally, the algorithm finds the \ktc $H_2$ with the maximum attribute score in Figure \ref{fig.subcom}(b), which, for this example, is the optimal solution. 
\end{example}

\eat{ 
****IS $H_2$ THE OPTIMAL ANSWER FOR THIS EXAMPLE?**** 

\xin{**** Yes, it is optimal and shown in example now. ****}
}

\subsection{Complexity Analysis}

\cut{Let $n=|V(G)|$ and $m=|E(G)|$, and let $d_{max}$ be the maximum vertex degree in $G$.} 
In each iteration $i$ of Algorithm \ref{algo:basic}, we delete at least one vertex and its incident edges from $G_i$. Clearly, the number of removed edges is no less than $k-1$, and so the total number of iterations is $t \leq \min\{n-k, m/(k-1)\}$. 
\add{We have: }
\cut{We have the following result on the time and space complexity of Algorithm \ref{algo:basic}. 
We note that we do not need to keep all candidate \ktcs in the implementation, but merely maintain a removal record of the vertices/edges in each iteration.}

\begin{theorem}\label{theorem.framework}
Algorithm \ref{algo:basic}  takes $O(m\rho$ $+t(|W_q| n+|V_q|m))$ time and  $O(m+|\doc(V)|)$ space, where $t\in O(\min\{n, m/k\})$, and $\rho$ is the arboricity of graph $G$ with $\rho \leq \min\{d_{max}, \sqrt{m}\}$.
\end{theorem}

\cut{
\begin{proofskecth}
The time cost of Algorithm \ref{algo:basic} mainly comes from three key parts: query distance computation, $k$-truss maintenance, and attribute score computation.

For query distance computation, finding the set of vertices $S$ within query distance $d$ from $V_q$ can be done by computing the shortest distances using a BFS traversal starting from each query node $q\in V_Q$, which takes $O(|V_q| m)$ time. Since the algorithm runs in $t$ iterations, the total time cost of this step is $O(t |V_q| m)$.

Second, consider the cost of $k$-truss identification and maintenance. Finding and maintaining a series of $k$-truss graphs $\{G_0, $ $..., $ $G_{l-1}\}$ in each iteration takes $O(\rho\cdot m)$ time in all, where $\rho$ is the arboricity of graph $G_0$. It has been shown that $\rho \leq \min\{d_{max}, \sqrt{m}\}$ \cite{ChibaN85}.

Third, consider the cost of computing attribute score contribution. In each iteration, the computation of attribute score contribution for every vertex takes time $O(\sum_{v\in V(G)}$ $ \min\{\doc(v), |W_q|\}) =$ $ O(\min $ $\{|\doc(V)|, $ $|W_q|\cdot n\}) $ $\subseteq $ $ O(|W_q|\cdot n) $. Thus, the  total cost of attribute score computation is $O(t |W_q| n)$.

%in the candidate \ktc $G_l$,  the total compuation of $|V(G_l)\cap V_w|$ for every keyword $w\in W_q$ takes $O(\sum_{v\in V(G_l)} doc(v) + |W_q|) \subseteq O(|\doc(V)|+ |W_q|)$, in each iteration, 

Therefore, the overall time complexity of Algorithm \ref{algo:basic} is $O(m\rho $ $+t$ $(|W_q|n $ $+|V_q|m)$ $)$. 

Next, we analyze the space complexity. For graphs $\{G_0, ..., G_l\}$,  we  record the sequence of removed edges from $G_0$: attaching a corresponding  label to graph $G_i$ at each iteration $i$, takes $O(m)$ space in all.  For each vertex $v\in G_i$, we keep $\dist(v, Q)$, 
%instead of all query distances $\dist(v, q)$ for $q\in Q$, 
which takes $O(n)$ space. Hence, the space complexity is $O(m+n+|\doc(V)|)$, which is $O(m+|\doc(V)|)$, due to the assumption $n\leq m$.  
\end{proofskecth}
}
\add{
\begin{proofskecth}
The time cost of Algorithm \ref{algo:basic} mainly comes from three key parts: (1) query distance computation takes $O(t |V_q| m)$ time by BFS traversals; (2) $k$-truss maintenance takes $O(\rho\cdot m)$ time; (3) attribute score computation takes $ O(t |W_q|\cdot n) $ time. A complete proof is
available in \cite{ArxivATC}. It was shown in \cite{ChibaN85} that $\rho \leq \min\{d_{max}, \sqrt{m}\}$. \qed 	
\end{proofskecth}
}

\note[Laks]{Work done per iteration is not uniform in general, given that in an iteration, we count recursive deletions of nodes and/or edges. It would be better to do a tighter analysis of the complexity.} 
\note[Xin]{Will do.}

\subsection{An improved greedy algorithm}
The greedy removal strategy of \basic is simple, but suffers from the following limitations on  quality and efficiency. Firstly, the attribute score contribution myopically considers the removal vertex $v$ only, and ignores its  impact on triggering removal of other vertices, due to violation of $k$-truss or distance constraints. If these vertices have many query attributes, it can severely limit the effectiveness of the algorithm. Thus, we need to look ahead the effect of each removal vertex, and then decide which ones are better to be deleted. Secondly, \basic removes only one vertex from the graph in each step, which leads to a large number of iterations, making the algorithm inefficient. 
%It finds the initial $d$-close-$k$-truss $G_0$ using the truss decompostion algorithm for computing $k$-truss. If we can keep the structural information of $k$-truss in record, then we can efficiently avoid   
%In each iteration, the algorithm consistently processes the removal of single node in worst cases, which costs expensive. 
%revised attribute marginal gain function and a bulk removal strategy.

In this section, we propose an improved greedy algorithm called \bulk, which is outlined in Algorithm \ref{algo:bulk}. \bulk uses the notion of attribute marginal gain and a bulk removal strategy.

\stitle{Attribute Marginal Gain.} We begin with a definition. 

\begin{definition}[Attribute Marginal Gain]
Given a graph $H$, attribute query $W_q$, and a vertex $v\in V(H)$, the attribute marginal gain is defined as $\gain_{H}(v, W_q)= \keyw(H, W_q) - \keyw(H-S_{H}(v), W_q)$, where $S_{H}(v)\subset V(H)$ is $v$ together with the set of vertices that violate  \kdtruss after the removal of $v$ from $H$. 
\end{definition}
%the community $H_4$ in Figure \ref{fig.subcom}(d)
Notice that by definition, $v \in S_H(v)$. 
For example, consider the graph $G$ in Figure \ref{fig.community} and the query $Q=(\{q_1\}, \{ML\})$, with $k=3$ and $d=2$. The vertex $v_9$ has no attribute ``ML'', and the attribute score contribution is $\keyw_{G}(v_9, W_q) =0$ by Definition \ref{def.nas}, indicating no attribute score contribution by vertex $v_9$. However, the fact is that $v_9$ is an important bridge for connections among the vertices $q_1$, $v_8$, and $v_{10}$ with attribute ``ML". The deletion of $v_9$ will thus lead to the deletion of $v_8$ and $v_{10}$, due to the 3-truss constraint. Thus, $S_{G}(v_9)=\{v_8, v_9, v_{10}\}$. The marginal gain of $v_9$ is $\gain_{G}(v_9, W_q)= \keyw(G, W_q) - \keyw(G-S_{G}(v_9), W_q)= \frac{3}{4}-\frac{1}{9}>0$.
%$ = -\frac{23}{36}$ 
This shows that the deletion of $v_9$ from $G$ decreases the attribute score. It illustrates that attribute marginal gain can more accurately estimate the effectiveness of vertex deletion than score attribute contribution, by naturally incorporating look-ahead. %On the other hand, the deletion of other vertices without attribute ``ML'' except for $v_9$, will increase the attribute score. 

One concern is that $\gain_{H}(v, W_q)$ needs the exact computation of $S_H(v)$, which has to simulate the deletion of $v$ from $H$ by invoking \kdtruss maintenance, which is expensive. An important observation is that if vertex $v$ is to be deleted, its neighbors $u \in N(v)$ with degree $k-1$ will also be deleted, to maintain \ktruss. %In addition, this phenomenon will be recurively propagating during this deletion process. 
Let $P_H(v)$ be the set of $v's$ 1-hop neighbors with  degree  $k-1$ in $H$, i.e., $P_H(v)=\{u\in N(v): $ $\deg_{H}(u)=k-1\}$. We propose a local attribute marginal gain, viz., $\hat{\gain_{H}}(v, W_q)$ $= $ $\keyw(H, W_q)$ $ - $ $\keyw(H-P_{H}(v), W_q)$, to approximate $\gain_{H}(v, W_q)$. Continuing with the above example, in graph $G$, for deleting vertex $v_9$, note that  $\deg(v_8)=\deg(v_{10})=2 = k-1$, so we have $P_{G}(v_9)$ $=\{v_8, v_9, v_{10}\}$, which coincides with $S_{G}(v_9)$. In general, $\hat{\gain_{H}}(v, W_q)$ serves as a good approximation to $\gain_{H}(v, W_q)$ and can be computed more efficiently. 

\stitle{Bulk Deletion.} The second idea incorporated in \bulk is bulk deletion. The idea is that instead of removing one vertex with the smallest attribute marginal gain, we remove a small percentage of vertices from the current candidate graph that have the smallest attribute marginal gain. More precisely, let $G_i$ be the current candidate graph and let $\epsilon > 0$. We identify the set of vertices $S$ such that $|S|= \frac{\epsilon}{1+\epsilon} |V(G_i)|$ and the vertices in $S$ have the smallest attribute marginal gain, and remove $S$ from $G_i$, instead of removing a vertex at a time. 
\eat{ 
Next, we propose a strategy of bulk deletion. In each iteration $i$, the algorithm using bulk deletion is to remove a set of vertices $S$ with the smallest marginal gains from graph $G_i$, instead of single one vertex. The size of $S$ has $|S|= \frac{\varepsilon}{1+\varepsilon} |V(G_i)|$, where $\varepsilon >0$.
} 
Notice that the resulting \atc $G_{i+1}$ has size $|V(G_{i+1})| $ $\leq $ $\frac{1}{1+\epsilon} |V(G_i)|$ after the deletion of $S$. We can safely terminate the algorithm once the size of $G_i$ drops below $k$ vertices and return the best \atc obtained so far, due to the constraint of \ktruss.  Thus, it follows that the number of iterations $t$ drops from $O(\min\{n, m/k\})$ to %$t \in $ 
$O(\log_{1+\epsilon}{\frac{n}{k}})$.

\begin{algorithm}[t]
\small
\caption{\bulk($G$, $Q$)} \label{algo:bulk}
\textbf{Input:} A graph $G=(V, E)$, a query $Q=(V_q, W_q)$, numbers $k$ and $d$, parameter $\varepsilon$.\\
\textbf{Output:} A \kdtruss $H$ with the maximum $\keyw(H, W_q)$.\\
\vskip -0.6cm
\
\begin{algorithmic}[1]

\STATE  Find the maximal \kdtruss $G_0$.

\STATE  Let  $l\leftarrow 0$;

\STATE	\textbf{while}  $\con_{G_{l}}(Q)=$ \textbf{true} \textbf{do}

\STATE  \hspace{0.3cm}  Find a set of vertices $S$ of the smallest $\hat{\gain_{G_l}}(v, W_q)$ with the size of $|S|= \frac{\varepsilon}{1+\varepsilon} $ $|V(G_i)|$;

\STATE  \hspace{0.3cm}  Delete $S$ and their incident edges from $G_{l}$;

\STATE  \hspace{0.3cm}  Maintain the \kdtruss of $G_{l}$;

\STATE  \hspace{0.3cm}  $G_{l+1} \leftarrow G_{l}$; $l \leftarrow l+1$;

\STATE  $H \leftarrow \arg\max_{G'\in\{G_0, ..., G_{l-1}\}}{\keyw(G', W_q)}$;

\end{algorithmic}
\end{algorithm}

%% file: index.tex
\eat{ 
As stated before, Algoirthm \ref{algo:basic} for finding \atc againts a graph takes the polynomial time complexity. 
\note[Laks]{We prove that KTC-1 is NP-hard. So what you mean is that the basic greedy algorithm we propose in the previous section, which has no guarantee, is polynomial time.} \note[Xin]{Yes.}
However, when the graph $G$ is large in size or the query $Q$ has multiples nodes and attributes, it is still very costly to process \atc queries. In order to speed up the query processing, in this section, we introduce an index-based query processing approach. We design a novel attributed-truss index (\ati) that can keep known graph structure and attribute information, and then propose an corresponding efficient algorithms for quickly detect \ktcs based on \ati. %In brief, XXXX.
} 

While the \bulk algorithm based on the framework of Algorithm~\ref{algo:basic} has polynomial time complexity, when the graph $G$ is large and the query $Q$ has many attributes, finding \ktcs entails several \ktc queries, which can be expensive. To help efficient processing of \ktc queries, we propose a novel index called attributed-truss index (\ati). It maintains known graph structure and attribute information. 

\subsection{Attributed Truss Index}
The \ati consists of three components: \emph{structural trussness}, \emph{attribute trussness}, and \emph{inverted attibute index}.

%Since our \ktc is a $k$-truss based community of \emph{hierarchical structure}, that is, $k$-truss is always contained in the $(k-1)$-truss for any $k\geq 3$. 
\stitle{Structural Trussness.} Recall that trusses have  a hierarchical structure, i.e., for $k\geq 3$, a $k$-truss is always contained in some $(k-1)$-truss \cite{huang2014}. For any vertex or any edge, there exists a $k$-truss with the largest $k$ containing it. We define the trussness of a subgraph, an edge, and a vertex as follows. 
%For every edge $e\in E(G)$, if we can keep in record of the largest value $k$ such that there exists a $k$-truss of $G$ containing $e$, then the retrival of a $k$-truss containing query nodes can be efficiently obtained to satisfy the cohesive constraint. We define such the largest $k$ as the trussness of a subgraph, an edge, and a vertex as follow. 

\begin{definition}
[Trussness] Given a subgraph $H $ $\subseteq $ $ G$, the trussness of $H$ is the minimum support of an edge in $H$ plus $2$, i.e., $\tau(H) = 2+\min_{e\in E(H)}\{sup_{H}(e)\}$. The trussness of an edge $e\in E(G)$ is  $\tau_{G}(e)= \max_{H\subseteq G \wedge e\in E(H)}\{\tau(H)\}$. The trussness of a vertex $v\in V(G)$ is  $\tau_{G}(v)= \max_{H\subseteq G \wedge v \in V(H)} $ $\{\tau(H)\}$.
\end{definition}

Consider the graph $G$ in Figure~\ref{fig.community}, and let the subgraph $H$ be the triangle $\triangle_{q_1v_1v_2}$. Then the trussness of $H$ is $\tau(H)$ $=2$ $+\min_{e\in H}$ ${sup_{H}(e)}$ $=3$, since each edge is contained in one triangle in $H$. However, the trussness of the edge $e(q_1, v_1)$ is 4, because there exists a 4-truss containing $e(q_1, v_1)$ in Figure \ref{fig.subcom}(b), and any subgraph $H$ containing $e(q_1,v_1)$ has $\tau(H)\leq 4$, i.e., $\tau_{G}(e(q_1, v_1)) = $ $ \max_{H\subseteq G \wedge e\in E(H)} $ $\{\tau(H)\} $ $ =4$. In addition, the vertex trussness of $q_1$ is also 4, i.e., $\tau_{G}(q_1) = 4$.

%{\bf From Laks: I visually checked these claims and they seem to be correct. Once you draw the final figures, we should carefully check the correctness of the various claims.} 

Based on the trussness of a vertex (edge), we can infer in constant time whether there exists a $k$-truss containing it. \cut{We construct the structural trussness index as follows. For each vertex $v\in V$, we keep the vertex trussness of $v$, and also maintain the edge trussness of its incident edges in decreasing order of trussness.  This supports efficient checking of whether vertex $v$ or its incident edge is present in a $k$-truss, avoiding expensive $k$-truss search. Also, it can efficiently retrieve  $v$'s incident edges with a given trussness value. In addition, we use a hashtable to maintain all the edges and their trussness.} Notice that for a graph $G$, $\taubar(\emptyset)$ denotes the maximum structural trussness of $G$.

\stitle{Attributed Trussness.}
Structural trussness index is not sufficient for \atc queries. Given a vertex $v$ in $G$ with structural trussness $\tau_{G}(v) \geq k$, there is no guarantee  that $v$ will be present in a \kdtruss with large attribute score  w.r.t. query attributes. E.g., consider the graph $G$ and vertex $v_1$ with $\tau_{G}(v_1) =4 $ in Figure \ref{fig.community}. Here, $v_1$ will not be  present in an \atc for query attributes $W_q=\{``ML"\}$ since it does not have attribute ``ML''. On the contrary, $v_1$ is in a \atc w.r.t. $W_q=\{``DM"\}$. By contrast, $v_9$ is \emph{not} present in a 4-truss w.r.t. attribute ``DM'' even though it has that attribute. To make such searches efficient, for each attribute $w\in \mathcal{A}$, we consider an attribute projected graph, which only contains the vertices associated with attribute $w$, formally defined below. 

\begin{definition}
\textbf{(Attribute Projected graph \& Attributed Trussness).} Given a graph $G$ and an attribute $w\in A(V)$, the projected graph of $G$ on attribute $w$ is the induced subgraph of $G$ by $V_{w}$, i.e., $G_{w}=(V_{w}, E_{V_{w}})\subseteq G$. Thus, for each vertex $v$ and edge $e$ in $G_{w}$, the attributed trussness of $v$ and $e$ w.r.t. $w$ in $G_w$ are respectively defined as $\tau_{G_{w}}(v)=\max_{H\subseteq G_w \wedge v \in V(H)} $ $\{\tau(H)\}$ and $\tau_{G_{w}}(e) = \max_{H\subseteq G_w \wedge e\in E(H)}\{\tau(H)\}$.
\end{definition}

For instance, for the graph $G$ in Figure~\ref{fig.community}, the projected graph $G_w$ of $G$ on $w=``DB"$ is the graph $H_1$ in Figure~\ref{fig.subcom}(a). For vertices $v_1$ and $v_4$, even though both have the same structural trussness $\tau_G(v_1)=\tau_G(v_4)=4$, in graph $H_1$, vertex $v_4$ has attribute trussness $\tau_{H_1}(v_4) = 4$ w.r.t. $w=``DB"$, whereas vertex $v_1$ is not even present in $H_1$, indicating that vertex $v_4$ is more relevant with ``DB" than $v_1$. %As another example, it is easy to verify that the attribute trussness of $v_9$ w.r.t. attribute ``DM'' is 2. 

%\stitle{Inverted Attibutes Index.} We introduce an inverted index for attibutes, \invA. For each attibute $w$ in the graph $G$, int the \invA, it maintains an invert list to store the set of nodes $V_w$, where every node $v\in V_w$ has the attibute $w$, in the descreased order of the vertex trussness. Thus, the inverted index of attribute $w$, denoted by $\invk_{w}$, is in the format of $\{(v_1, \tau(v_1), ..., (v_l, \tau(v_L))\}$. Also, for each vertex $v$, we also design the invented attribute index for incidents edge, i.e., for a given attribute $w$, we can efficient find all edges $(v, u)$, where $u\in N(v)$ and $w\in A(u)$.
\stitle{Inverted Attribute Index.} We propose an inverted index for each attribute $w\in \mathcal{A}$, denoted $\invA_w$. It maintains an inverted list of the vertices in $V_w$, i.e., the vertices containing attribute $w$, in  decreasing order of the vertex structural trussness. Thus, $\invk_{w}$ is in the format $\{(v_1, \tau_{G}(v_1)), ..., (v_l, \tau_{G}(v_l))\}$, $\tau_G(v_j) \ge \tau_G(v_{j+1})$, $j\in[l-1]$. The inverted attribute index and structural trussness index can both be used to speed up Algorithms \ref{algo:basic} and \ref{algo:bulk}.

%$V_w =\{v:w\in doc(v)\}$ 

%Next, we will show how to use the inverted indexs to find a candidate \ktc for a \ktcs query.

\stitle{\ati Construction. } Algorithm \ref{algo:aticon} outlines the procedure of \ati construction. 
%The first step is to invoke the truss decomposition algorithm for computing %trussness \cite{WangC12}. 
It first constructs the index of structural trussness using the structural decomposition algorithm of \cite{WangC12}, then constructs the index of attribute trussness and finally the inverted attribute index.   Now, we analyze the time and space complexity of construction algorithm and the space requirement of \ati.  It takes $O(m\rho)$ time and $O(m)$ space for applying the truss decomposition algorithm on the graph $G$ with $m$ edges \cite{huang2014}, where $\rho$ is the arboricity of $G$, and $\rho \leq \min\{d_{max}, \sqrt{m}\}$. Then, for each keyword $w\in \mathcal{A}$, it invokes the truss decomposition algorithm on the projected graph $G_{w} \subseteq G$, which takes $O(|E(G_w)| \rho)$ time and $O(m)$ space. In implementation, we deal with each $G_{w}$ separately, and release its memory  after the completion of truss decomposition and write attribute trussness index to disk.  Overall, \ati construction takes $O(\rho(m+\sum_{w\in \mathcal{A}}|E(G_w)|))$ time and $O(m)$ space, and the index occupies $O(m+\sum_{w\in \mathcal{A}}|E(G_w)|)$ space on disk. 

\eat{ 
****HOW CAN THE SPACE COMPLEXITY OF THE CONSTRUCTION ALGORITHM BE LESS THAN THESPACE REQUIREMENT OF THE \ati INDEX?**** 

\xin{******  It is beacuase for each attribute $w$, after the truss decompostion on $G_{w}$, we release the memory space of $G_{w}$.  The total of memory is graph stucture with $O(m)$ space. ****}
} 

%Thus, the total time of truss decompositioin on all projected graphs consume $O(\rho\cdot \sum_{w\in \mathcal{A}}|E(G_w)|)\subseteq O(\rho|A(V)|) $. Obviously, the disk size of \ati takes space. In this algorithm, the sorts can be done in linear time with regard of its size.  Overall.

\begin{algorithm}[t]
\small
\caption{\ati Construction($G$)} \label{algo:aticon}
\textbf{Input:} A graph $G=(V, E)$.\\
\textbf{Output:} \ati of $G$.\\
\vskip -0.6cm
\
\begin{algorithmic}[1]

\STATE	Apply the truss decomposition algorithm\cite{WangC12} on $G$.

\STATE	\textbf{for} $v\in G$ \textbf{do}

\STATE  \hspace{0.3cm}  Keep the structural trussness of $v$ and its incident edges in record.
%, in the decresing order of structural trussness.
%\STATE  \hspace{0.3cm}  Maintain a hashtable to retrival all incident edges $(v, u)$ on attributes of vertex $u$.

\STATE	\textbf{for} $w \in A$ \textbf{do}

\STATE  \hspace{0.3cm}  Project $G$ on attribute $w$ as $G_w$.

\STATE  \hspace{0.3cm}  Apply the truss decomposition algorithm\cite{WangC12} on $G_w$.

\STATE  \hspace{0.3cm}  Construct an inverted node list $\invk_{w}$.

\STATE	\textbf{for} $e=(u,v) \in G$ \textbf{do}

\STATE  \hspace{0.3cm}  Build a hashtable to preserve its structural trussness value $\tau_G(e)$ and attribute trussness value $\tau_{G_w}(e)$, where $w\in A(v)\cap A(u)$. 

\end{algorithmic}
\vskip -0.09cm
\end{algorithm}

\subsection{Index-based Query Processing}
\begin{algorithm}[t]
\small
\caption{\LATC($G$, $Q$)} \label{algo:local}
\textbf{Input:} A graph $G=(V, E)$, a query $Q=(V_q, W_q)$.\\
\textbf{Output:} A \kdtruss $H$ with the maximum $\keyw(H, W_q)$.\\
\vskip -0.6cm
\
\begin{algorithmic}[1]

\STATE  Compute an attribute Steiner tree $T$ connecting $V_q$ using attribute truss distance as edge weight;

\STATE  Iteratively expand $T$ into graph $G_{t}$ by adding adjacent vertices $v$, until $|V(G_t)|> \eta$;

\STATE  Compute a connected $k$-truss containing $V_q$ of $G_t$ with the largest trussness $k=k_{max}$; 

\STATE  Let the $k_{max}$-truss as the new $G_t$. 

\STATE 	Apply Algorithm \ref{algo:bulk} on $G_{t}$ to identify \atc with parameters $k=k_{max}$ and $d=\dist_{G_t}(G_t, V_q)$.

\end{algorithmic}
\vskip -0.09cm
\end{algorithm}

\eat{ 
****$\eta$ WAS NEVER INTRODUCED OR MOTIVATED IN THIS PAPER.**** 

\xin{*** I think that in algorithm overview, we want to find \atc in a small graph instead of global grpah. The small graph has the size of  vertices no larger than $\eta$ . }
} 

In this section, we propose an \ati-based query processing algorithm by means of local exploration, called \LATC. 

\stitle{Algorithm overview.} Based on the \ati, the algorithm first efficiently detects a small neighborhood subgraph around query vertices, which tends to be densely and closely connected with the query attributes. Then, we apply Algorithm \ref{algo:bulk} to shrink the candidate graph into a \kdtruss with  large attribute score. The outline of the algorithm \LATC is presented in Algorithm \ref{algo:local}. Note that, when no input parameters $k$ and $d$ are given in \LATC, we design an auto-setting mechanism for parameters $k$ and $d$, which will be explained in Section \ref{sec:exp}. 

%is %parameter-free, which is 
%also easily extended to work with the given $k$ and $d$.  design an auto-setting mechanism for parameters $k$ and $d$, and

To find a small neighborhood candidate subgraph, the algorithm starts from the query vertices $V_q$, and finds a Steiner tree connecting the query vertices. It then expands this tree by adding attribute-related vertices to the  graph. Application of standard Steiner tree leads to poor quality, which we next explain and address. 

\stitle{Finding attributed Steiner tree $T$.} As discussed above, a Steiner tree connecting query vertices is used as a seed for expanding into a \kdtruss. A naive method is to find a minimal weight Steiner tree to connect all query vertices, where the weight of a tree is the number of edges. 
Even though the vertices in such a Steiner tree achieve close distance to each other, using this tree seed may produce a result with a small trussness and low attribute score. For example, for the query $Q=(\{q_1, q_2\}, \{DB\})$ (see Figure~\ref{fig.community}), the tree $T_1$ $=\{(q_1, v_1)$  $, (v_1, q_2)\}$ achieves a weight of 2, which is optimal. However, 
%according to our \ati, 
the edges $(q_1, v_1)$ and $(v_1, q_2)$ of $T_1$ will not be present in any 2-truss with the homogeneous   attribute of ``DB". Thus it suggests growing $T_1$ into a larger graph 
%$H_3$ in Figure \ref{fig.subcom}(c) and 
will yield a low attribute score for $W_q=``DB"$. On the contrary, the Steiner tree $T_2 $ $=\{(q_1, v_4)$  $, (v_4, q_2)\}$ also has a total weight of 2, and both of its edges have the attribute trussness of 4 w.r.t. the attribute ``DB", indicating it could be expanded into a community with large attribute score. For discriminating between such Steiner trees, we propose a notion of attributed truss distance. 

\begin{definition}
[Attribute Truss Distance] Given an edge $e=(u,v)$ in $G$ and query attributes $W_q$, let $\mathcal{G} = \{G_w: w\in W_q\}\cup\{G\}$. Then the attribute truss distance of $e$ is defined as $\hat{\dist}_{W_q}(e) = 1+ $ $\gamma (\sum_{g\in \mathcal{G}} (\taubar(\emptyset) - \tau_{g}(e)))$, where $\taubar(\emptyset)$ is the maximum structural trussness in graph $G$.
\end{definition}\label{def.trussdist}

The set $\mathcal{G}$ consists of $G$ together with all its attribute projected graphs $G_w$, for $w\in W_q$ and the difference $(\taubar(\emptyset) - \tau_{g}(e))$ measures the shortfall in the attribute trussness of edge $e$ w.r.t. the maximum trussness in $G$. 
The sum $\sum_{g\in \mathcal{G}} (\taubar(\emptyset) - \tau_{g}(e))$ indicates the overall shortfall of $e$ across $G$ as well as all its attribute projections. Smaller the shortfall of an edge, lower its distance. Finally, $\gamma$ controls the extent to which small value of  structural and attribute trussness, i.e., a large shortfall, is penalized. 
%The larger $\gamma$ is, the more important edge trussness is in distance calculations. 
Using \ati, for any edge $e$ and any attribute $w$, we can access the structural trussness  $\tau_{G}(e)$  and attribute trussness $\tau_{G_w}(e)$ in $O(1)$ time. Since finding minimum weight Steiner tree is NP-hard, we apply the well-known algorithm of \cite{kou1981fast,mehlhorn1988faster} to obtain a 2-approximation, using attributed truss distance. 
The algorithm takes $O(m|W_q|+m+n\log{n})\subseteq O(m|W_q|+n\log{n})$ time, where $O(m|W_q|)$ is the time taken to compute the attributed truss distance for $m$ edges. 

\stitle{Expand attribute Steiner tree $T$ to Graph $G_t$.} 
Based on the attribute Steiner tree $T$ built above, we locally expand $T$ into a graph $G_{t}$ as a candidate \kdtruss with numerous query attributes. Lemma \ref{lemma.majority} gives a useful principle to expand the graph with insertion of a vertex at a time, while increasing the attribute score. %However, we only consider the case of one majority attribute in Lemma \ref{lemma.majority}. Given a new vertex $v$, the associated query attributes exceed 1, i.e., $|W_q \cap \doc(v)|>1$. Here, we extend the principle in Lemma \ref{lemma.majority}. Let the fraction of vertices of $G_t$ containing an attribute $w\in W_q \cap \doc(v)$ is denoted as $\theta(G_t, W_q\cap \doc(v))$. 
Specifically, if $\theta(G_t, W_q\cap \doc(v)) \geq \frac{\keyw(G_t, W_q)}{2|V(G_t)|}$, then graph $G_{T}\cup \{v\}$ has a larger attribute score than $G_T$. We can identify such vertices whose attribute set includes majority attributes of the current candidate graph and add them to the current graph. 

\eat{ 
****I DON'T LIKE THE PARAGRAPH ABOVE AND THE APPROACH TAKEN. YOU SHOULD REPLACE LEMMA 2 WITH A MORE GENERAL RESULT WITH PROOF, WHICH COVERS THE CASE WHERE $v$ CONTAINS MULTIPLE QUERY ATTRIBUTES. CURRENTLY, THE ABOVE PARAGRAPH SOUNDS LIKE AN AFTER-THOUGHT AND LIKE A LAST MINUTE ADD-ON. REVIEWERS WON'T LIKE THAT. SO, PLEASE UPGRADE LEMMA 2 FOR THIS GENERAL CASE AND DIRECTLY REFER TO IT ABOVE.****

\xin{**** I have updated the Lemma 2 and use it in this section. *****}
} 

%Here, we give a useful principle to expand graph with the insertions of vertices with query attributes.
% Based on the built attribute Steiner tree $T$, we locally expand $T$ into graph $G_{t}$ as a candidate of \kdtruss with numerous query attributes. Here, we give a useful principle to expand graph with the insertions of vertices with query attributes.

% \begin{lemma}\label{lemma.nomonotone}
% Given two \ktcs $H$ and $H'$ where $V(H') = V(H)\cup \{v\}$, and query attribute $W_q$, if $\sum_{w\in W_q}(2|V(H)\cap V_w|+1) \geq \frac{\keyw(H, W_q)}{2|V(H)|}$, then $\keyw(H', W_q) > \keyw(H, W_q)$ holds.
% \end{lemma}

% \begin{proof}
% The proof is similar with Lemma \ref{lemma.majority}, which is ommited for limited space.
% \end{proof}

% As we can see that, for a vertex $v$ and graph $H$, the score of $\sum_{w\in W_q}(2|V(H)\cap V_w|+1) $ is exactly the node attribute score of $v$ in graph $H'=H\cup{v}$, denoted as $\keyw_{H}(v, W_q)$. Based on this lemma, we can iteratively adding the nodes with the large node attibute score to increase the whole graph's attibute score. 
Now, we discuss the expansion process, conducted in a BFS manner. We start from vertices in $T$, and iteratively insert adjacent vertices with the largest vertex attribute scores into $G_{t}$ until the vertex size exceeds a threshold $\eta$, i.e., $|V(G_t)|\leq \eta $, where $\eta $ is empirically tuned. After that, for each vertex $v\in V(G_t)$, we add all its adjacent edges $e$ into $G_t$.  %where at least one of endpoints of $e$ should be in $G_t$ and $\tau(e)\geq k$.

%\stitle{Apply \bulk on $G_t$ with auto-setting parameters.} Based on the graph $G_t$ constructed above, we apply Algorithm \ref{algo:bulk} with given parameters $k$ and $d$ on $G_t$ to find an \atc. If input parameters $k$ and $d$ are not supplied, we can set them automatically as follows. We first compute a $k$-truss with the largest $k$ connecting all query vertices. Let $k_{max}$ denote the maximum trussness of the subgraph found. We set the parameter $k$ to be $k_{max}$. We also compute the query distance of $G_t$ and assign it to $d$, i.e., $d := \dist_{G_t}(G_t, V_q)$. We then invoke the \bulk algorithm on $G_t$ with parameters $k, d$ to obtain a \ktc with large trussness and high attribute cohesiveness. 

\stitle{Apply \bulk on $G_t$ with auto-setting parameters.} Based on the graph $G_t$ constructed above, we apply Algorithm \ref{algo:bulk} with given parameters $k$ and $d$ on $G_t$ to find an \atc. If input parameters $k$ and $d$ are not supplied, we can set them automatically as follows. We first compute a $k$-truss with the largest $k$ connecting all query vertices. Let $k_{max}$ denote the maximum trussness of the subgraph found. We set the parameter $k$ to be $k_{max}$. We also compute the query distance of $G_t$ and assign it to $d$, i.e., $d := \dist_{G_t}(G_t, V_q)$. We then invoke the \bulk algorithm on $G_t$ with parameters $k, d$ to obtain a \ktc with large trussness and high attribute cohesiveness.

\stitle{Friendly mechanism for query formulation.} 
\LL{Having to set values for many parameters for posing queries using \LATC can be daunting. To mitigate this, we make use of the auto-setting of parameters $k$ and $d$. Additionally, we allow the user to omit the query attribute parameter $W_q$ in a query $Q(V_q,W_q)$ and write $Q(V_q, \_)$. Thus, only query nodes need to be specified. Our algorithm will automatically set $W_q := \bigcup_{v\in V_q} A(v)$ by default. The rationale is that the algorithm will take the the whole space of all possible attributes as input, and leverage our community search algorithms to find communities with a proper subspace of attributes, while achieving high scores. For example, consider the query $Q = (\{q_1,q_2\}, \_)$ on graph $G$ in Figure~\ref{fig.community}, \LATC automatically sets $W_q := \{DB, DM, ML\}$. The discovered community is shown in Figure \ref{fig.subcom}(b), which illustrates the feasibility of this strategy. This auto-complete mechanism greatly facilitates query formulation, 

This auto-complete query formulation is useful to identify relative attributes for discovered communities, which benefits users in a simple way.
} 
\eat{ 
\new{If users do not know the underlying scenario of given datasets, it may make them feel hard to formulate a useful query $Q(V_q, W_q)$. Now,  we offer users a simple way to formulate query $Q(V_q, W_q)$. Beside the above auto-setting parameters $k$ and $d$ in \LATC, we consider an auto-setting of query attributes $W_q$ where $W_q$ is not given or empty.  In this case, our algorithm automatically sets $W_q = \bigcup_{v\in V_q} A(v)$ by default. The rational principle is that it takes the whole space of all possible attributes as input, and leverage our community search algorithms on searching communities with a proper subspace of attributes, which achieves high scores. For example, consider the query $Q = (\{q_1,q_2\}, \emptyset)$ on graph $G$ in Figure~\ref{fig.community}, \LATC automatically sets $W_q=\{DB, DM, ML\}$. The discovered community is shown in Figure \ref{fig.subcom}(b), which shows the feasibility of this strategy. This auto-complete query formulation is useful to identify relative attributes for discovered communities, which benefits users in a simple way.}
} 

%{\bf From Laks: I have not yet edited the next para. See my message on slack about what needs to be done.} 

\stitle{Handling bad queries.} \new{In addition to auto-complete query formulation, we discuss how to handle bad queries issued by users. Bad queries contain query nodes and query attributes that do not constitutes a community. Our solution is to detect outliers of bad queries and then suggest good candidate queries for users. The whole framework includes three steps. 
First, it identifies bad queries. 
%\\================\\ We define a bad query $Q(V_q, W_q)$ such that there exists no \kdtruss neither containing $V_q$ nor achieving positive score for attributes $W_q$. \\================\\ 
Based on the structural constraint of \kdtruss, if query nodes span a long distance and loosely connected in graphs, it tends to be bad queries. In addition, if none of query attributes present in the proximity of query nodes, it suggests to have no communities with homogeneous attributes, indicating bad queries as its. Instances of bad queries $Q(V_q, W_q)$ have no \kdtruss neither containing $V_q$ nor achieving non-zero score for attributes $W_q$.
Second, it recommends candidates of good queries. Due to outliers existed in bad queries, we partition the given query into several small queries. Based on the distribution of graph distance, graph cohesiveness, and query attribute, we partition given query nodes into several disjoint good queries. 
%\\================\\ Specifically, we start from one query node, and find the \kdtruss community. All query nodes and attributes present in this community is regarded as one query. The above process can be repreated until all query nodes are presented in one community or no such one community cotaining it.\\================\\
Specifically, we start from one query node, and find the \kdtruss community containing it.  The query nodes and query attributes present in this community are formed as one new query. The process is repreated until all query nodes are presented in one community or no such one community cotaining it. Thus, we have several new queries that are good to find \atc.
Third, our approach quickly terminates by returning no communities, due to the violation of \kdtruss and irrelevant query attributes. }

%% file: exp.tex
\cut{In this section, we evaluate the efficiency and effectiveness of our proposed \atc model and algorithms. All algorithms are implemented in C++, and the experiments are
conducted on a Linux Server with Intel Xeon CUP X5570 (2.93 GHz) and
50GB main memory.}
In this section, we test all proposed algorithms on a Linux Server with Intel Xeon CUP X5570 (2.93 GHz) and 50GB main memory.
\subsection{Experimental Setup}
\stitle{Datasets.} We conduct experimental studies using 7 real-world networks.%, where all networks are treated as undirected.  
 The network statistics are reported in Table~\ref{tab:dataset}. 

The first dataset is PPI network, Krogan 2006, from the BioGRID database, where the PPI data are related to the yeast Saccharomyces cerevisiae \cite{hu2013utilizing}. %, and are  obtained with a higher reliability than its extended dataset. 
Each protein has three kinds of attributes: biological processes, molecular functions, and cellular components. There are 255 known protein complexes for Sacchromyces cerevisiae in the MIPS/CYGD \cite{hu2013utilizing}, which we regard as ground-truth communities. 

The second dataset is Facebook ego-networks. For a given user id $X$ in Facebook network $G$, the ego-network of $X$, denoted ego-facebook-$X$, is the induced subgraph of $G$ by $X$ and its neighbors. 
The dataset contains 10 ego-networks indicated by its ego-user $X$, where $X \in \{0, 107, 348,$ $  414, 686, $ $698, 1684,$ $ 1912, 3437,$ $ 3890\}$. For simplicity, we abbreviate ego-facebook-$X$ to f$X$, e.g., f698. Vertex  attributes are collected from real profiles and anonymized, e.g., political, age, education, etc. Each ego-network has several overlapping ground-truth communities, called friendship circles \cite{mcauley2012learning}. Note that the statistics of Facebook in Table~\ref{tab:dataset} are results averaged over 10 networks. 

\new{The third and fourth datasets are web graphs respectively gathered from two universities of Cornell and Texas.\footnote{\scriptsize{\url{http://linqs.cs.umd.edu/projects/projects/lbc/}}} Webpages are partitioned into five groups including course, faculty, student, project, and staff. Vertex attributes are unique words  frequently present in webpages.}

The other 5 networks, Amazon, DBLP, Youtube, LiveJournal and Orkut, contain 5000 top-quality ground-truth communities. %However, since the vertices on these networks have no attributes, we generate an attribute set consisting of $|A(V)|=0.05\cdot |V(G)|$ different attribute values in each network $G$, \new{where 0.05 is a safe number following the portions of attributes in datasets with real attributes in Table~\ref{tab:dataset}. } 
\new{
However, since the vertices on these networks have no attributes, we generate an attribute set consisting of $|\mathcal{A}|=0.005\cdot |V|$ different attribute values in each network $G$. The average number of attribute/vertex $\frac{|\mathcal{A}|}{|V|} = 0.005$ is less than the proportion of attributes to vertices in datasets with real attributes (e.g., the value of 0.12 in Facebook) in Table~\ref{tab:dataset}. A smaller attribute pool $\mathcal{A}$ makes homogeneity of synthetic attributes in different communities more likely, which stresses testing our algorithms.
%The parameter choice of 0.05 is comparable to the proportion of attributes to vertices in datasets with real attributes in Table~\ref{tab:dataset}.  The value of 0.05 is smaller than the average value of attributes/vertex $\frac{|\mathcal{A}|}{|V|} = 0.12$ found in the Facebook dataset, which makes homogeneity of synthetic attributes in different communities more likely, thus stress testing our algorithms for small attribute pools $\mathcal{A}$.
%However, since the vertices on these networks have no attributes, we generate an attribute set consisting of $|\mathcal{A}|=0.05\cdot |V|$ different attribute values in each network $G$. The parameter choice of 0.05 is comparable to the proportion of attributes to vertices in datasets with real attributes in Table~\ref{tab:dataset}. 
%The value of 0.05 is smaller than the value of $\frac{|\mathcal{A}|}{|V|} = 0.12$ found in the Facebook dataset, which makes homogeneity of synthetic attributes less likely, thus stress testing our algorithms.
} 
%\new{\\\\================\\The value of 0.05 is smaller than the average value of attributes/vertex $\frac{|\mathcal{A}|}{|V|} = 0.12$ found in the Facebook dataset, which makes homogeneity of synthetic attributes in different communities more likely, thus stress testing our algorithms for small attribute pools $\mathcal{A}$.\\================\\\\}
\note[Laks]{See my comments on this on slack, though.} 
For each ground-truth community, we randomly select 3 attributes, and  assign each of these attributes to each of random 80\% vertices in the community. In addition, to model noise in the data, for each vertex in graph, we randomly assign a random integer of $[1,5]$ attributes to it.  Except Krogan, all other datasets are available from the Stanford Network Analysis Project.\footnote{\scriptsize{\url{snap.stanford.edu}}} 
%\new{For each ground-truth community, we randomly select 3 attributes, and assign each of these attributes to each of random $Y$\% vertices in the community, where $Y$ is also a random number in [50, 90].} In addition, to model noise in the data, for each vertex in graph, we randomly assign a random integer of $[1,5]$ attributes to it.  Except Krogan, all other datasets are available from the Stanford Network Analysis Project.\footnote{\scriptsize{\url{snap.stanford.edu}}} 

 %\vskip -0.1in
\begin{table}[t]
\begin{center}\vspace*{-0.6cm}
\scriptsize
%\small
\caption[]{Network statistics (K $=10^3$ and M $=10^6$)}\label{tab:dataset}
\begin{tabular}{|l|r|r|r|r|r|r|}
\hline
{\bf Network} & $|V|$ & $|E|$ & $d_{max}$ & $\taubar(\emptyset)$  & $|\mathcal{A}|$ & $|\doc(V)|$\\
\hline \hline
Krogan & 2.6\textbf{K} & 7.1\textbf{K} & 140 & 16 & 3064 & 28151 \\ \hline
%Facebook	 & 4\textbf{K}	 & 88\textbf{K} & 1,045  & 97 & XXX & XXX\\ \hline
Facebook	 & 1.9\textbf{K}	 & 8.9\textbf{K} & 416  & 29 & 228 & 3944\\ \hline
Cornell & 195 & 304 & 94 & 4 & 1588 & 18496 \\ \hline
Texas	 & 187	 & 328 & 104 & 4  & 1501 & 15437\\ \hline
Amazon	& 335\textbf{K} &	926\textbf{K} &	549 &	7 & 1674 & 1804406\\ \hline
DBLP 	 & 317\textbf{K}	 & 1\textbf{M} & 342  &114 & 1584 & 1545490\\ \hline
Youtube & 1.1\textbf{M}	 & 3 \textbf{M}		 & 28,754 & 19 & 5327 & 2163244\\ \hline
LiveJournal & 4\textbf{M} & 35\textbf{M} & 14,815 & 352 & 11104 & 12426432 \\ \hline
Orkut	 & 3.1\textbf{M}	 & 117\textbf{M} & 33,313 & 78  & 9926 & 10373866\\ \hline
%WikiTalk & 2.4\textbf{M} & 5\textbf{M} & 100029 & 53 \\  \hline
%Flickr	& 80\textbf{K} &	11.8\textbf{M} & 5706 & 308\\ \hline
\end{tabular}\vspace*{-0.6cm}
\end{center}
%\vspace*{-0.3cm}
\end{table}
%\vskip -0.15in

\stitle{Algorithms Compared.} To evaluate the efficiency and effectiveness of our proposed index and algorithms,
we evaluate and compare the three algorithms -- \textbf{\basic}, \textbf{\bulk}, and \textbf{\LATC}. Here, \textbf{\basic} is the top-down greedy approach in Algorithm \ref{algo:basic}, which removes single node with the smallest node attribute contribution in each iteration. \textbf{\bulk} is an improved greedy algorithm in Algorithm \ref{algo:bulk}, which removes a set of nodes with size $\frac{\epsilon}{1+\epsilon} |V(G_i)|$ from graph $G_i$ in each iteration. We empirically set $\epsilon = 0.03$.  \textbf{\LATC} is the bottom-up local exploration approach in Algorithm \ref{algo:local}. For all methods, we set the parameter $k=4$ and $d=4$ by default. For \LATC, we empirically set the parameter $\eta = 1000$ and $\gamma = 0.2$, where $\eta=1000$ is selected in order to achieve stable quality and efficiency by testing $\eta$ in the range $[100, 2000]$, and $\gamma = 0.2$ is selected to balance the cohesive structure and homogeneous  attributes for communities explored. %the requirements of cohesive and close structure, and attribute scores for communities explored

In addition, to evaluate the effectiveness of the \atc model on attributed graphs, we implemented three
state-of-the-art community search methods -- \ACC, \MDC and \LCTC. The $k$-core based attribute community search (\ACC) \cite{FangCLH16} finds a connected $k$-core containing one given query node with the maximum number of common query attributes shared in this community. The minimum degree-based community search (\MDC) \cite{sozio2010} globally finds a dense subgraph containing all query nodes with the highest minimum degree under distance and size constraints. The closest truss community search (\LCTC) \cite{huang2015approximate} locally finds a connected $k$-truss with the largest $k$ containing all query nodes, and a small diameter. Note that both \MDC and \LCTC only consider the graph structure and ignore the attributes. \ACC considers both graph structure and attributes, but it only deals with a single query node with query attributes\cut{ and uses $k$-core as community model}. 
 %All methods are implemented using the same data structures for the graph, Steiner tree, and hashtable as we do for \LCTC in experiments.

%For the effectiveness of eliminating ``free riders'', we compare our methods with \TRUS (Algorithm \ref{algo:findg0}), which finds the connected $k$-truss graph containing query nodes with the largest $k$only.

\stitle{Queries.} For each dataset, we randomly test 100 sets of queries $Q=(V_q, W_q)$, where we set both the number of query nodes $|V_q|$, and the number of query attributes $|W_{q}|$ to 2 by default. 

\stitle{Evaluation Metrics.}
To evaluate the quality of communities found by all algorithms, we measure the F1-score reflecting the alignment between a
discovered community $C$ and a ground-truth community $\hat{C}$. Given a ground-truth community $\hat{C}$, we randomly pick query vertices and query attributes from it and query the graph using different algorithms to obtain the discovered community $C$. Then,
$F1$ is defined as $F1(C, \hat{C})$ $=\frac{2\cdot
prec(C, \hat{C}) \cdot recall(C, \hat{C})}{prec(C, \hat{C}) +
recall(C, \hat{C})}$
where $prec(C, \hat{C}) =\frac{|C\cap \hat{C}|}{|C|}$ is the precision
and $recall(C, \hat{C}) =\frac{|C\cap \hat{C}|}{|\hat{C}|}$ is the
recall. \new{For all efficiency experiments, we consistently report the running time in seconds. }
%If the running time of a query exceeds 1 hour, we treat it as infinite.  precision and recall

\eat{
****HOW IS A DISCOVERED COMMUNITY MATCHED WITH A GT COMMUNITY? DO YOU 
FIND THE ONE WITH THE MAX. F1-SCORE?**** 

\xin{**** For each ground-truth community, we randomly select the query nodes and attributes from it, and use this query for finding community by all methods. The F1-SCORE is measured on this ground-truth community. Thus, we do not find the maximum one.****}
}

\subsection{Quality Evaluation}\label{sec.exp-quality}

%synthetic attributes
% with ground-truth communities

\begin{figure}[t]
%\vskip -0.2cm
\centering \mbox{\hskip -0.1in
\includegraphics[width=1.0\linewidth]{\PP 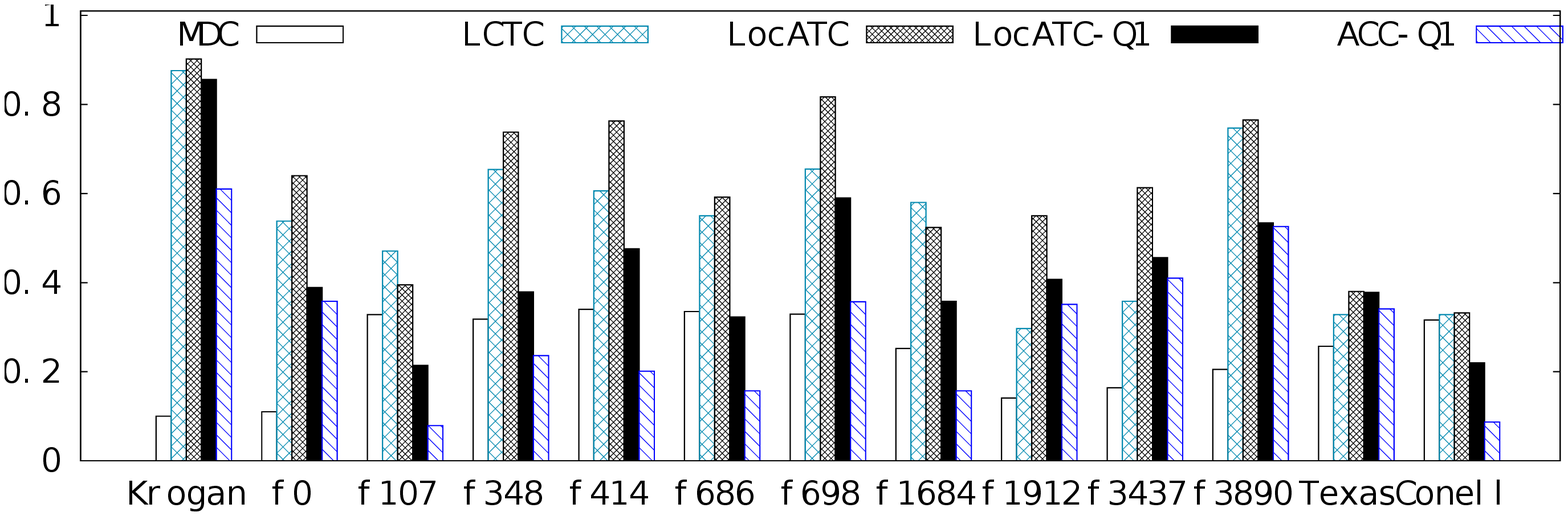} \hskip -0.1in
%\includegraphics[width=1.0\linewidth]{./Figure/exp/facebook_F1.eps} \hskip -0.1in
%\subfigure[Querytime]{\includegraphics[width=0.45\linewidth]{./Figure/exp/quality_time.eps}}
}
\vskip -0.2in
%\caption{Quality evaluation on Krogan network and 10 Facebook Ego-networks with ground-truth communities} 
\caption{Quality evaluation ($F_1$ score) on networks with real-world attributes and ground-truth communities}
\label{fig.fb_F1}
\vspace*{-0.3cm}
\end{figure}

\begin{figure}[t]
%\vskip -0.4cm
\centering 
\includegraphics[width=0.5\linewidth,height=2.5cm]{\PP 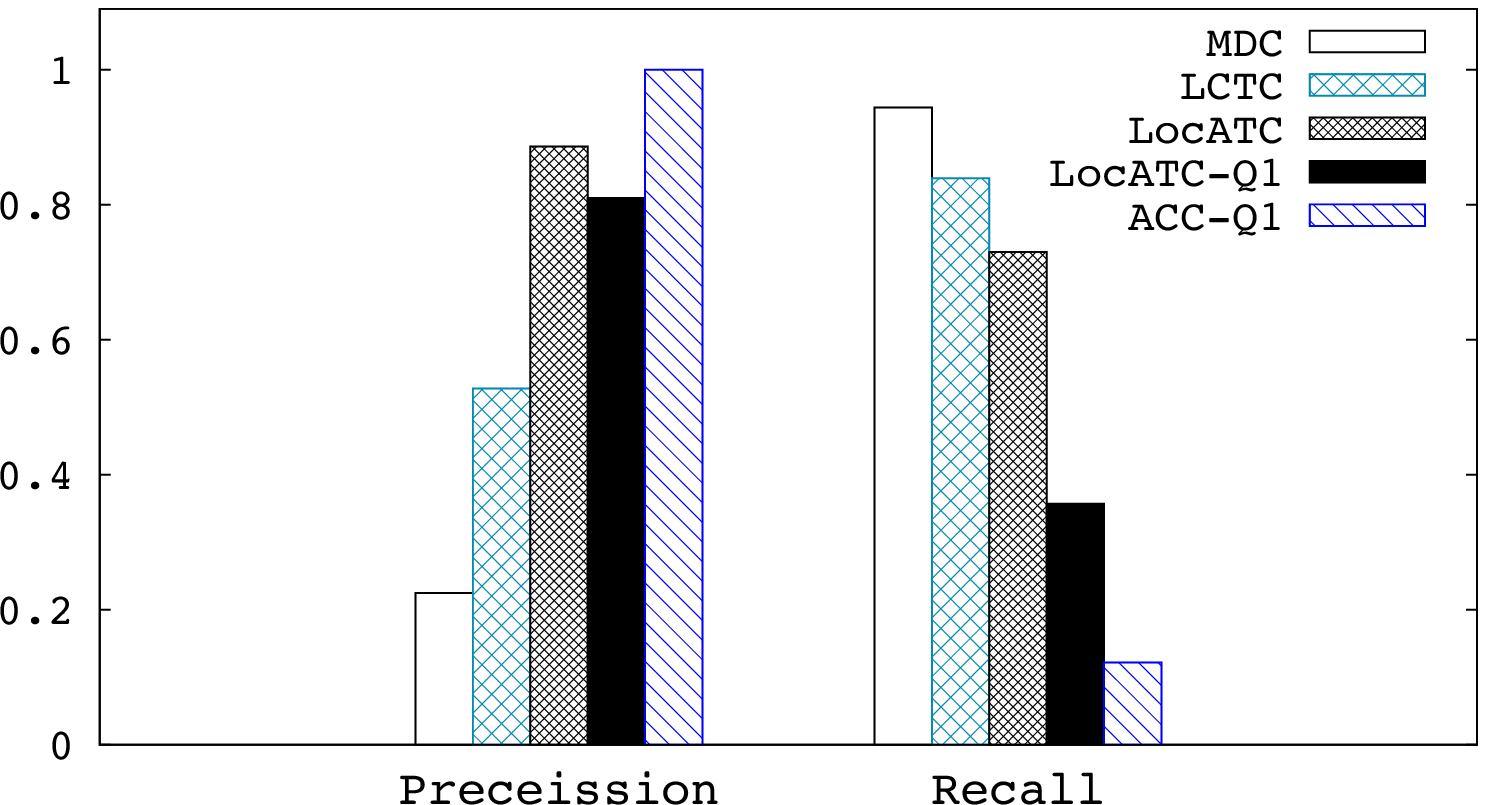}
\vskip -0.2in
\caption{Comparison of precision and recall on f414 network.} \label{fig.prec_recall}
\vspace*{-0.1cm}
\end{figure}

\begin{figure}[t]
\vskip -0.2cm
\centering \mbox{\hskip -0.1in
\includegraphics[width=1.0\linewidth]{\PP 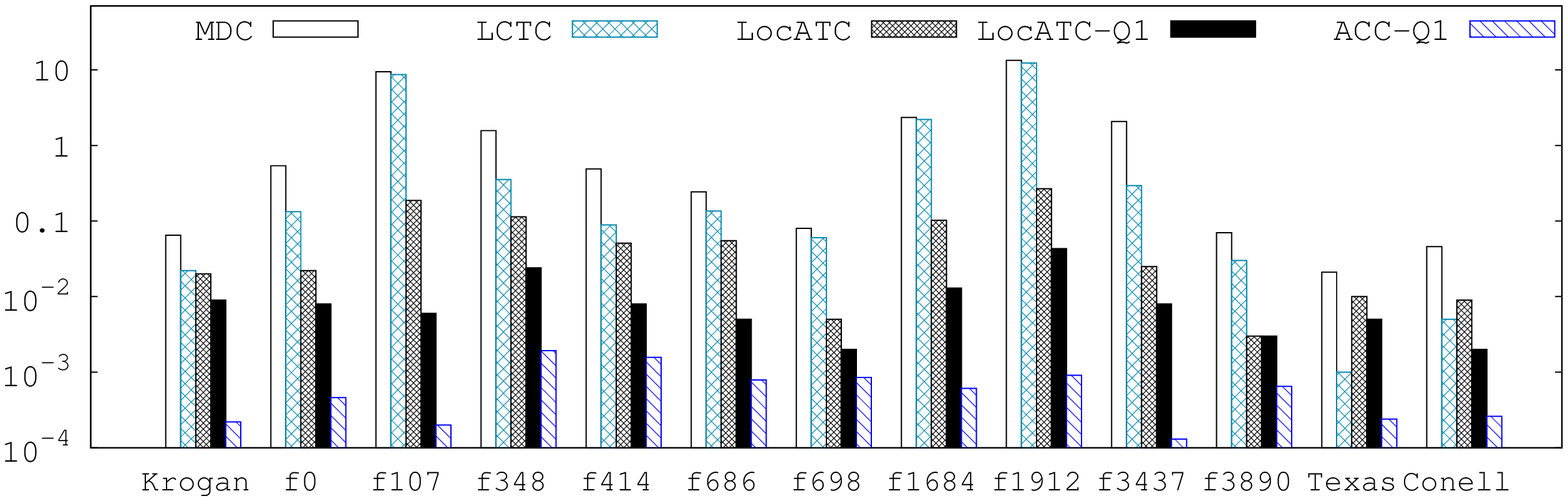} \hskip -0.1in
%\includegraphics[width=1.0\linewidth]{./Figure/exp/facebook_Time.eps} \hskip -0.1in
%\subfigure[Querytime]{\includegraphics[width=0.45\linewidth]{./Figure/exp/quality_time.eps}}
}
\vskip -0.2in
%\caption{Efficiency evaluation on Krogan network and 10 Facebook Ego-networks with ground-truth social circles} \label{fig.fb_time}
%\caption{Efficiency evaluation on Krogan network and 10 Facebook Ego-networks with ground-truth communities} 
\caption{\new{Efficiency evaluation (query time in seconds) on networks } with real-world attributes and ground-truth communities}
\label{fig.fb_time}
\vspace*{-0.3cm}
\end{figure}

\begin{figure}[t]
\vskip -0.2cm
\centering \mbox{\hskip -0.1in
\subfigure[$F_1$ score]{\includegraphics[width=0.50\linewidth]{\EP 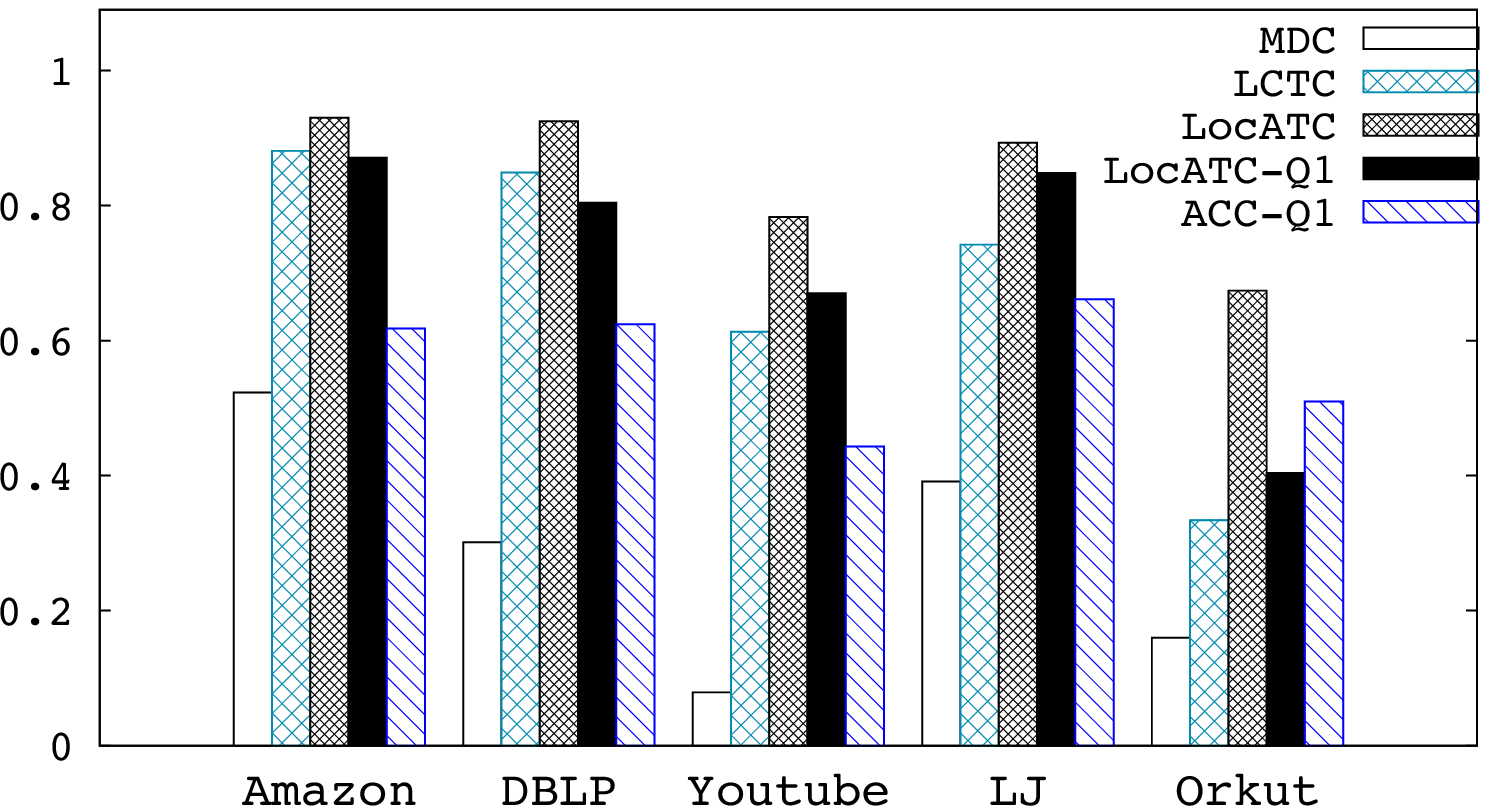}} \hskip -0.1in
\subfigure[\new{Query time (in seconds)}]{\includegraphics[width=0.50\linewidth]{\EP 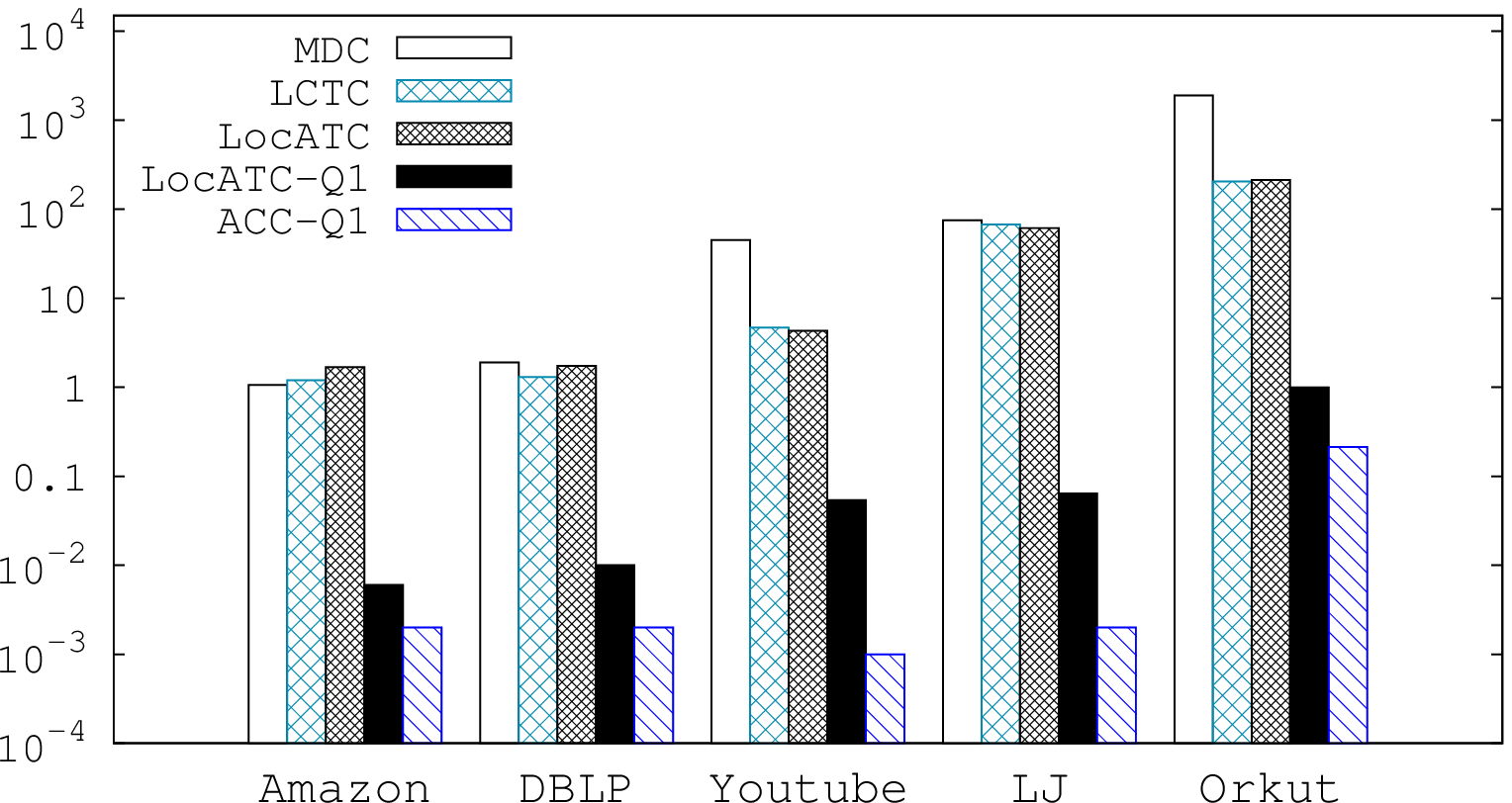}}
}
\vskip -0.2in
\caption{Evaluation on networks with synthetic attributes and ground-truth communities} \label{fig.quality}
\vspace*{-0.3cm}
\end{figure}

To evaluate the effectiveness of different community models, we compare \LATC with
three state-of-the-art methods -- \ACC, \MDC and \LCTC on attributed networks with ground-truth
communities.
%three other methods \MDC, \LCTC and \TRUS 

\stitle{Networks with real-world attributes.} We experiment with the Krogan network and the 10 Facebook ego-networks, all having real-world attributes. For every ground-truth community, we randomly select a set of query nodes with  size drawn uniformly at random from $[1,16]$. We use 2  representative attributes from the community as query attributes. We choose attributes occurring most frequently in a given community and rarely occurring in other communities as representative attributes. We evaluate the accuracy of detected communities and report the averaged F1-score over all queries on each network.
%Since all We randomly select query nodes that appear in a unique ground-truth
%community, and select 1,000 sets of such query nodes with the size
%randomly ranging from 1 to 10. 

\new{Figure \ref{fig.fb_F1} shows the F1-score on Krogan, Cornell, Texas, and the 10 Facebook ego-networks.} 
%Figure \ref{fig.fb_F1} shows the F1-score on Krogan network and on 10 Facebook ego-networks where each is denoted by its ego user.
Our method (\LATC) achieves the highest F1-score on most networks, except for facebook ego-networks f104 and f1684. \xin{The reason is that vertices of ground-truth communities in f104 and f1684 are strongly connected in structure, but are not very homogeneous on query attributes.}
\eat{\\ 
****LCTC IS SLIGHTLY BETTER THAN LATC ON THOSE TWO FB NETWORKS. CAN WE OFFER AN EXPLANATION?**** \\
\xin{****See above.****}\\
}\LCTC has the second best performance, and   outperforms \MDC on all networks. We can see that \MDC and \LCTC do not perform as well as \LATC, because those  community models only consider  structure metrics, and ignore  attribute features. Note that for each query with multiple query vertices, the attribute community search method \ACC randomly takes one query vertex as input. We make this explicit and denote it as \ACC-Q1 in Figure \ref{fig.fb_F1}. For comparison, we apply the same query on our method \LATC, and denote it as \LATC-Q1. \LATC-Q1 clearly outperforms \ACC-Q1 in terms of F1-score, showing the superiority of our \atc model. In addition, \LATC achieves higher score than \LATC-Q1, indicating our method can discover more accurate communities with more query vertices. \LL{Furthermore, we also compare the precision and recall of all methods on f414 network in Figure \ref{fig.prec_recall}. \MDC perform the worst on precision, since it considers no query attributes and  includes many nodes that are not in ground-truth communities. \ACC-Q1 is the winner on precision, which is explained by the strict attribute constriantin its definition. On the other hand, in terms of recall, \ACC-Q1 is the worst method as it only identifies a small part of ground-truth communities. Overall, \LATC achieves a good balance between precision and recall. This is also reflected in \LATC achieving  the best F1-score on most datasets (Figure \ref{fig.fb_F1}). }

%Figure \ref{fig.fb_time} shows the running time performance of all methods. In terms of supporting multiple query vertices, \LATC is the fastest solution, and runs up to two orders of magnitude faster than \MDC and \LCTC on small ego-networks in Facebook. For one query vertex, \ACC-Q1 runs faster than \LATC-Q1, since $k$-cores can be computed quicker than $k$-trusses.

Figure \ref{fig.fb_time} shows the running time performance of all methods. \new{In terms of supporting multiple query vertices, \LATC runs up to two orders of magnitude faster than \MDC and \LCTC on small ego-networks in Facebook, and \LCTC is the winner on Cornell and Texas  networks. For one query vertex, \ACC-Q1 runs faster than \LATC-Q1, since $k$-cores can be computed quicker than $k$-trusses.}  

\stitle{Networks with synthetic attributes.} In this experiment, we test on 5 large networks -- DBLP,
Amazon, Youtube, LiveJournal, and Orkut, with ground-truth communities and synthetic attributes~\cite{YangL12}. We randomly select 1000 communities from 5000 top-quality ground-truth communities as answers.
For each community, we generate a query $Q=(V_q, W_q)$, where query vertices $V_q$ are randomly selected from this community with a  size randomly drawn from $[1,16]$, and query attributes $W_q$ are the 3 community attributes.
%generate 100 set of queries, in which a set of query nodes with the size ranging from 1 to 16 are selected from a random ground-truth community, and 2 community attributes are seletcted as query attributes. 
Figure \ref{fig.quality} (a) shows the F1-score.  Our method \LATC achieves the best F1-score among all compared methods on all networks, and \MDC is the worst. The results clearly show the effectiveness and superiority of our \atc model for attributed community search. Moreover, \LATC-Q1 outperforms \ACC-Q1 on most networks. 

Figure \ref{fig.quality} (b) reports the running times of all methods on all networks. As we can see, \LATC runs much faster than \MDC, and is close to \LCTC. This indicates that \LATC can achieve high quality over large networks with an efficiency comparable to \LCTC. Thus, compared to \LCTC, the additional overhead of reasoning with attribute cohesiveness is small while the improvement in quality of communities discovered is significant. In addition, \LATC-Q1 runs much faster than \LATC, which shows the high efficiency of local exploration for one query vertex.% without using Steiner tree search algorithm.

\subsection{Efficiency Evaluation}
We evaluate the various approaches using different queries on ego-facebook-414 (aka f414) and DBLP.

\stitle{Varying query vertex size  $|V_q|$.} We test 5 different values of $|V_q|$, i.e., $\{1, 2, 4, 8, 16\}$ with the default query attribute size $|W_q|=2$. For each value of $|V_q|$, we randomly generate 100 sets of queries, and report the average running time in seconds.  The
results for f414 and DBLP are respectively shown in Figure \ref{fig.qsize_time} (a) and (b). \LATC achieves the best performance, and increases smoothly with the increasing query vertex size.  \bulk is more effective than \basic, thanks to the bulk deletion strategy. Most of the cost  of \bulk and \basic comes from computing the maximal \kdtruss $G_0$. All methods takes less time on f414 than on DBLP network, due to the small graph size of f414. %and efficicent truss-index avoiding computing $G_0$ from scratch. All methods take less time cost on Ego-facebok-414 than DBLP. This is because Ego-facebok-414 only contains XXX vertices, which is far less then DBLP network.  

\stitle{Varying query attribute size $|W_q|$.}  We test 5 different values of $|W_q|$ from 1 to 5. For each value of $|W_q|$, we randomly generate 100 sets of queries, and report the average running time. We show the result for f414 and DBLP respectively in Figure \ref{fig.wsize_time} (a) and (b). Figure \ref{fig.wsize_time} shows all methods register only a modest increase in running time as $|W_q|$ increases. Again, the local exploration method \LATC significantly outperforms other methods.

\begin{figure}[t]
\vskip -0.4cm
\centering \mbox{
\subfigure[f414]{\includegraphics[width=0.5\linewidth,height=2.5cm]{\EP 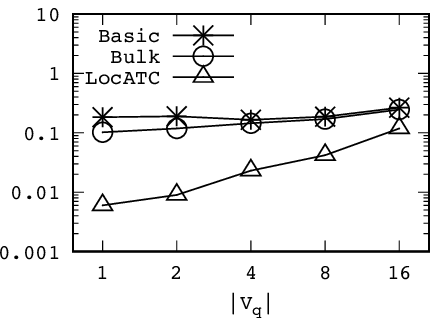}}%\hskip -0.15in
\subfigure[DBLP]{\includegraphics[width=0.5\linewidth,height=2.5cm]{\EP 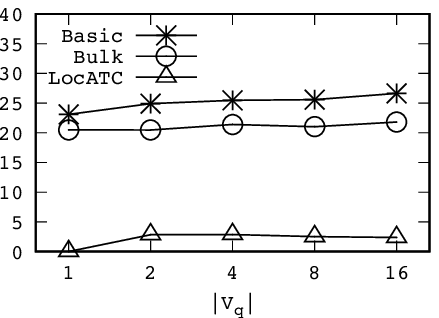}}
%\hskip -0.15in
} \vskip -0.2in
\caption{Varying query vertex size $|V_q|$: Query Time} \label{fig.qsize_time}
\vspace*{-0.1cm}
\end{figure}

\begin{figure}[t]
\vskip -0.4cm
\centering \mbox{
\subfigure[f414]{\includegraphics[width=0.5\linewidth,height=2.5cm]{\EP 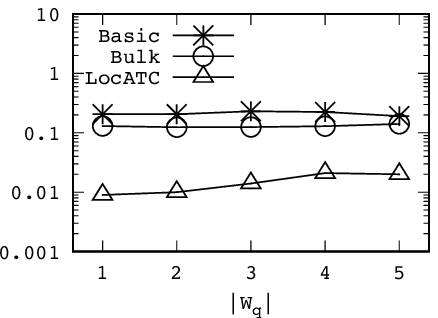}}%\hskip -0.15in
\subfigure[DBLP]{\includegraphics[width=0.5\linewidth,height=2.5cm]{\EP 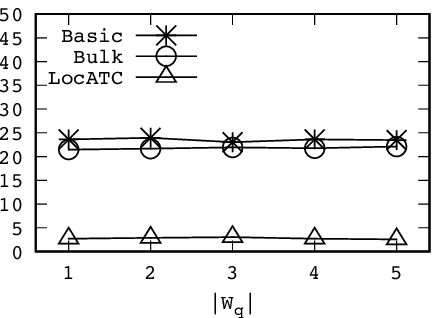}}
%\hskip -0.15in
} \vskip -0.2in
\caption{Varying query attribute size $|W_q|$: Query Time} \label{fig.wsize_time}
\vspace*{-0.1cm}
\end{figure}

\subsection{Index Construction}

% \begin{table}[H]
% \begin{center}\vspace*{-0.25cm}
% \scriptsize
% %\small
% \caption[]{\textbf{Comparison of index size (in Megabytes) and index construction time (wall-clock time in seconds)}}\label{tab:index}
% \begin{tabular}{|c|c| |c|c| |c|c|}
% \hline
% \multirow{2}{*}{Network} & Graph & \multicolumn{2}{|c||}{Index Size} & \multicolumn{2}{|c|}{Index Time} \\
% \cline{3-6} & Size & K-Truss & ACT-Index  & K-Truss & ACT-Index  \\\hline
% Krogan2006 & 0.24 & 0.15 & 1.8 & 0.11 & 0.786\\\hline
% Amazon & 13 & 19 & 52 & 6.7 & 17.7 \\\hline
% DBLP & 15 & 20 & 37 & 14.2 & 31.2  \\\hline
% Youtube & 39 & 59 & 80 & 75.6 & 107.6 \\\hline
% LiveJournal & 484 & 666 & 921 & 2142 & 3528 \\\hline
% Orkut & 1658 & 2190& 3161 &  21011& 27773 \\\hline
% %LiveJournal	 & 672 & 	1003 & 	3174 & 	1176 & 	1686 \\\hline
% %Orkut	 & 1792 	 & 2662 	 & 8714 	 & 2291 	 & 3342 \\\hline
% \end{tabular}%\vspace*{-0.2cm}
% \end{center}
% \end{table}

\begin{table}[H]
\begin{center}\vspace*{-0.5cm}
\scriptsize
%\small
\caption[]{\textbf{Comparison of index size (in Megabytes) and index construction time (wall-clock time in seconds)}}\label{tab:index}
\begin{tabular}{|c|c| |c|c| |c|c|}
\hline
\multirow{2}{*}{Network} & Graph & \multicolumn{2}{|c||}{Index Size} & \multicolumn{2}{|c|}{Index Time} \\
\cline{3-6} & Size & K-Truss & \ati  & K-Truss & \ati  \\\hline
Krogan & 0.24 & 0.15 & 1.8 & 0.11 & 0.786\\\hline
% Amazon & 24&  19 & 70 & 6.7&  21.7 \\\hline
% DBLP & 23 & 20 & 54&  14.2&  35.2  \\\hline
% Youtube & 52 & 59&  471&  75.6 & 113.6 \\\hline
% LJ & 545 & 666&  1046&  2142&  3556  \\\hline
% Orkut & 1729 & 2190 & 3333 & 21011 & 26545 \\\hline
Amazon &  24  & 19  & 75  & 6.7  & 21.7\\\hline
 DBLP  & 23  & 20  & 57  & 14.2 &  35.2 \\\hline
 Youtube  & 52  & 59  & 105  & 75.6 &  113.6\\\hline
  LiveJournal  & 568  & 666  & 1091  & 2142  & 3556 \\\hline
  Orkut  & 1710 &  2190 &  3451  & 21011 &  26545\\\hline
\end{tabular}\vspace*{-0.6cm}
\end{center}
\end{table}

Table \ref{tab:index} reports the size (MB) and construction time (seconds) of the structural $k$-truss index (K-Truss) and \ati, along with the size of each network. The 10 Facebook ego-networks have similar results, omitted from Table \ref{tab:index} for brevity.  \xin{The size of \ati is comparable to the original graph size and structural $k$-truss index.}
%\\The size of \ati is less than 10 times of the original graph size, which is very compcat. \\ 
% ****CAN WE CLAIM THAT 10x THE ORIGINAL GRAPH IS COMPACT? WHAT OTHER SUCH INDICES HAVE BEEN PROPOSED BY OTHERS? CAN WE SAY OURS IS MORE COMPACT OR COMPARABLE TO OTHERS?**** \\ 
It confirms that the \ati scheme has $O(m+\sum_{w\in \mathcal{A}}|E(G_w)|)$ space complexity. 
\cut{Given $|A|>1000$ on all these networks, it shows the projected attribute graphs are very sparse.} 
%\\ 
% ****A LINE ABOUT HOW WE CAN INFER THIS WOULD BE GREAT.**** \\ 
% \xin{***** See above ****}\\
The \ati construction time is comparable to $k$-truss index construction and is nearly as efficient. 
\eat{ Index construction for the largest network experimented completes in 8 hours. 
} 
\LL{
%Furthermore, thanks to the useful \ati, the efficiency of our query processing algorithms greatly improves. 
It can be seen that query processing efficiency is greatly aided by the \ati. 
For instance, consider the index construction and query processing times on  DBLP network. In Figure \ref{fig.wsize_time}(b) and Figure \ref{fig.qsize_time}(b), the query time of \bulk and \basic without \ati scheme take nearly 20 seconds, while the construction time of \ati is only 35.2 seconds (Table \ref{tab:index}). That is, merely processing two queries more than pays off for the index construction effort. }

\subsection{Parameter Sensitivity Evaluation}

\begin{figure}[t]
\vskip -0.4cm
\centering 
\includegraphics[width=0.5\linewidth,height=2.5cm]{\PP 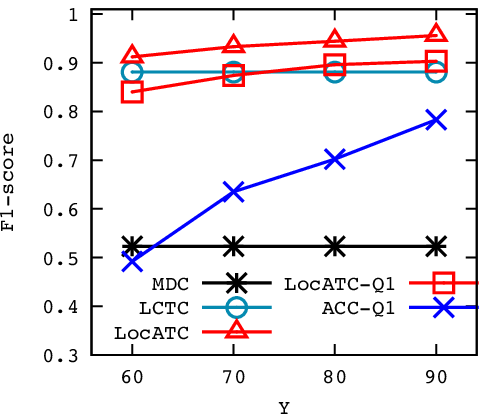}
\vskip -0.2in
\caption{Varying homogeneous attributes on Amazon: F1-score} \label{fig.Amazon_para}
\vspace*{-0.1cm}
\end{figure}

\begin{figure}[t]
\vskip -0.4cm
\centering 
\includegraphics[width=0.5\linewidth,height=2.5cm]{\PP 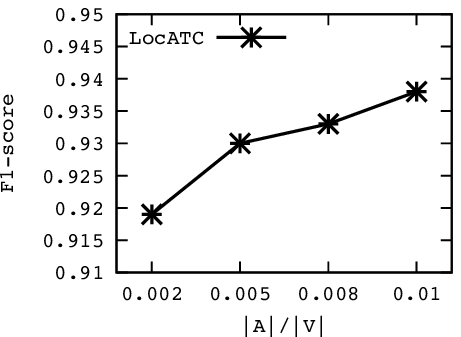}
%\vskip -0.2in
\caption{Varying parameter $\frac{|\mathcal{A}|}{|V|}$ on Amazon: F1-score} \label{fig.Amazon_AV}
\vspace*{-0.1cm}
\end{figure}

\new{In this experiment, we vary various parameters in used in the synthetic data generation, query generation, and in algorithm definitions, and evaluate the quality and efficiency performance of \LATC. }

\stitle{Varying homogeneous attributes in synthetic datasets. } \new{For each ground-truth community in Amazon, we randomly select 3 attributes, and assign each of these attributes to each of $Z$\% vertices in the community, where $Z$ is a random number in [50, $Y$]. Note that different attributes may have different value of $Z$. The parameter $Y$ is varied from 60 to 90. %The results of F1-socre are respectively shown in Figure \ref{fig.qsize_time} (a) and (b). 
As $Y$ is increased, intuitively the level of homogeneity in the network and in its communities increases. 
The results of F1-socre are shown in Figure \ref{fig.Amazon_para}. As homogeneous attributes in communities increase, \MDC and \LCTC maintain the same F1-socre, while the F1-score of all methods of attributed community search -- \LATC, \LATC-Q1, and \ACC-Q1 -- increases as homogeneity increases. Once again, \LATC is the best method even when the proportion  of homogeneous attributes falls in [50, 60].  \LATC-Q1 beats \ACC-Q1 for all settings of homogeneity.  Similar results can be also observed on other synthetic datasets. }

\stitle{Varying the average number of attribute/vertex $\frac{|\mathcal{A}|}{|V|}$ in synthetic datasets. } \new{In this experiment, we vary the average number of attribute/vertex $\frac{|\mathcal{A}|}{|V|}$ to generate different attribute sets in Amazon. The results are shown in Figure \ref{fig.Amazon_AV}. With the increased $\frac{|\mathcal{A}|}{|V|}$, \LATC performs better. This is because the size of attribute set $\mathcal{A}$ becomes larger, which makes homogeneity of synthetic attributes in different communities more likely. Finally, it bring more challenges to detected accurate communities for a smaller $\frac{|\mathcal{A}|}{|V|}$.    We can obtain similar results on other synthetic datasets. }

\stitle{Varying query vertex size $|V_q|$ and query attribute size $|W_q|$.}  \new{We test the quality performance of \LATC using different queries by varying $|V_q|$ and $|W_q|$. The results are shown in Figure \ref{fig.qw_score} (a) and (b). As we can see, given more information of query vertices and query attributes within communities, our algorithm accordingly performs better.}

\stitle{Varying parameters $\epsilon$, $\gamma$, and $\eta$.} \new{We test the performance of \LATC by varing  $\epsilon$, $\gamma$, and $\eta$.  We used the same query nodes that are selected in Sec. \ref{sec.exp-quality} on f414 network. Similar results can be also observed on other networks with real attributes. The results of F1-socre and query time by varying  $\epsilon$ are respectively reported in Figure \ref{fig.epsilon} (a) and (b). As we can see, \LATC removes a smaller portion of nodes, which achieves a higher F1-score using more query time. In addition, we test different values of  $\gamma$ and report the results in Figure \ref{fig.gamma} (a) and (b). The F1-score remains stable  as $\gamma$ increases from 0.1 to 0.5, and then decreases a little bit for a larger value of $\gamma=1.0$. Thus, the default choice of $\gamma =0.2$ is good at balancing the cohesive structure and homogeneous attributes in an  efficiency way.  Furthermore, we also report the results by varying the parameter $\eta$ in Figure \ref{fig.eta} (a) and (b). As can be seen, the F1-score remains stable with increasing $\eta$, while the running time increases a little with larger $\eta$. The results show that the default setting $\eta =1000$ is large enough for achieving a good balance of efficiency and quality. }

\begin{figure}[t]
\vskip -0.4cm
\centering \mbox{
\subfigure[Varying  $|V_q|$.]{\includegraphics[width=0.5\linewidth,height=2.5cm]{\PP 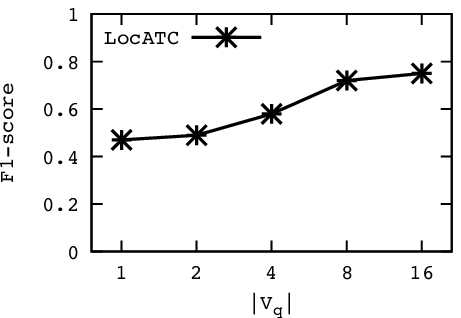}}%\hskip -0.15in
\subfigure[Varying $|W_q|$.]{\includegraphics[width=0.5\linewidth,height=2.5cm]{\PP 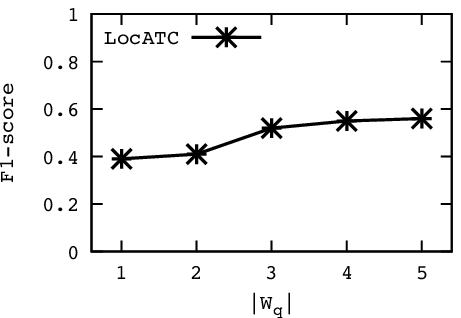}}
%\hskip -0.15in
} \vskip -0.2in
\caption{Varying queries on f414: F1-score} \label{fig.qw_score}
\vspace*{-0.1cm}
\end{figure}

\begin{figure}[t]
\vskip -0.4cm
\centering \mbox{
\subfigure[F1-score]{\includegraphics[width=0.5\linewidth,height=2.5cm]{\PP 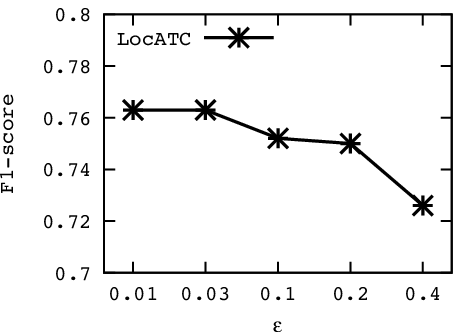}}%\hskip -0.15in
\subfigure[Query Time]{\includegraphics[width=0.5\linewidth,height=2.5cm]{\PP 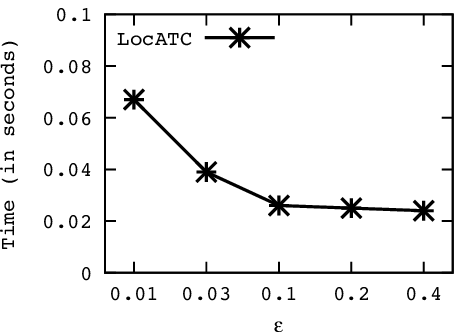}}
%\hskip -0.15in
} \vskip -0.2in
\caption{Varying $\epsilon$ on f414} \label{fig.epsilon}
\vspace*{-0.1cm}
\end{figure}

\begin{figure}[t]
\vskip -0.4cm
\centering \mbox{
\subfigure[F1-score]{\includegraphics[width=0.5\linewidth,height=2.5cm]{\PP 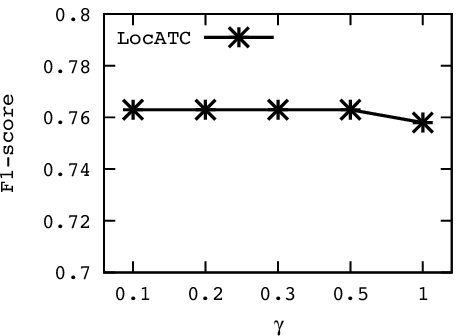}}%\hskip -0.15in
\subfigure[Query Time]{\includegraphics[width=0.5\linewidth,height=2.5cm]{\PP 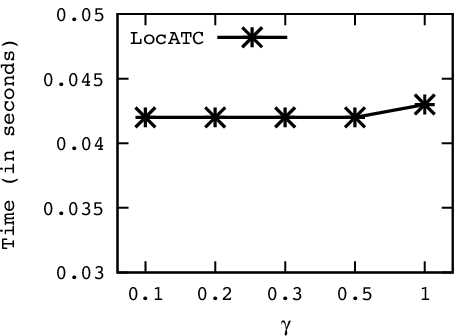}}
%\hskip -0.15in
} \vskip -0.2in
\caption{Varying $\gamma$ on f414} \label{fig.gamma}
\vspace*{-0.1cm}
\end{figure}

\begin{figure}[t]
\vskip -0.4cm
\centering \mbox{
\subfigure[F1-score]{\includegraphics[width=0.5\linewidth,height=2.5cm]{\PP 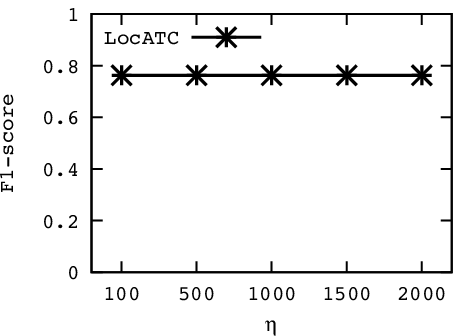}}%\hskip -0.15in
\subfigure[Query Time]{\includegraphics[width=0.5\linewidth,height=2.5cm]{\PP 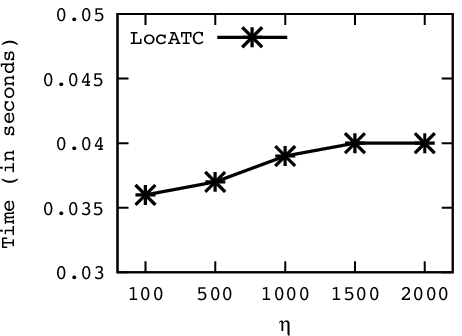}}
%\hskip -0.15in
} \vskip -0.2in
\caption{Varying $\eta$ on f414} \label{fig.eta}
\vspace*{-0.1cm}
\end{figure}

\subsection{Bad Query Evaluation}\label{sec.badexp}

\begin{figure}[h]
\vskip -0.4cm
\centering \mbox{
\subfigure[Density]{\includegraphics[width=0.5\linewidth,height=2.5cm]{\PP 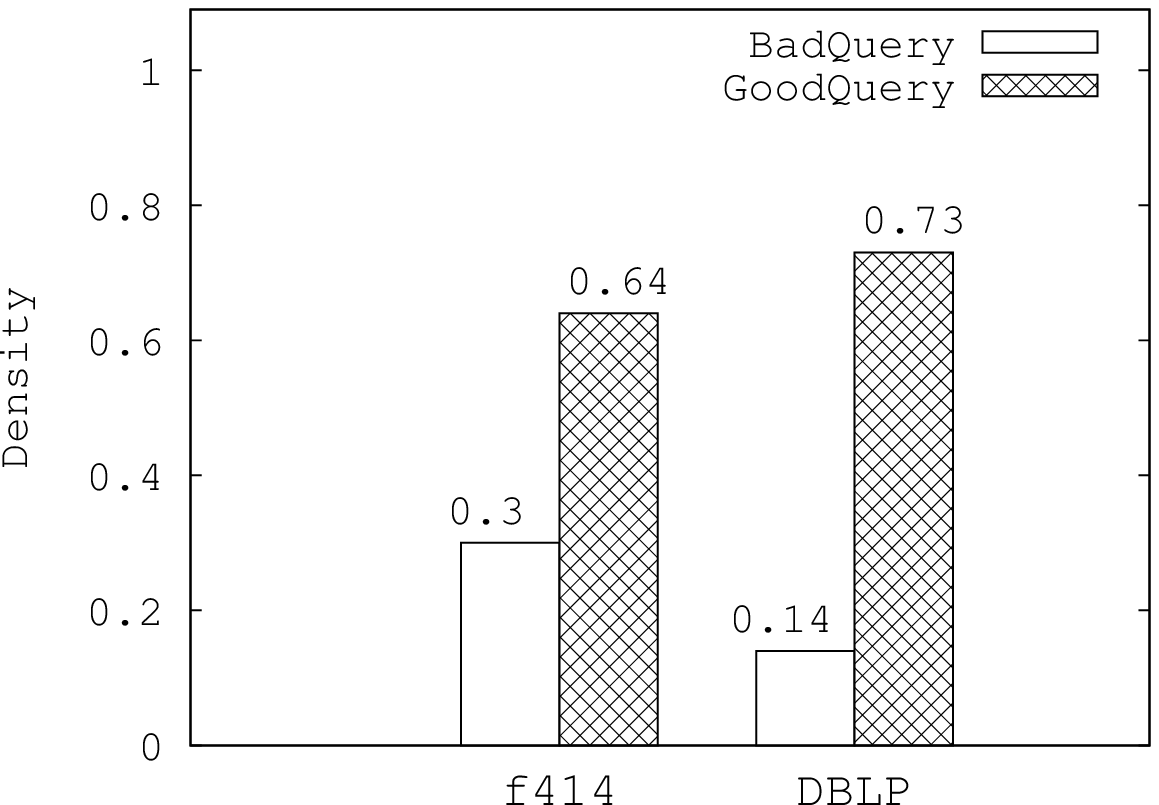}}
%\subfigure[Density]{\includegraphics[width=0.5\linewidth,height=2.5cm]{\PP badquery_density-eps-converted-to.pdf}}
%\hskip -0.15in
\subfigure[Query Time]{\includegraphics[width=0.5\linewidth,height=2.5cm]{\PP 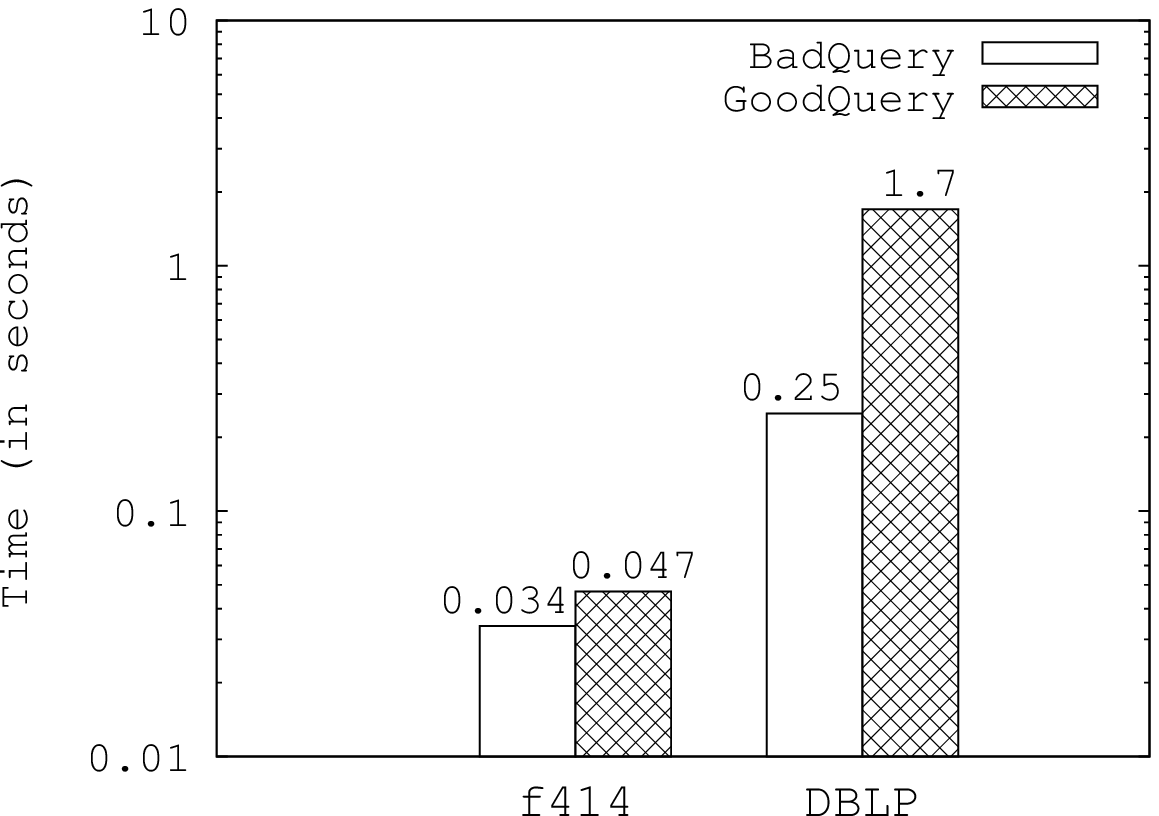}}
%\hskip -0.15in
} \vskip -0.2in
\caption{Testing Bad Queries on f414 and DBLP} \label{fig.bad}
\vspace*{-0.1cm}
\end{figure}

\new{In this experiment, we use bad queries to test the performance of \LATC on f414 and DBLP networks. 
We generate bad queries by randomly choosing query nodes and query attributes from \emph{different} ground truth communities. We test a set of 100 bad queries generated in this manner.  We also test 100 queries that are selected in Sec.~\ref{sec.exp-quality} as good queries. We intuitively expect bad queries to result in discovered communities with poor density, compared to good queries. We compare the quality of discovered communities $C$ in terms of edge density, i.e.,  $\frac{|E(C)|}{|V(C)|}$, \LL{averaged across the 100 queries.} Figure \ref{fig.bad}(a) shows the results of edge density. Compared with bad queries, \LATC can find communities with larger densities for good queries. Figure \ref{fig.bad}(b) shows the average running times on good and bad queries. \LATC processes bad queries much faster than good queries, which achieves 6.8 times of efficiency improvement on DBLP network. Intuitively, \LATC can quickly return empty answers for bad queries, if the algorithm can determine the weak structure of query nodes and heterogeneous attributes of neighbors in the proximity of query nodes.}

%\new{We randomly test 100 sets of bad queries $Q=(V_q, W_q)$, where query vertices $V_q$ and query attributes $W_q$ are randomly selected from \emph{different}  communities. We also test 100 queries that are selected in Sec. \ref{sec.exp-quality} as good queries.  Since there are no ground-truth communities for bad queries, we compare the quality of discovered communities $C$ in terms of edge density, i.e.,  $\frac{|E(C)|}{|V(C)|}$, \LL{averaged across the 100 queries.} Figure \ref{fig.bad}(a) shows the results of edge density. Compared with bad queries, \LATC can find communities with larger densities for good queries. Figure \ref{fig.bad}(b) shows the average running times on good and bad queries, where it can be seen that \LATC processes bad queries much faster than good queries. Intuitively, \LATC can quickly return empty answers for bad queries, if the algorithm can determine the weak structure of query nodes and heterogeneous attributes of neighbors in the proximity of query nodes.} 

% {\bf From Laks: In the running time plot in Fig. 15b, please add the values of running times on top of those bars.} 
% \\
% \purple{From Xin: Done.}

\begin{figure}[h]
\vskip -0.3cm
\centering \mbox{
\subfigure[\atc \& a protein complex]{\includegraphics[width=0.45\linewidth,height=2.5cm]{\PP 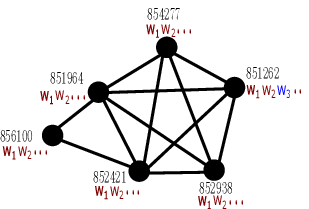}}
\subfigure[(3,3)-truss]{\includegraphics[width=0.55\linewidth,height=2.5cm]{\PP 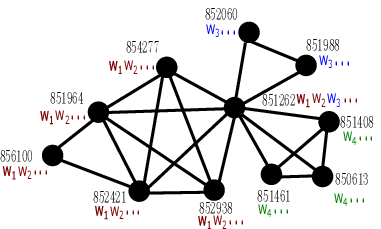}}%\hskip -0.15in
%\hskip -0.15in
} \vskip -0.2in
\caption{$Q=(\{q_1, q_2\}, \{w_1, w_2\})$ where $q_1=854277$, $q_2=856100$ and $w_1=$ ``GO:0001009", $w_2=$``GO:0001041"} \label{fig.case}
\vspace*{-0.5cm}
\end{figure}

\subsection{Case Study on PPI network}
%Here, we apply our \LATC algorithm to solve the attribute-driven community search problem. 
\new{Besides the quality evaluation measured by F1-score, we also apply the \LATC algorithm on the protein-protein interaction (PPI) network Krogan. Recall that the performance of  \LATC handling bad queries has been tested in Section \ref{sec.badexp}, and we test good queries here.
%for community search. 
We examine the details of the discovered protein complexes to investigate biologically significant clues, which help us to better understand the protein complexes.  
Figure \ref{fig.case}(a) shows one complex ``transcription factor TFIIIC complex'' in sccharomyces cerevisiae, which is identified by biologists previously. %, where each complex consists of proteins represented by gene symbols. 
% \\ 
% {\bf From Laks: What does ``identified'' mean? Identified by \LATC or identified by biologists previously?}
% \\
% \purple{From Xin: biologists.}
% \\ 
The graph contains 6 nodes and 12 edges, with density 0.8 and diameter 2.
We adopt the following procedure for checking whether a protein is present in a complex. Taking gene id  ``854277'' as an example, we can go to the NCBI\footnote{\scriptsize{\url{https://www.ncbi.nlm.nih.gov/}}}, input ``854277'' in the search box, and select the category of ``Gene'', then we  will obtain information related to this gene, from which we can check  whether this gene is one of the proteins in the protein complex\footnote{\scriptsize{\url{http://wodaklab.org/cyc2008/resources/CYC2008_complex.tab}}}. 
 %The information of gene ids (e.g. 854277) contains gene symbols (e.g. TFC7) that is corresponding to particular proteins. It can be known by searching ``854227'' and selecting the category of ``Gene'' in this website\footnote{\scriptsize{\url{https://www.ncbi.nlm.nih.gov/}}}.    
Similar with the procedure of good query generation in Sec.~\ref{sec.exp-quality}, we randomly sample a  query as $Q=(V_q, W_q)$ where $V_q=$\{854277, 856100\} and $W_q=$\{``GO:0001009", ``GO:0001041"\},
%$W_q=$\{``GO:0001009", ``GO:0001041", ``GO:0001003", ``GO:0001005"\},
 and set the parameters $k=3 $ and $d=3$.  
% \\ 
% {\bf From Laks: How did we randomly pick these nodes and attributes? Are they present in Fig. 16a?} 
% \\ 
% \purple{From Xin: we generate the query totally the same as procedure for datasets with real attributes in Section 8.2.  we randomly select two nodes from graph in Fig. 16a, and choose attributes occurring most frequently in this community and rarely occurring in others. All are present in Fig. 16a.}
% \\ 
 To illustrate the importance of the consideration of protein attributes in detecting protein complexes, we simply use the structure and find the $(3, 3)$-truss shown in Figure \ref{fig.case}(b). This community contains 11 proteins including 6 proteins of the ground-truth complex of Figure \ref{fig.case}(a). The other 5 proteins not present in the ground-truth complex are associated with no query attributes, but have other attributes $w_3$ and $w_4$, as shown in Figure \ref{fig.case}(b). When we look up the database of Gene Ontology\footnote{\scriptsize{\url{http://geneontology.org/ontology/go-basic.obo}}}, we know that  the attributes of ``biological processes"  as ``GO:0001009" and ``GO:0001041" respectively represent  ``transcription from RNA polymerase III hybrid type promoter" and ``transcription from RNA polymerase III type 2 promoter". 
 Except query attributes, we omitted details of other attributes from Figure \ref{fig.case} for simplicity.
% \\
% \purple{Except query attributes, we omitted details of other attributes from Figure \ref{fig.case} for simplicity.}
% \\
 {\sl \LATC is able to identify all proteins that preform the same biological process of transcription from RNA polymerase.} Overall, \LATC successfully identifies all proteins that constitute the ground-truth complex in Figure \ref{fig.case}(a). Other than these two homogeneous attributes, interestingly, we also  discover another two attributes shared by all proteins in terms of ``molecular functions". 
%  \\ 
% {\bf From Laks: Are these additional attributes shown in the figure? If not, let's say that they are omitted for simplicity.} 
% \\
% \purple{From Xin: see above purple sentences.}
%  \\ 
% \\ 
Specifically, the attributes ``GO:0001003" and ``GO:0001005" respectively perform DNA binding activity  as``RNA polymerase III type 2 promoter sequence-specific DNA binding" and `` RNA polymerase III type 1 promoter sequence-specific DNA binding". Overall, this complex exists in the cell nucleus, according to the same attribute ``cellular components" of ``GO:0005634" in all proteins. } 

%% file: KW_Community.bbl
\begin{thebibliography}{10}

\bibitem{dbxplorer02}
S.~Agrawal, S.~Chaudhuri, and G.~Das.
\newblock Dbxplorer: A system for keyword-based search over relational
  databases.
\newblock In {\em ICDE}, pages 5--16, 2002.

\bibitem{bahmani2012densest}
B.~Bahmani, R.~Kumar, and S.~Vassilvitskii.
\newblock Densest subgraph in streaming and mapreduce.
\newblock {\em PVLDB}, 5(5):454--465, 2012.

\bibitem{barbieri2015efficient}
N.~Barbieri, F.~Bonchi, E.~Galimberti, and F.~Gullo.
\newblock Efficient and effective community search.
\newblock {\em DMKD}, 29(5):1406--1433, 2015.

\bibitem{BatageljZ03}
V.~Batagelj and M.~Zaversnik.
\newblock An o (m) algorithm for cores decomposition of networks.
\newblock {\em arXiv preprint cs/0310049}, 2003.

\bibitem{bhalotia2002keyword}
G.~Bhalotia, A.~Hulgeri, C.~Nakhe, S.~Chakrabarti, and S.~Sudarshan.
\newblock Keyword searching and browsing in databases using banks.
\newblock In {\em ICDE}, pages 431--440, 2002.

\bibitem{bothorel2015clustering}
C.~Bothorel, J.~D. Cruz, M.~Magnani, and B.~Micenkova.
\newblock Clustering attributed graphs: models, measures and methods.
\newblock {\em Network Science}, 3(03):408--444, 2015.

\bibitem{cheng2012clustering}
H.~Cheng, Y.~Zhou, X.~Huang, and J.~X. Yu.
\newblock Clustering large attributed information networks: an efficient
  incremental computing approach.
\newblock {\em DMKD}, 25(3):450--477, 2012.

\bibitem{ChibaN85}
N.~Chiba and T.~Nishizeki.
\newblock Arboricity and subgraph listing algorithms.
\newblock {\em SIAM J. Comput.}, 14(1):210--223, 1985.

\bibitem{cohen2008}
J.~Cohen.
\newblock Trusses: Cohesive subgraphs for social network analysis.
\newblock Technical report, National Security Agency, 2008.

\bibitem{CuiXWLW13}
W.~Cui, Y.~Xiao, H.~Wang, Y.~Lu, and W.~Wang.
\newblock Online search of overlapping communities.
\newblock In {\em SIGMOD}, pages 277--288, 2013.

\bibitem{cui2014local}
W.~Cui, Y.~Xiao, H.~Wang, and W.~Wang.
\newblock Local search of communities in large graphs.
\newblock In {\em SIGMOD}, pages 991--1002, 2014.

\bibitem{ding2007finding}
B.~Ding, J.~X. Yu, S.~Wang, L.~Qin, X.~Zhang, and X.~Lin.
\newblock Finding top-k min-cost connected trees in databases.
\newblock In {\em ICDE}, pages 836--845, 2007.

\bibitem{Edachery99graphclustering}
J.~Edachery, A.~Sen, and F.~J. Brandenburg.
\newblock Graph clustering using distance-k cliques.
\newblock In {\em Proceedings of the 7th International Symposium on Graph
  Drawing}, pages 98--106, 1999.

\bibitem{FangCLH16}
Y.~Fang, R.~Cheng, S.~Luo, and J.~Hu.
\newblock Effective community search for large attributed graphs.
\newblock {\em {PVLDB}}, 9(12):1233--1244, 2016.

\bibitem{gajewar2012multi}
A.~Gajewar and A.~D. Sarma.
\newblock Multi-skill collaborative teams based on densest subgraphs.
\newblock In {\em SDM}, pages 165--176, 2012.

\bibitem{gunnemann2011db}
S.~G{\"u}nnemann, B.~Boden, and T.~Seidl.
\newblock Db-csc: a density-based approach for subspace clustering in graphs
  with feature vectors.
\newblock In {\em ECML/PKDD}, pages 565--580, 2011.

\bibitem{hristidis2003efficient}
V.~Hristidis, L.~Gravano, and Y.~Papakonstantinou.
\newblock Efficient ir-style keyword search over relational databases.
\newblock In {\em PVLDB}, pages 850--861, 2003.

\bibitem{hristidis2002discover}
V.~Hristidis and Y.~Papakonstantinou.
\newblock Discover: Keyword search in relational databases.
\newblock In {\em PVLDB}, pages 670--681, 2002.

\bibitem{hu2013utilizing}
A.~L. Hu and K.~C. Chan.
\newblock Utilizing both topological and attribute information for protein
  complex identification in ppi networks.
\newblock {\em TCBB}, 10(3):780--792, 2013.

\bibitem{huang2014}
X.~Huang, H.~Cheng, L.~Qin, W.~Tian, and J.~X. Yu.
\newblock Querying k-truss community in large and dynamic graphs.
\newblock In {\em SIGMOD}, pages 1311--1322, 2014.

\bibitem{huang2015dense}
X.~Huang, H.~Cheng, and J.~X. Yu.
\newblock Dense community detection in multi-valued attributed networks.
\newblock {\em Information Sciences}, 314:77--99, 2015.

\bibitem{huang2015approximate}
X.~Huang, L.~V. Lakshmanan, J.~X. Yu, and H.~Cheng.
\newblock Approximate closest community search in networks.
\newblock {\em PVLDB}, 9(4):276--287, 2015.

\bibitem{kacholia2005bidirectional}
V.~Kacholia, S.~Pandit, S.~Chakrabarti, S.~Sudarshan, R.~Desai, and
  H.~Karambelkar.
\newblock Bidirectional expansion for keyword search on graph databases.
\newblock In {\em VLDB}, pages 505--516, 2005.

\bibitem{kargar2011discovering}
M.~Kargar and A.~An.
\newblock Discovering top-k teams of experts with/without a leader in social
  networks.
\newblock In {\em CIKM}, pages 985--994, 2011.

\bibitem{khuller2009finding}
S.~Khuller and B.~Saha.
\newblock On finding dense subgraphs.
\newblock In {\em ICALP}, pages 597--608, 2009.

\bibitem{kou1981fast}
L.~Kou, G.~Markowsky, and L.~Berman.
\newblock A fast algorithm for steiner trees.
\newblock {\em Acta informatica}, 15(2):141--145, 1981.

\bibitem{lappas2009finding}
T.~Lappas, K.~Liu, and E.~Terzi.
\newblock Finding a team of experts in social networks.
\newblock In {\em KDD}, pages 467--476, 2009.

\bibitem{li2008ease}
G.~Li, B.~C. Ooi, J.~Feng, J.~Wang, and L.~Zhou.
\newblock Ease: an effective 3-in-1 keyword search method for unstructured,
  semi-structured and structured data.
\newblock In {\em SIGMOD}, pages 903--914, 2008.

\bibitem{li2015influential}
R.-H. Li, L.~Qin, J.~X. Yu, and R.~Mao.
\newblock Influential community search in large networks.
\newblock {\em PVLDB}, 8(5), 2015.

\bibitem{mcauley2012learning}
J.~J. McAuley and J.~Leskovec.
\newblock Learning to discover social circles in ego networks.
\newblock In {\em NIPS}, volume 272, pages 548--556, 2012.

\bibitem{mehlhorn1988faster}
K.~Mehlhorn.
\newblock A faster approximation algorithm for the steiner problem in graphs.
\newblock {\em Information Processing Letters}, 27(3):125--128, 1988.

\bibitem{qin2009querying}
L.~Qin, J.~X. Yu, L.~Chang, and Y.~Tao.
\newblock Querying communities in relational databases.
\newblock In {\em ICDE}, pages 724--735, 2009.

\bibitem{ruan2013efficient}
Y.~Ruan, D.~Fuhry, and S.~Parthasarathy.
\newblock Efficient community detection in large networks using content and
  links.
\newblock In {\em WWW}, pages 1089--1098, 2013.

\bibitem{sariyuce2015finding}
A.~E. Sariyuce, C.~Seshadhri, A.~Pinar, and U.~V. Catalyurek.
\newblock Finding the hierarchy of dense subgraphs using nucleus
  decompositions.
\newblock In {\em WWW}, pages 927--937, 2015.

\bibitem{sozio2010}
M.~Sozio and A.~Gionis.
\newblock The community-search problem and how to plan a successful cocktail
  party.
\newblock In {\em KDD}, pages 939--948, 2010.

\bibitem{WangC12}
J.~Wang and J.~Cheng.
\newblock Truss decomposition in massive networks.
\newblock {\em PVLDB}, 5(9):812--823, 2012.

\bibitem{wu2015robust}
Y.~Wu, R.~Jin, J.~Li, and X.~Zhang.
\newblock Robust local community detection: On free rider effect and its
  elimination.
\newblock {\em PVLDB}, 8(7), 2015.

\bibitem{YangL12}
J.~Yang and J.~Leskovec.
\newblock Defining and evaluating network communities based on ground-truth.
\newblock In {\em ICDM}, pages 745--754, 2012.

\bibitem{zhou2009graph}
Y.~Zhou, H.~Cheng, and J.~X. Yu.
\newblock Graph clustering based on structural/attribute similarities.
\newblock {\em PVLDB}, 2(1):718--729, 2009.

\end{thebibliography}
